\documentclass[11pt]{article}
\usepackage{graphicx}
\usepackage{epsf}
\usepackage{amsmath,amsfonts}

\def\date{\today}  

\newtheorem{proposition}{proposition}[section]
\newtheorem{lemma}[proposition]{Lemma}

\newtheorem{theorem}[proposition]{Theorem}

\newtheorem{claim}[proposition]{Claim}
\newcommand{\qed}{\hfill$\square$\bigskip}
\newcommand{\DEF}[1]{{\em #1\/}}
\newcommand\myitemsep{\setlength{\itemsep}{0pt}}

\newcommand\fw{\hbox{\tt fw}}  
\newcommand{\eopf}{$\square$}
\newcommand\drop[1]{}
\def\showlabel#1{}
\def\showfiglabel#1{}
\newcommand\bsize{{\tt bsize}}
\newenvironment{proof}%
{\noindent{\bf Proof.}\ }%
{\hfill\eopf\par\bigskip}%

\begin{document}
\newif\ifproofmode
\proofmodetrue

\baselineskip=13pt \phantom{a}\vskip .25in \centerline{{\large{\bf
Graph Isomorphism for Bounded Genus Graphs In Linear Time}}}
\vskip.4in

\centerline{{\bf Ken-ichi Kawarabayashi}\footnote{Research partly
supported by Japan Society for the Promotion of Science,
Grant-in-Aid for Scientific Research\\ Email address: {\tt k\_keniti@nii.ac.jp}}} \centerline{National Institute of
Informatics and JST ERATO Kawarabayashi Large Graph Project} \centerline{2-1-2 Hitotsubashi, Chiyoda-ku, Tokyo
101-8430, Japan}


\vskip 0.8in
\centerline{\bf Abstract}
\bigskip
\parshape=1.5truein 5.5truein

For every integer $g$, isomorphism of graphs of Euler genus at most $g$ can be decided in linear time.

This improves previously known algorithms whose time complexity is
$n^{O(g)}$ (shown in early 1980's), and in fact, this is the first
fixed-parameter tractable algorithm for the graph isomorphism problem for
bounded genus graphs in terms of the Euler genus $g$.
Our result also generalizes the seminal result of Hopcroft and Wong in 1974, which says that
the graph isomorphism problem can be decided in linear time for planar graphs.

Our proof is quite lengthly and complicated, but if we are satisfied with an $O(n^3)$ time algorithm for the same problem, the proof is shorter and easier.

\bigskip
\bigskip
\bigskip

\noindent March 17, 2012, revised \date.


\vfil\eject

\ifproofmode
  \def\sectionbreak{\vfil\eject}
   \newcount\remarkno
   \def\REMARK#1{{%
      \footnote{\baselineskip=11pt #1
      \vskip-\baselineskip}\global\advance\remarkno by1}}
\else
  \def\sectionbreak{}
  \def\REMARK#1{}
\fi

\baselineskip 13.3pt

\section{Introduction}

\subsection{The Graph Isomorphism Problem}

The graph isomorphism problem asks whether or not two given graphs
are isomorphic. It is considered by many as one of the most
challenging problems today in theoretical computer science. While
some complexity theoretic results indicate that this problem might not
be NP-complete (if it were, the polynomial hierarchy would collapse to
its second level, see \cite{BM,BHZ,OMW,GS,sch}), no polynomial time algorithm
is known for it, even with extended resources like randomization or
quantum computing.

 On the other hand, there is a number of important classes of
graphs on which the graph isomorphism problem is known to be
solvable in polynomial time. For example, in 1990, Bodlaender
\cite{bdtr} gave a polynomial time algorithm for the graph
isomorphism problem for graphs of bounded tree-width.
Many NP-hard problems can be solved in polynomial time, even
in linear time, when input is restricted to graphs of tree-width
at most $k$ \cite{Arn,bod}. So, Bodlaender's result may not be
surprising, but the time complexity in \cite{bdtr} is $O(n^{k+2})$, and
no one could improve the time complexity to $O(n^{O(1)})$
until quite recently \cite{focs14}.
This indicates that even for graphs of bounded tree-width,
the graph isomorphism problem is not trivial at all.

Another important family of graphs is the planar graphs.
In 1966, Weinberg \cite{wein} gave a very simple $O(n^2)$
algorithm for the graph isomorphism problem for planar graphs.
This was improved by Hopcroft and Tarjan \cite{HT1,HT2}
to $O(n\log n)$.
Building on this earlier work, Hopcroft and Wong \cite{HW}
published in 1974 a seminal paper, where they presented a
linear time algorithm for the graph isomorphism problem for planar graphs.

There are some other classes of graphs on which the graph isomorphism
problem is solvable in polynomial time. This includes
minor-closed families of graphs \cite{miller2,minorclosed,Po89}, and graphs without a fixed graph as a topological minor \cite{gromar12}.
A powerful approach based on group theory was introduced by Babai \cite{Babai79}.
Based on this approach, Babai et al.\ \cite{babai1} proved that the graph isomorphism
problem is polynomially solvable for graphs of bounded eigenvalue multiplicity,
and Luks \cite{luks} described his well-known group theoretic
algorithm for the graph isomorphism problem for graphs of bounded degree.
Babai and others \cite{babai2,babai3} investigated the graph isomorphism
problem for random graphs.


\subsection{Bounded Genus Graphs}

Leaving the plane to consider graphs on surfaces of higher genus,
the graph isomorphism problem seems much harder. In 1980, Filotti,
Mayer \cite{FM} and Miller \cite{miller1} showed that for every
orientable surface $S$, there is a polynomial time algorithm for the
graph isomorphism problem for graphs that can be embedded in $S$,
but the time complexity is $n^{O(g)}$,
where $g$ is the Euler genus of $S$.
Lichtenstein \cite{li} gives an $O(n^3)$ algorithm for the graph
isomorphism problem for projective planar graphs. These works came out
in the early 1980's. These classes of graphs were extensively
studied from other perspectives. For example, Grohe and Verbitsky
\cite{grohe1,grohe2}, who studied this problem from a logic point of view,
made some interesting progress. However, no one could improve the time
complexity in the last 30 years. This can be perhaps explained in
the following way. We can rather easily reduce the problem to
3-connected graphs.
For planar graphs, the famous result of Whitney tells us that
embeddings of 3-connected graphs in the plane are (combinatorially)
unique. This allows us to reduce the graph isomorphism problem to the map isomorphism problem, which is
easier (see Hopcroft and Wong \cite{HW} and Theorem \ref{thm:main3}).
But for every nonsimply connected surface $S$, there exist
3-connected graphs with exponentially many embeddings.
This makes an essential difference between planar graphs and graphs in surfaces of
higher genus. In addition, Thomassen \cite{thomassen} proved that it is NP-complete to determine Euler genus of
a given graph.


A graph $G$ embedded in a surface $S$ has \DEF{face-width} or
\DEF{representativity} at least $k$, $\fw(G)\ge k$, if every
non-contractible closed curve in the surface intersects the graph in
at least $k$ points. This notion turns out to be of great importance in
the graph minor theory of Robertson and Seymour, cf.~\cite{MKsurvey},
and in topological graph theory, cf.~\cite{MT}. If an embedding of $G$ in $S$ is of face-width $k$, then
we sometimes call this embedding \DEF{face-width $k$ embedding}.

If $G$ is 3-connected and $\fw(G)\ge3$, then the embedding has
properties that are characteristic for 3-connected planar graphs.
The main property is that the faces are all simple polygons and
that they intersect nicely -- if two distinct faces are not disjoint,
their intersection is either a single vertex or a single edge.
Therefore such embeddings are sometimes called \DEF{polyhedral embeddings}.

%
%
The important property about 3-connected graphs that have a polyhedral embedding in a surface is the following
in \cite{kmstoc08,MR}.
\begin{lemma}
\label{finitely}
\showlabel{finitely}
Let $G$ be a 3-connected polyhedrally embeddable graph in a surface $S$ of Euler genus $g$.
There is a function $f(g)$ such that $G$ has at most $f(g)$ different polyhedral embeddings in $S$.
\end{lemma}
In fact, in \cite{kmstoc08}, the following was shown.
\begin{theorem}
\label{thm:main2}
\showlabel{thm:main2}
For each surface $S$, there
is a linear time algorithm for the following problem: Given an
integer $k \geq 3$ and a graph $G$, either find an embedding of $G$
in $S$ with face-width at least $k$, or conclude that $G$ does not
have such an embedding. Moreover, if there is an embedding in $S$ of
face-width at least $k$ and $G$ is 3-connected, the algorithm gives rise to all embeddings with
this property.
Furthermore, the number of such embeddings is
at most $f(g)$, where $f(g)$ comes from Lemma \ref{finitely}.
\end{theorem}

We have to require the face-width of the embedding to be at least 3
in Theorem \ref{thm:main2}, since there are 3-connected graphs with
exponentially many embeddings in any surface (other than the
sphere). If we want to have a unique embedding in the surface of
Euler genus $g$ (which is an analogue of Whitney's theorem on
the uniqueness of an embedding in the plane), then the face-width must
be $\Theta(\log g/\log\log g)$. Sufficiency of this was proved
in \cite{uni1,uni2}, necessity in \cite{Archd}.

\subsection{Our Main Result}

Our main result of this paper is the following.

\begin{theorem}
\label{thm:main1} \showlabel{thm:main1}
For every integer $g$, isomorphism of graphs of Euler genus at most $g$ can be decided in linear time.
\end{theorem}

Let us point out that the proof is quite lengthly and complicated,
but if we are satisfied with an $O(n^3)$ time algorithm for the
same problem, the proof becomes shorter and easier.
In particular, the proof of Theorem \ref{algotwo}, which is
the most technical in our proof, becomes much simpler (we will mention this point in the proof of Theorem \ref{algotwo}).

Theorem \ref{thm:main1} is a generalization of the seminal
result of Hopcroft and Wong \cite{HW} that says that there is a
linear time algorithm for the graph isomorphism problem for planar
graphs. As remarked above, the time complexity of previously known
results for the graph isomorphism problem for graphs embeddable in a
surface of the Euler genus $g$ is $n^{O(g)}$, and this was proved in the
early 1980's. Theorem \ref{thm:main1} is the first improvement in
these 30 years, and the first fixed-parameter tractable result in
terms of the Euler genus $g$ for the graph isomorphism problem of this
class of graphs.

Let us point out that if we are satisfied with an $O(n^3)$ time algorithm
for  Theorem \ref{thm:main1}, the proof will be much easier and simpler.
Indeed, it seems to us that the hard part of our proof will be significantly simplified (cf., proofs of Theorem \ref{algotwo}
and Lemma \ref{faceone}).

In Section \ref{overview}, we shall give overview of
our algorithm.  Before that, we give several basic definitions.

\subsection{Basic Definitions}

Before proceeding, we review basic definitions concerning our work.

For basic graph theoretic definitions, we refer the reader to the book
by Diestel \cite{diestel}. For the notions of topological graph theory
we refer to the monograph by Mohar and Thomassen \cite{MT}.
A separation $(A,B)$ is a pair of sets $G=A \cup B$ such that there are no edges
between $A-B$ and $B-A$. The order of the separation $(A,B)$ is $|A \cap B|$.
By an \DEF{embedding} of a graph in a surface $S$ we mean
a \DEF{$2$-cell embedding} in $S$, i.e., we always assume that every face
is homeomorphic to an open disk in the plane.
Such embeddings can be
represented combinatorially by means of \DEF{local rotation} and
\DEF{signature}. See \cite{MT} for details. The local rotation and signature
define \DEF{rotation system}.
We define the \DEF{Euler genus} of a surface $S$ as $2-\chi(S)$,
where $\chi(S)$ is the Euler characteristic of $S$. This parameter coincides
with the usual notion of the genus, except that it is twice as large if
the surface is orientable.

A graph $G$ embedded in a surface $S$ has \DEF{face-width}
(or \DEF{representativity}) at least $\theta$ if every closed curve
in $S$, which
intersects $G$ in fewer than $\theta$ vertices and does not cross edges is contractible
(null-homotopic) in $S$. Alternatively, the \DEF{face-width}
of $G$ is equal to the minimum number of
facial walks whose union contains a cycle which is non-contractible
in $S$. It is known that if face-width of $G$ is at least two, then every face bounds a disk.
See \cite{MT} for further details.
Given a non-contractible curve in a
non-orientable surface, there are two kind of non-contractible curves; either
orientation-preserving or not orientation-preserving.

Let $W$ be an embedding of $G$ in a surface $S$ (given by means of a
rotation system and a signature). A \DEF{surface minor} is defined as follows. For each edge $e$ of $G$,
$W$ induces an embedding of both $G-e$ and $G/e$ ($/$ means contraction). The induced embedding
of $G/e$ is always in the same surface (unless $e$ is a loop), but the removal of $e$ may give
rise to a face which is not homeomorphic to a disk, in which case
the induced embedding of $G-e$ may be in another surface (of smaller genus).
A sequence of contractions and deletions of edges results in
a $W'$-embedded minor $G'$ of $G$, and we say that the $W'$-embedded minor
$G'$ is a \DEF{surface minor} of the $W$-embedded graph $G$.

Let $K$ be a subgraph of $G$. A \DEF{$K$-bridge} $B$ in $G$ (or a
\DEF{bridge} $B$ of $K$ in $G$) is a subgraph of $G$ which is either an
edge $e \in E(G) \backslash E(K)$ with both endpoints in $K$, or it
is a connected component of $G-K$ together with all edges (and their
endpoints) between the component and $K$. The vertices of $B \cap K$
are the \DEF{attachments} of $B$. A vertex of $K$ of degree
different from 2 in $K$ is called a \DEF{branch vertex} of $K$.
A \DEF{branch} of $K$ is any path in $K$ (possibly closed) whose
endpoints are branch vertices but no internal vertex on this path is
a branch vertex of $K$. Every subpath of a branch $e$ is a
\DEF{segment} of $e$.
If a $K$-bridge is attached to a single branch
$e$ of $K$, it is said to be \DEF{local}. Otherwise it is called
\DEF{stable}. The number of branch vertices of $K$ is denoted by
$\bsize(K)$.

In this paper, we use the concept ``cylinder''. Let $G$ be a graph
embedded in a surface $S$. Let $C_1, C_2$ be non-contractible curves
in the same homotopy in $S$ (which is not a sphere) that do not cross.
Then a \emph{cylinder} $W$ is an embedded subgraph of $G$ bounded by curves $C_1,C_2$. So $W$ can be considered as a plane graph with the outer face boundary $C'_1$, and with the inner face boundary face $C'_2$, such that the face $C'_i$ is obtained by
cutting along this curve $C_i$ for $i=1,2$.
Hence all the vertices of $G$ hitting the curve $C_i$ must be in the face $C'_i$ of the cylinder for $i=1,2$. Note that $C'_1$ and $C'_2$ could intersect, but since $C_1, C_2$ do not cross, we may assume that
$C'_1$ is the outer face boundary and $C'_2$ is the inner face boundary.

\subsection{2-connected components, Triconnected components and decomposition}

In this paper, we want to work on 3-connected graphs. The
importance of 3-connectivity stems from the fact that if a planar graph is 3-connected (triconnected), then it has a unique embedding on a sphere. Hence
an efficient algorithm that decomposes a graph into triconnected components is sometimes
useful as a subroutine in problems like planarity testing and planar graph isomorphism.

We now define this decomposition formally.

A \emph{biconnected component tree decomposition} of a given graph $G$ consists of a tree-decomposition $(T,R)$
such that for every $tt' \in E(T)$, $R_t \cap R'_t$ consists of a single vertex and for every $t \in T$, $R_t$ consists of a 2-connected graph (i.e., block).
$T$ is called a \emph{biconnected component tree}.

Let $G$ be a 2-connected graph. A \emph{triconnected component tree decomposition} of $G$ consist of
a tree-decomposition $(T,R)$
such that for every $tt' \in E(T)$, $R_t \cap R'_t$ consists of exactly two vertices and for every $t \in T$,
the torso $R_t^*$, which is obtained from $R_t$ by adding an edge between $R_t \cap R_{t'}$ for all $tt' \in T$, consists
of a 3-connected graph (i.e., a 3-connected graph or a triangle or a $k$-bond for $k \geq 3$, i.e., two vertices with $k$ edges between them).
$T$ is called a \emph{triconnected component tree}.

The followings are known in \cite{bill}. Their algorithmic parts are from Hopcroft
and Tarjan \cite{tri}.
\begin{theorem}\label{2conunique}\showlabel{2conunique}
For any graph $G$, a biconnected component tree decomposition is unique. Moreover, there is an $O(n)$ time algorithm to construct
a biconnected component tree decomposition.
\end{theorem}

\begin{theorem}\label{3conunique}\showlabel{3conunique}
For any 2-connected graph $G$, a triconnected component tree decomposition is unique. Moreover, there is an $O(n)$ time algorithm to construct
a triconnected component tree decomposition.
\end{theorem}

\section{Overview of our algorithm}
\label{overview}

Theorem \ref{thm:main1} can be shown by two steps. The first step is our structural theorems. This is the most technical part.
So let us give a sketch of our proof in the next subsection.
The second step is concerning ``map isomorphism'' which will
be detained in the following subsection.

\subsection{Structural results and their proof techniques}

Our main structural result is concerning a 3-connected graph $G$ that can be embedded in the surface $S$ of Euler genus $g$, but
cannot be embedded in a surface $S'$ of Euler genus at most $g-1$. Let us point that the standard arguments allow us to reduce to 3-connected graphs
in linear time (see Section \ref{secmain} for more details). Thus the main arguments in this paper deal with 3-connected graphs.

Below, if we say an embedding of $G$ then it means an embedding of $G$ in $S$ of Euler genus $g$.

If $G$ has a polyhedral embedding, then apply Theorem \ref{thm:main2} to obtain all polyhedral embeddings in $O(n)$ time (there are at most $f(g)$
different polyhedral embeddings, where $f(g)$ comes from Lemma \ref{finitely}). This means that we can test graph isomorphism of
two graphs $G_1, G_2$ if both $G_1$ and $G_2$
have polyhedral embeddings, because we have all different polyhedral embeddings of $G_1$ and $G_2$, respectively (Indeed, this is exactly the main result in \cite{kmstoc08}. Essentially, we can reduce the graph isomorphism problem to the ``map isomorphism problem'', because one map of a polyhedral embedding of $G_1$ is map isomorphic to some map of a polyhedral embedding of $G_2$, if $G_1$ and $G_2$ are isomorphic. See Theorem \ref{thm:main3}). So the difficult case is when
$G$ does not have a polyhedral embedding. So let us consider the following case:

\medskip

{\bf Case A.} $G$ does not have any polyhedral embedding, but has an embedding of face-width exactly two.

\medskip

One difference between Case A and the polyhedral embedding case is that there may be exponentially many embeddings.
Figure \ref{figa} illustrates an example on a torus that has exponentially many embeddings.
To see this, degree four vertices could be embedded in two ways.
\begin{figure*}
\centering
\includegraphics[height=6cm]{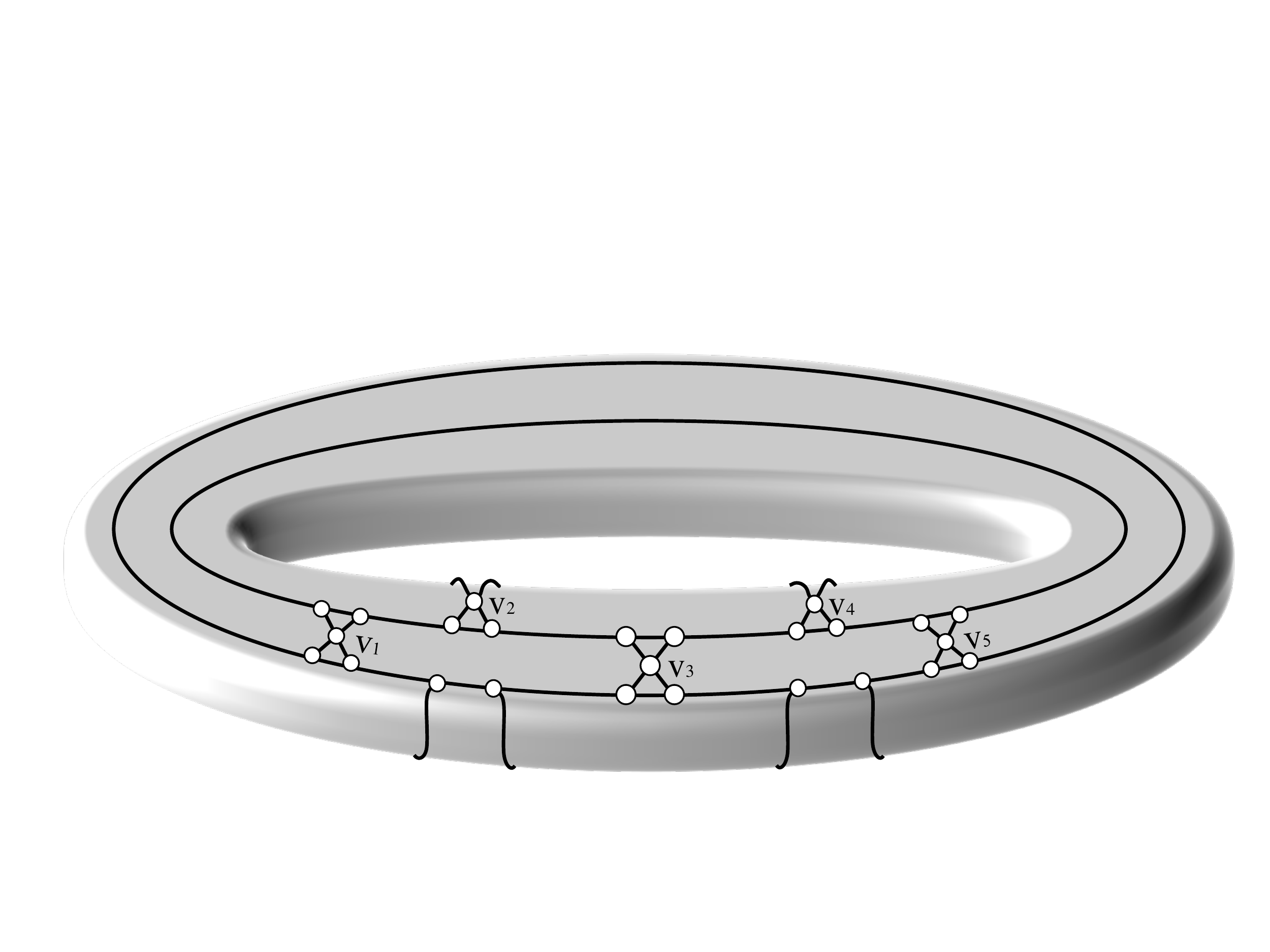}
\caption{Exponentially many embeddings on torus}
  \label{figa}
\end{figure*}
So the embedding is more flexible and the flexibility of ``bridges'' is the main issue.
But we can see from Figure \ref{figa} that if we cut along some two non-contractible curves of order two, then we obtain a ''thin'' cylinder that contains all flexible bridges.

To be more precise, let us look at Figure \ref{fige}. What we want is to take a curve $C_1$ hitting only $c,d$ and a curve hitting $C_2$ hitting only $e,f$.
Then we obtain the graph bounded by $C_1$ and $C_2$, which is the ''cylinder'' we want to take and which contains all flexible bridges.
Then we want to recurse our algorithm to
the rest of the graph. Note that all the non-contractible curves that are homotopic to $C_1$ (and $C_2$) and that hits exactly two vertices are in
this ``thin'' cylinder. Moreover the rest of the graph can be embedded in a surface of smaller Euler genus. 

This figure motivates us what to do.
Specifically, concerning the structural result for Case A, we try to find, in $O(n)$ time,
a constant-sized collection of pairs of subgraphs that contain
all non-contractible curves that hit exactly two vertices in some embedding of face-width two, as follows:
\begin{figure*}
\centering
\includegraphics[height=4cm]{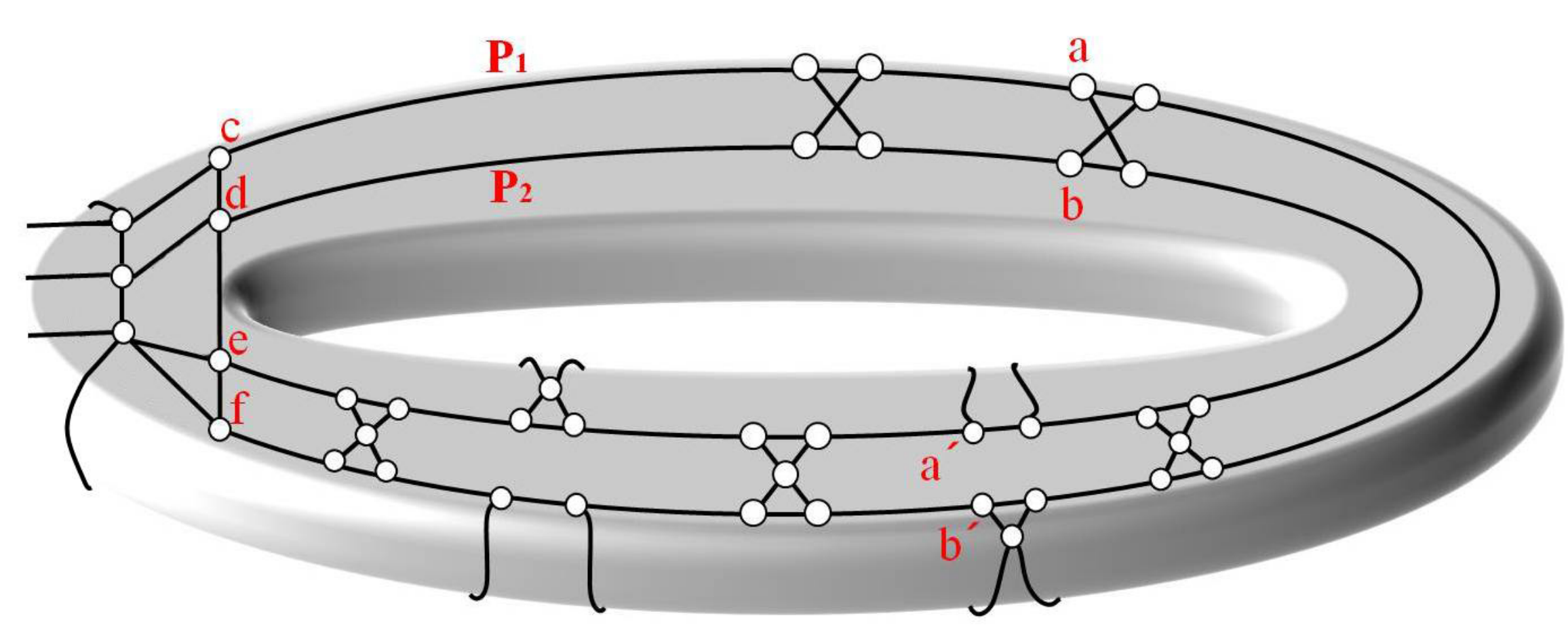}
\caption{Finding a constant-sized collection of pairs of subgraphs of $G$ that contain
all non-contractible curves that hit exactly two vertices in some embedding of face-width two.}
  \label{fige}
\end{figure*}

{\bf Structural Result:}
There is a $q'(g)$ for some function $q'$ of $g$ such that
\begin{enumerate}
\item
there are $q' \leq q'(g)$ pairs $(G'_1,L'_1),\dots,(G'_{q'},L'_{q'}) \in \mathcal{Q}$,
\item
pairs $(G'_i, L'_i)$ are canonical in a sense that graph isomorphism would preserve these pairs (see more details at the end of Case A for the meaning of this item),
\item
for all $i$, $G=G'_i \cup L'_i$ and $|G'_i \cap L'_i| = 4$,
\item
for all $i$, $G'_i$ can be embedded in a surface of Euler genus at most $g-1$,
\item
for all $i$, $L'_i$ is a cylinder with the outer face $F_1$ and
the inner face $F_2$ with the following property: there
 is a non-contractible curve $C_j$ that hits exactly two vertices $x_j,y_j$ in some embedding of $G$ of face-width two for $j=1,2$, and $x_1, y_1$ are contained in $F_1$ and $x_2, y_2$ are contained in $F_2$  (so $L'_i$ attaches to the rest of the graph $G'_i$ at vertices $x_1, x_2, y_1, y_2$),
\item
an embedding of $G$ of face-width two in $S$ can be obtained from
some embedding of $G'_i$ in  a surface of Euler genus at most $g-1$ and
the embedding of the cylinder $L'_i$ by identifying the respective copies of $x_1,x_2, y_1$ and $y_2$ in $G'$ and $L'$ (so $G'_i$ also contains all the vertices
$x_1,x_2,y_1,y_2$ and they are on the border of $G'_i$ and $L'_i$, respectively), and
\item
for any non-contractible curve that hits exactly two vertices $x, y$ in some embedding of $G$ of face-width two, both $x$ and $y$ are contained in $L'_i$ for some $i$.
\end{enumerate}

In Figure \ref{fige}, the cylinder bounded by the non-contractible curve $C_1$ hitting only $c,d$ and the non-contractible curve $C_2$ hitting only $e,f$,
is $L'_i$, and the rest graph obtained by splitting $c,d,e,f$ is $G'_i$.

\medskip

\paragraph{Remark for the non orientation-preserving case.}
We need to clarify difference between the orientation-preserving case and the non orientation-preserving case.
In 1-7 above, we only deal with the orientation-preserving curve. On the other hand, when we deal with the non orientation-preserving curve,
there is one difference. Namely in 5, the definition of the cylinder
is different. Figure \ref{figb} tells us what happens to the non-orientation-preserving curve. We first split $a$ and $b$ into $a,a'$ and $b,b'$ respectively.
Then we flip the component containing $a'$ and $b'$. This is what happens in Figure \ref{figb}.

Now suppose there is a non orientation-preserving curve $C$ of order exactly two. Then it is straightforward to see that there is a face $W$ that
any non orientation-preserving curve of order exactly two that is homotopic to $C$ must hit two vertices of $W$ (see Figure \ref{figc}).
Following Figure \ref{figc}, we cut along the curve through $a$ and $b$, and then split $a$ and $b$ into $a,a'$ and $b,b'$ respectively, and finally we flip the component containing $a'$ and $b'$, as in Figure \ref{figb}. Then we obtain the situation as in Figure \ref{figd}. Namely, we have a new face $W'$
which is obtained from $W$ by taking the part between $a$ and $b$, and the flipped part between $b$ and $a$ (i.e., the upper part between $b'$ and $a'$ in Figure \ref{figd}). Then all non-contractible curves of order exactly two that are homotopic to $C$ must hit two vertices of the resulting face $W'$, with one vertex in the upper part between $b'$ and $a'$, and the other vertex in the lower part between $a$ and $b$.

Intuitively,
what we need for 5 is to cut along $e=x_1, d=y_1$, and to cut along $c=x_2, f=y_2$ in Figures \ref{figc} and \ref{figd}, with the condition that there is no
non-contractible curve of order exactly two that hits two vertices of the face $W'$, with one vertex in the upper part between $e$ and $c$ and the other vertex in the lower part between $f$ and $d$. Then what we obtain is the following:

\begin{quote}[5']
$L'_i$ is a planar graph with the outer face boundary $W$ with four vertices $x_1, y_2, x_2,y_1$ appearing in this order listed when we walk along $W$, with
the following property: there
 is a non-contractible curve $C_j$ that hits exactly two vertices $x_j,y_j$ in some embedding of $G$ of face-width two for $j=1,2$.
 See Figure \ref{figh}, which will be explained later.
\end{quote}

But all other points (1-4, 6,7) are the exactly same.

\begin{figure*}
\centering
\includegraphics[height=6cm]{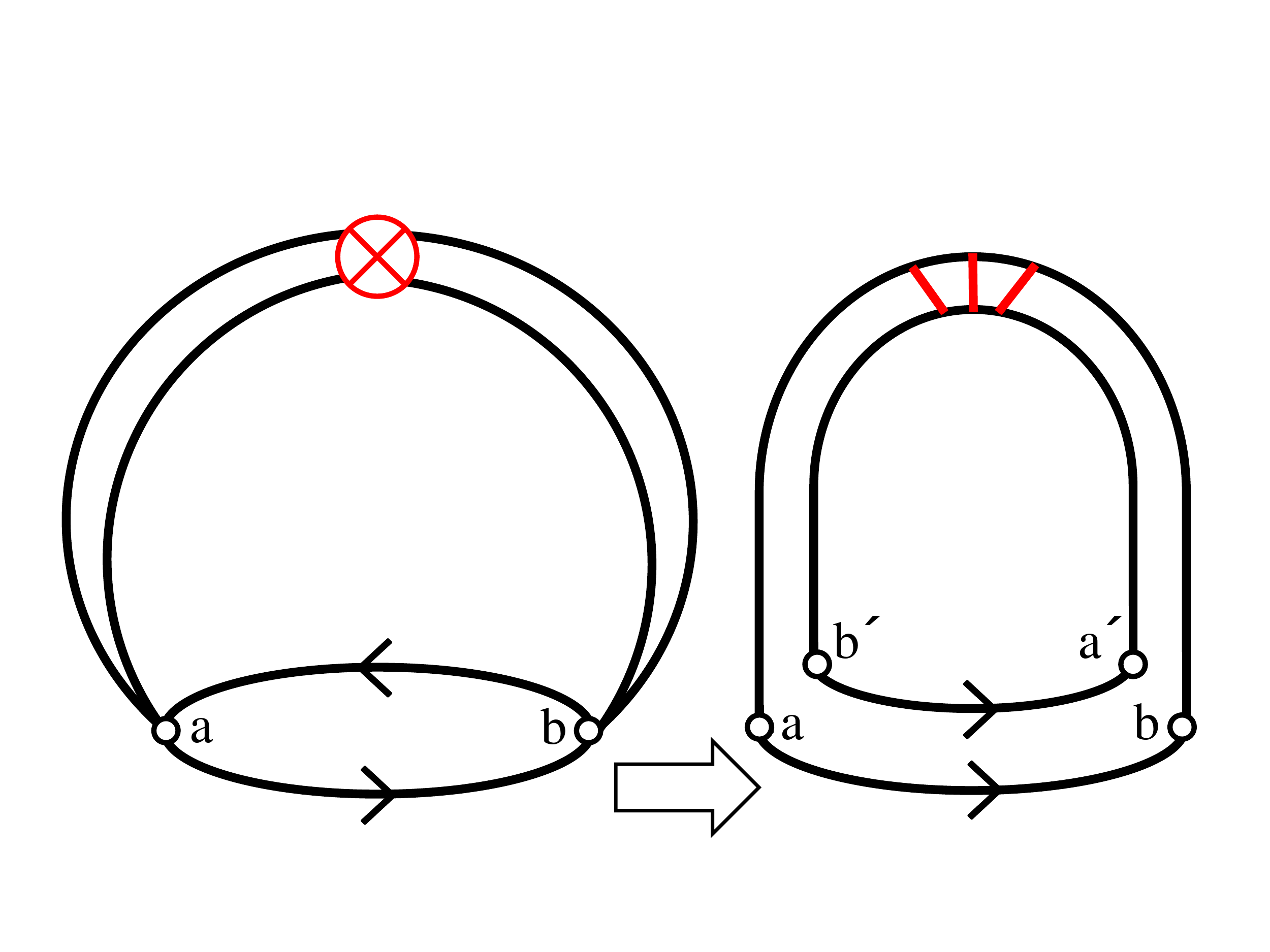}
\caption{Cutting along a non-orientation-preserving curve}
  \label{figb}
\end{figure*}

\begin{figure*}
\centering
\includegraphics[height=6cm]{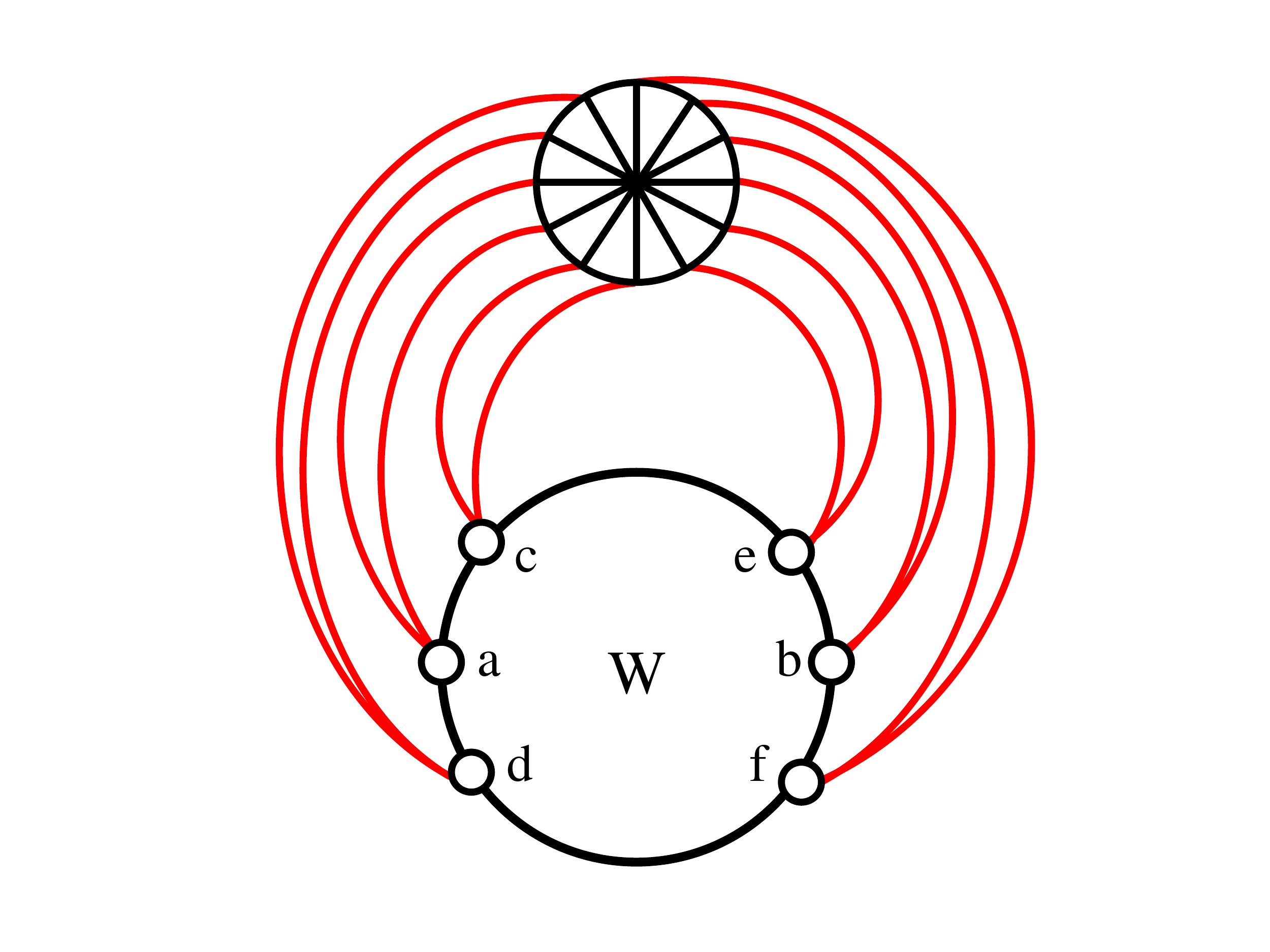}
\caption{Non orientation preserving curves of order two}
  \label{figc}
\end{figure*}

\begin{figure*}
\centering
\includegraphics[height=8cm]{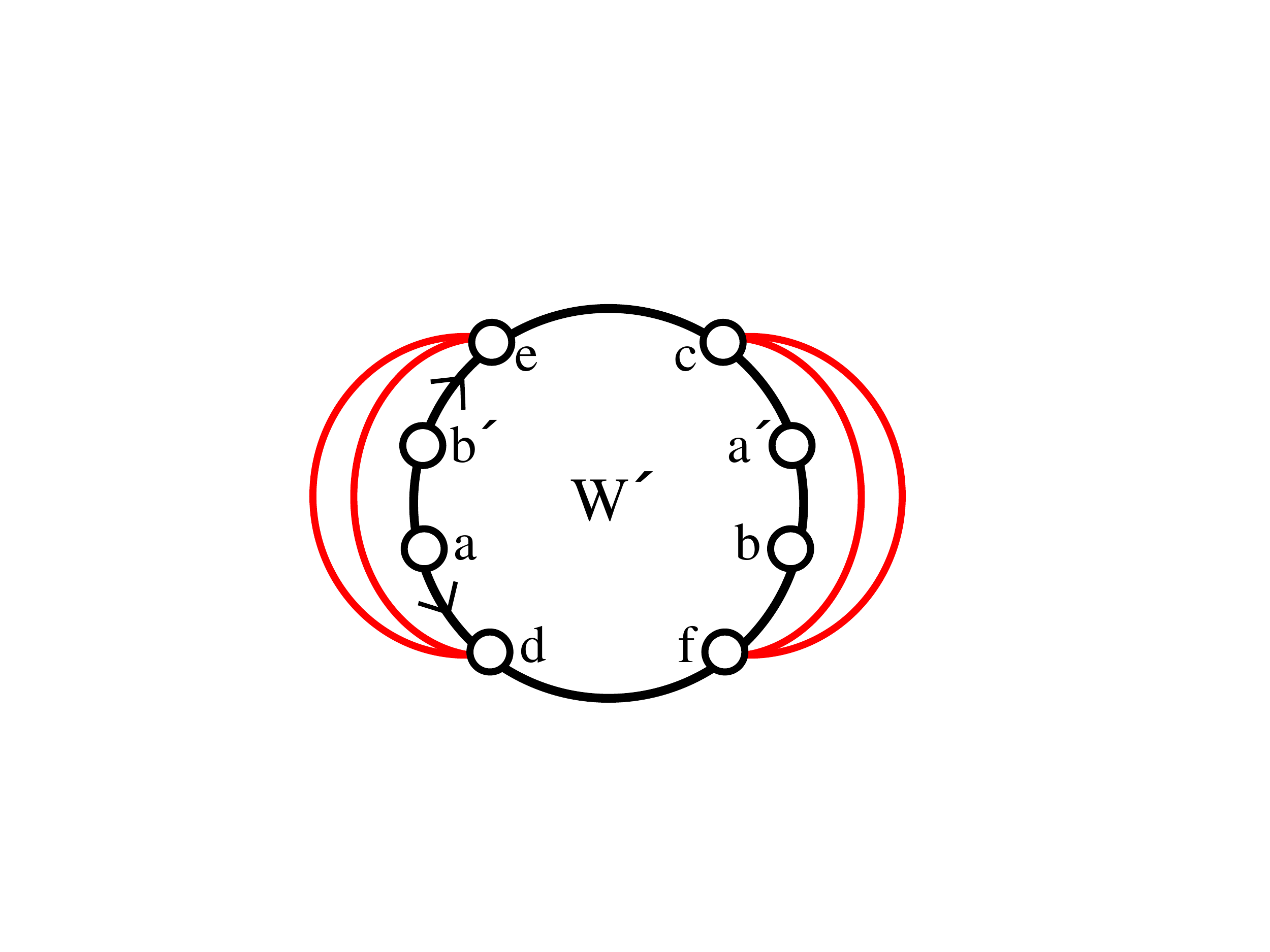}
\caption{Cutting along a non orientation-preserving curve. We cut along the curve through $a$ and $b$, and flip the component
as in Figure \ref{figb}.}
  \label{figd}
\end{figure*}

\medskip

{\bf Remark 1.} Let us observe that we only care about non-contractible curves of length two that are NOT separating, because the graph $G$ is 3-connected.
Moreover, if $G$ is a cylinder with the boundaries $C_1$ and
$C_2$ (so $G$ is obtained by gluing $C_1$ and $C_2$), then we have to do something else because in this case $G'_i$ could be empty but
$G$ itself is $L'_i$.
This is exactly the case when the surface $S$ is torus or the Kleinbottle, and moreover, cutting along a non-contractible curve of length two
reduces the Euler genus by two (thus when $S$ is the Kleinbottle, $H$ neither is surface-separating nor hits only one crosscap).
This ``degenerated'' case has to be dealt with separately, which is done in Theorem \ref{algotorus}.\par

\bigskip

The proof for this structural result consists of the following two step solutions:
\begin{enumerate}
\item[(1)]
Find a set of ``subgraphs(skeletons)'' ${\bf F'}$ in $G$ that can be extended to all the face-width two embeddings of $G$, in $O(n)$ time. Moreover, each subgraph $F' \in {\bf F'}$ has bounded number of branch vertices (that only depends on Euler genus $g$). The important property of ${\bf F'}$
is that each face-two embedding of $G$ can be obtained by
extending some member in ${\bf F'}$ (see below for more details).

Specifically
subgraphs(skeletons) ${\bf F'}$, together with some choice of ``bridge'' embeddings, give rise to all the face-width two embeddings of $G$.
\item[(2)]
Given a set of the subgraphs ${\bf F'}$,
we want to (in a canonical way) produce pairs $(G'_i,L'_i)$ (as above) that cover all the vertices that are contained in some non-contractible
curve of order two in some embedding of $G$.
\end{enumerate}

Let us give more details for (1) first.
The idea is that any embedding of $G$ in the surface $S$ of Euler genus $g$ can be obtained as the following two-stage process.
\begin{enumerate}
\item
We choose the subgraph $F'$ together with its embedding,
in a set of (embedded) subgraphs ${\bf F'}$ of $G$.
So $F'$ can be thought of a ``skeleton'' for the embedding of $G$.
\item
For every bridge of $F'$ in $G$, we choose a face of the embedding of $F'$ where to draw this bridge.
\end{enumerate}

Below, we present the properties of the subgraph $F'$ and of the set ${\bf F'}$, which we find in $O(n)$ time\footnote{Finding the set ${\bf F'}$ in $O(n)$ is also one of the most technical part. The proof was given in \cite{kmstoc08}, but for the completeness, we give a proof
in Section \ref{secmin} and in Section \ref{appendix1}.}, and are detailed in Lemma~\ref{expansion}.

\begin{enumerate}
\item
For each $F' \in {\bf F'}$,
$F'$ is in one of minimal (with respect to edge deletion and contraction) graphs of face-width two in $S$ of Euler genus $g$.
\item
$|{\bf F'}| \leq l(g)$ for some function $l$ of $g$.
\item
For each $F' \in {\bf F'}$, $\bsize(F') \leq l'(g)$ for some function $l'$ of $g$.
\item
For every embedding of $G$ of face-width two in $S$, there is a subgraph $F'$ (with
its corresponding embedding $II$ of face-width two) in ${\bf F'}$ such that the embedding $II$ of
$F'$ can be extended to this embedding of $G$.

Hence the embedding of $G$
can be seen as the embedding of $F'$, with some bridges embedded into   faces of the embedding of $F'$.
\item
Moreover, we can assume that every aforementioned bridge of $F'$ in $G$ is stable.
\end{enumerate}
More details concerning (1) are described in Section~\ref{secmin}.

\medskip

Let us move to (2).
We now try to (in a canonical way) produce pairs $(G'_i,L'_i)$ that cover all vertices $Q$ that are contained in some non-contractible
curve that hits exactly two vertices in some embedding of $G$ that extends the embedding of the skeleton $F' \in {\bf F'}$.

Here is a crucial observation.
\begin{quote}
Since the embedding of $F'$ is already of face-width two, all such vertices $Q$ are, in fact, in $F'$ (i.e., any non-contractible curve
of order two has to hit two vertices of $F'$). See Sections \ref{secmin} and \ref{sectwo} for more details.
\end{quote}

Here, we need to bound the number of homotopy types.
In Section~\ref{sechom} (see Lemma~\ref{homology}), it is shown that
there are at most $f(g)$ homotopy classes to consider. More specifically, we show that
curves from at most $f(g)$ homotopy classes may hit exactly two vertices of $F'$.
So it remains to produce pairs $(G'_i,L'_i)$ separately for one fixed graph $F' \in {\bf F'}$ and for one fixed homotopy class, which hereafter we assume.

\medskip

The rest of arguments in (2) are detailed in Theorems~\ref{algotwo} and~\ref{algotorus}. Here we give a sketch of proof of Theorem~\ref{algotwo}, which is
one of the most technical parts in this paper. For simplicity, let us first focus on an orientation-preserving curve.
Roughly, the argument goes as follows.

\medskip

{\bf Phase 1.} We try to find one such a non-contractible curve $C'$ (for some embedding $II'$ of $G$ that extends the embedding $II$ of $F'$). 
This is actually the most technical part of the proof in Theorem \ref{algotwo}, see Claim \ref{clone}.
Indeed, in the proof of Theorem \ref{algotwo}, we give a lengthly and involved proof to find such a non-contractible curve $C'$
in linear time\footnote{If we
are satisfied with an $O(n^3)$ algorithm for Theorem \ref{thm:main1}, then Phase 1 is much easier; we just guess these two vertices $x',y'$, and then add
two ``dummy vertices $z_1,z_2$ to both $G$ and $F$,
such that both $z_1$ and $z_2$ are only adjacent to both $x'$ and $y'$. Let $G'$ be the resulting graph of $G$ and $F'$ be the resulting graph of $F$.
Then we just need to figure out whether or not $G'$ has a face-two embedding that extends the embedding of $F'$. This can be done
in linear time. See more details in Remark 3 right after Theorem \ref{algotwo}.}. Let $x',y'$ be the vertices of $F'$ that this curve hits.


\medskip

{\bf Phase 2.} Once we find such two vertices $x',y'$ from Phase 1, we cut the graph along this curve (i.e., split $x'$ and $y'$ into two copies $x'_1, x'_2$ and $y'_1, y'_2$, respectively, and split the incident edges into the ''left'' side and the ``right'' side, such that the ''left'' side of edges of $x'$ ($y'$, resp.) are only incident with $x'_1$ ($y'_1$, resp.)). See Figure \ref{figf}. Let us remind the reader that at this moment, we only focus on the orientation-preserving curve.
\begin{figure*}
\centering
\includegraphics[height=6cm]{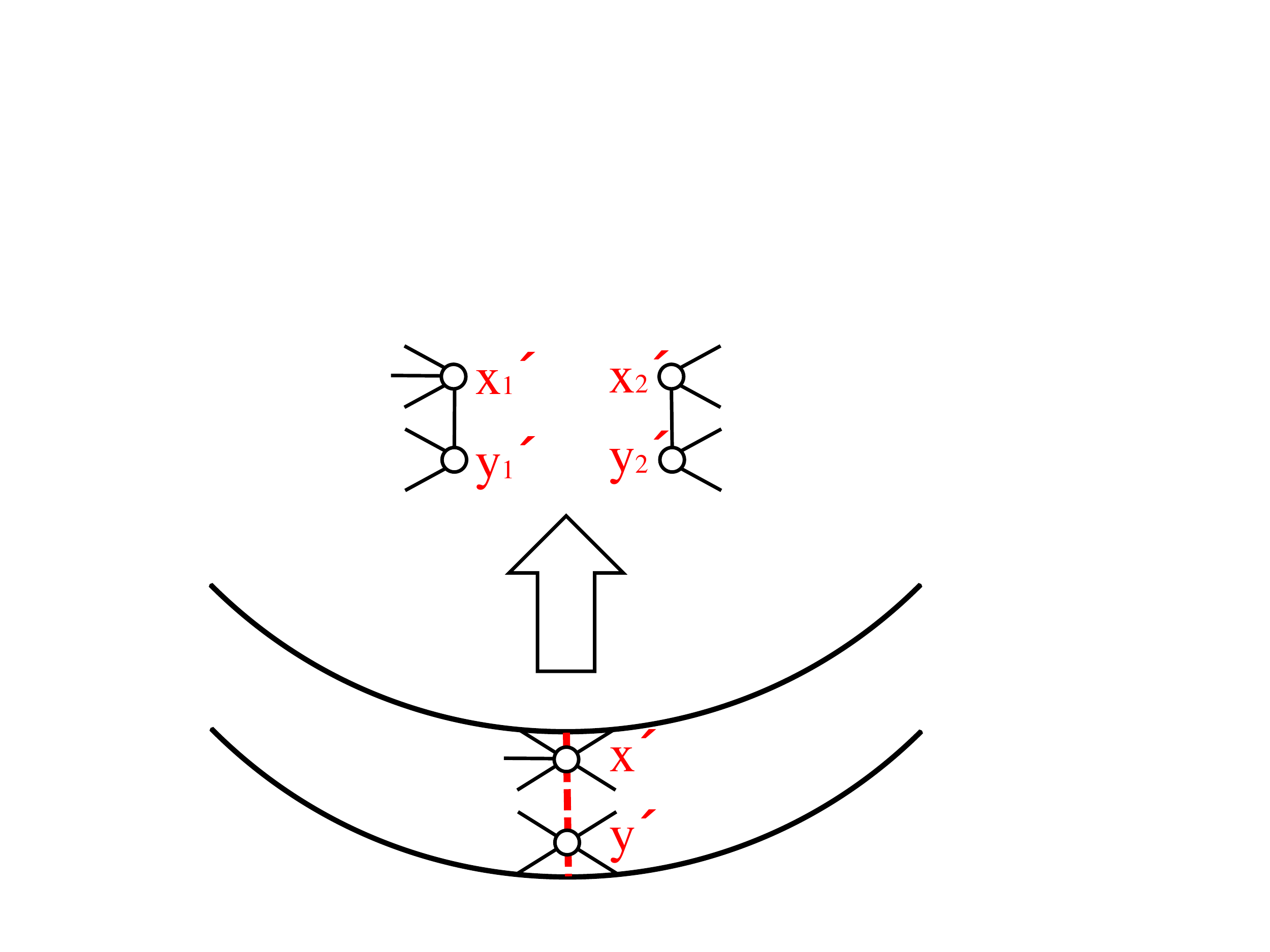}
\caption{Split $x'$ and $y'$ into two copies $x'_1, x'_2$ and $y'_1, y'_2$, respectively.}
  \label{figf}
\end{figure*}

We add the edges $x'_1y'_1$ and $x'_2y'_2$, and
let $G'$ be the modified graph of $G$ after the cutting. If there is a cutvertex in the modified graph $G'$, then it would be a witness for face-width one in the aforementioned embedding (otherwise it would be also a cutvertex in $G$, a contradiction because $G$ is 3-connected). So we can confirm that $G'$ is 2-connected. Hence there are two disjoint paths $P_1, P_2$ between $(x'_1,y'_1)$ and $(x'_2,y'_2)$. In Figure \ref{fige}, if we cut the surface with a non-contractible curve hitting only $a,b$ or $a',b'$, then we can obtain two disjoint paths obtained by $P_1, P_2$.

\medskip

{\bf Phase 3.}  We now apply Theorem~\ref{3conunique} to $G'$ to obtain a triconnected component tree decomposition $(T,R)$.
Note that the triconnected component tree decomposition is unique by Theorem~\ref{3conunique}.
Since $G$ is 3-connected, it can be shown that for any $tt' \in T$, $R_t \cap R_{t'}$ must contain one vertex in $P_1$ and the other vertex in $P_2$
(for otherwise if the separation does not involve at least one of $P_1, P_2$, then
there would be a 2-separation in $G'$ which would be also a 2-separation of $G$, a contradiction to the 3-connectivity of $G$. Note
that edges $x'_1y'_1$ and $x'_2y'_2$ are present, so both $x'_i$ and $y'_i$ are in the same component for $i=1,2$.).
This indeed implies that $T$ is a path $P$ with two endpoints $a, b$ such that
$R_a$ contains both $x'_1$ and $y'_1$ and $R_b$ contains both $x'_2$ and $y'_2$.

\medskip

{\bf Phase 4.} Take the vertex $v$ of $P$ such that
$\bigcup_{v \in P'} R_t$ induces a cylinder $T_1$ with $x'_1, y'_1$ in the outer face boundary $C_1$ and with $v_1,v_2$ in the inner face boundary $C_2$,
subject to that $P'$ is as long as possible, where $P'$ is a subpath of $P$ between $a$ and $v$, and $v_1, v_2 \in R_v \cap R_{v''}$ with $vv'' \in E(P)$ and $v'' \not\in P'$.
Since $\bigcup_{v \in P'} R_t$ induces a cylinder $T_1$, any non-contractible curve hitting only $v_1, v_2$ is in the same homotopy class as $C'$.

Similarly, we take the vertex $v'$ of $P$ such that
$\bigcup_{v' \in P''} R_t$ induces a cylinder $T_2$ with $y'_2, x'_2$ in the outer face boundary $C'_1$ and with $v'_1,v'_2$ in the inner face boundary $C'_2$,
subject to that $P''$ is as long as possible, where $P''$ is a subpath of $P$ between $b$ and $v'$, and $v'_1, v'_2 \in R_{v''} \cap R_{v'}$ with $v'v'' \in E(P)$ and $v'' \not\in P''$.
Again since $\bigcup_{v' \in P''} R_t$ induces a cylinder $T_2$, any  non-contractible curve hitting only $v'_1, v'_2$ is in the same homotopy class as $C'$.

\medskip

Then the cylinder bounded by $C_2$ and $C'_2$ (which is union of the cylinders $T_1$ and $T_2$) yields a desired pair $(G'_i,L'_i)$, where $L_i$ is the cylinder.

\medskip

\paragraph{Correctness.}
We now show that this choice allows us to be canonical; essentially this claim follows from the following two facts:
\begin{enumerate}
\item
The facts that we took the extremal $R_v, R_{v'}$, and
\item
the triconnected component tree decomposition is unique by Theorem~\ref{3conunique}.
\end{enumerate}

It can be shown that
if we start with a different non-contractible curve in the same homotopy class (as $C'$) that hits exactly two vertices,
it is hidden somewhere in the cylinder we constructed, and we would find the same cylinder. This indeed allows us to work on the same graph that can be embedded in a surface of smaller Euler genus, because for each
homotopy class, we obtain the same graph $G_i$. Let us give more intuition from Figure \ref{fige}.
If we start with the curve hitting only $a$ and $b$, we would obtain
the cylinder bounded by curves hitting $c,d$ and $e, f$, respectively. This cylinder certainly contains the curve $C'$ hitting $a'$ and $b'$.
Even we start with the curve $C'$, we would obtain the same cylinder.

%

\medskip

{\bf Remark 2.} Let us briefly look at the non orientation-preserving case. As in Phase 1, suppose we find one such a non-contractible curve $C'$; let $x',y'$ be the vertices of $F'$ that this curve hits. As in Phase 2, we cut the graph along this curve (i.e., twisting the edges of
one part of $x',y'$ by reversing their order in the embedding allows us to split the incident edges into two parts, so that we can define $x'_1,x'_2,y'_1,y'_2$. See Figures \ref{figc} and \ref{figd}.).  In Phase 2, we obtain two disjoint paths $P_1, P_2$, but in this case, $P_1$ joins $x'_1$ and $y'_1$, and $P_2$ joins $x'_2$ and $y'_2$. See Figure \ref{figg}. The rest of the arguments is the same. Note that the ``cylinder'' we shall find corresponds to Figure \ref{figh}.
Namely, we first follow $v_1$ to $v'_2$ along the face $W$, then walk from $v'_2$ to $v'_1$ through the non-contractible curve, then walk from $v'_1$ to $v_2$ through the face $W$, and finally walk from $v_2$ to $v_1$ through the non-contractible curve.
Thus we can obtain $L'_i$ which is a planar graph with the outer face $W'$ with four vertices $v_1,v'_2,v'_1,v_2$ appearing in this order listed when we walk along $W$.
This finishes Case A.

\begin{figure*}
\centering
\includegraphics[height=8cm]{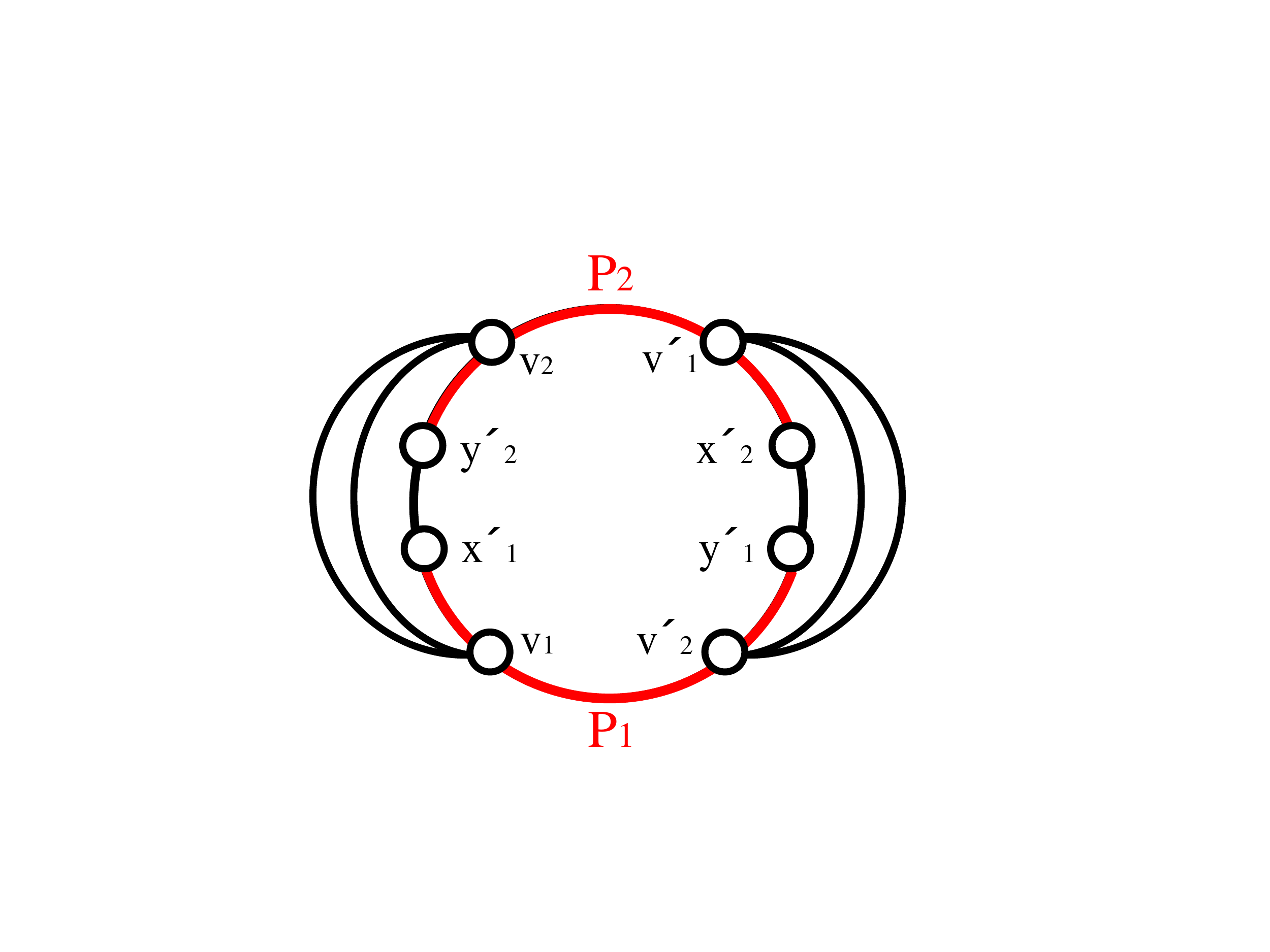}
\caption{The non orientation-preserving case}
  \label{figg}
\end{figure*}

\begin{figure*}
\centering
\includegraphics[height=8cm]{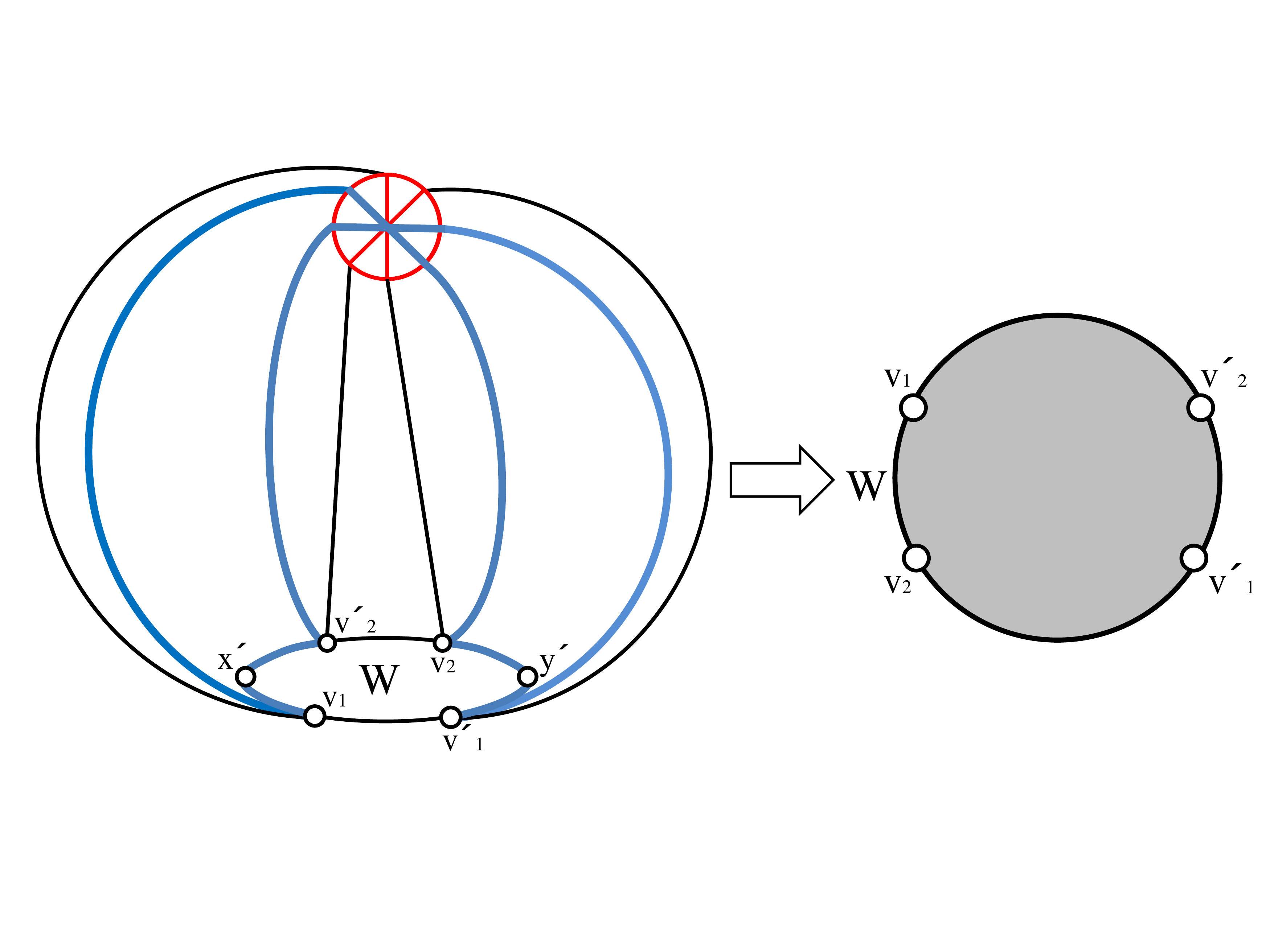}
\caption{Obtaining a cylinder for the non orientation-preserving case}
  \label{figh}
\end{figure*}

\medskip

{\bf Case B.} $G$ does not have any face-width two embedding, but has an embedding of face-width exactly one.

\medskip

In this case, we also have a ``two steps'' solution, as in Case A. As for the first step, we also obtain a set of ``skeletons'', as in Case A.
Then for the second step, we try to obtain, in $O(n)$ time,
the set of vertices $V_1$ of order $q(g)$ (for some function $q$ of $g$) in $O(n)$ time such that for each vertex $c \in V_1$, the following property holds;
\begin{quote}
there is an embedding of $G$ of face-width exactly one with a non-contractible curve $C$ hitting only $c$.
\end{quote}
Moreover, none of the vertices in $G-V_1$ satisfies this property and we are canonical (we will clarify what this means later).

\medskip

Let us look at the first step. To this end, we need the following result,  Theorem 6.1 of Mohar \cite{mohar2} (see Theorem~\ref{obst} later):
\begin{quote}
In $O(n)$, we can obtain a
subgraph $F$ of $G$ that cannot be embedded in a surface of smaller Euler
genus, but can be embedded in $S$. Moreover, $F$ is minimal with respect to this property (i.e, any deletion of
an edge or a vertex of $F$ results in a graph that is embeddable in a surface of smaller Euler genus), and $\bsize(F) \leq l''(g)$ for some function $l''$ of $g$.
\end{quote}

We find all embeddings of $F$ ${\bf F''}= \{\hat F_1,\dots,\hat F_l\}$ such that each of them can be extended to an embedding of $G$
in $O(n)$ time. This is possible since $l''(g)$ is a fixed constant only depending on $g$ (so $l$ is also a fixed constant only depending on $g$).

Moreover we also show a kind of the converse;
\begin{quote}
For each face-width one embedding of $G$ in $S$, there is an
embedding $\hat F_i$ of $F$ in ${\bf F''}$ such that
the embedding $\hat F_i$ can be
extended to the embedding of $G$.
\end{quote}
This can be shown by enumerating all the embeddings of $F$ in $S$ (this is possible since, again,
$l''(g)$ is a fixed constant only depending on $g$).
For more details, see Section~\ref{secone}.

\medskip

Let us move to the second step.
To this end, we first note that a non-contractible curve hitting exactly one vertex
must be orientation-preserving (see Lemma \ref{project}).
To find the vertex set $V_1$, here is a crucial observation.
\begin{quote}
Since the embedding of $F$ in $S$ is already of face-width one, all such vertices $V_1$ are in fact in $F$ (i.e., any non-contractible curve
of order one has to hit one vertex of $F$).  See Sections \ref{secmin} and \ref{secone} for more details.
\end{quote}

In Section~\ref{sechom} (see Lemma~\ref{homology}), it is shown that
there are at most $f(g)$ homotopy classes to consider.
More specifically, we can show that
curves from at most $f(g)$ homotopy classes may hit exactly one vertex of $F$.
So it remains to find such vertices separately for  one fixed homotopy class and for one fixed embedding of $F$.

Our important step is the following; We will show in Lemma \ref{canon1} that
if we walk along a face $W$ in the embedding of $F$, there are no four branches $R_1, R_2, R_3, R_4$ appearing in this order listed when we walk along $W$, such that $R_1=R_3$ and $R_2=R_4$, i.e, $R_1$ appears
twice in $W$ and $R_2$ appears twice in $W$ too.
This implies that
\begin{quote}
we are {\em canonical} in the following sense; suppose there is a non-contractible curve $C'$ that hits exactly one vertex $v$ in a branch $P$ of $F$
in an embedding of $G$ that extends the embedding of $F$. Then $C'$ uniquely splits the incidents edges of $v$ into the ``left'' side and the ``right side'' (note that the curve $C$ is orientation-preserving).
\end{quote}

This allows us to show the following, which will be proved in Lemma~\ref{faceone}.
\begin{quote}
The stable bridges, together 3-connectivity of $G$, give $O(1)$ candidates for an intersection point of a non-contractible curve hitting exactly one vertex  on every face of the embedding of $F$. Moreover, we are canonical.
\end{quote}
This allows us to obtain
the set of vertices $V_1$, as above, in $O(n)$ time.

\subsection{How do the structural results help?}

Our second step is about map isomorphism. Let us first mention
that a \DEF{map} is a graph together with a (2-cell) embedding
in some surface, and that \DEF{map isomorphism} between two
maps is an isomorphism of underlying graphs which preserves
the facial walks of the maps.
For the map isomorphism problem for graphs embeddable in a surface $S$ of Euler genus at most $g$,
we know the following result in \cite{kmstoc08}, which we shall use.

\begin{theorem}
\label{thm:main3}
\showlabel{thm:main3}
For every surface $S$ (orientable or non-orientable),
there is a linear time algorithm to decide whether or not
two embedded graphs in $S$ represent isomorphic maps\footnote{It is trivial to do this in $O(n^2)$ time, as two embeddings are fixed (so we just guess 
which vertex of one graph can map to which vertex of the other vertex).}.
\end{theorem}

The key of our algorithm for Theorem \ref{thm:main1} is that the first structural results allow us to reduce
the graph isomorphism problem for bounded genus graphs to the map isomorphism problem, which can be done by Theorem \ref{thm:main3}.

\subsection{Overview of our graph isomorphism algorithm}

We now give an overview of our algorithm for Theorem
\ref{thm:main1}. Suppose we want to test the graph
isomorphism of two graphs $G_1,G_2$, both admit an embedding in a surface $S$ of Euler genus $g$.

Let us give overview of our algorithm.

\medskip

{\bf Step 1.} Making both $G_1$ and $G_2$ 3-connected.

\medskip

Our first step is to reduce both graphs $G_1$ and $G_2$ to be 3-connected.
This is quite standard in this literature, see \cite{logspace, jgaa}, so we omit details, which will be described in Section \ref{secmain}.

\medskip

{\bf Step 2.} Finding the minimum Euler genus of a surface $S$ for which both $G_1$ and $G_2$ can be embedded.

\medskip

Our second step is to see if we can embed both $G_1$ and $G_2$ in a fixed surface.
This can be  done by a result of
Mohar \cite{mohar1,mohar2}.

\begin{theorem}[Mohar \cite{mohar1,mohar2}]
\label{3} For fixed $g$, there is a linear time algorithm to give
either an embedding of a given graph $G$ in a surface of Euler genus $g$ or
a minimal forbidden minor for the surface of Euler genus $g$ in $G$.
\end{theorem}

Alternatively, we can use a new linear time algorithm by
Kawarabayashi, Mohar and Reed \cite{KMR}. Hereafter, we assume that
both $G_1$ and $G_2$ can be embedded in the surface of Euler genus $g$ (otherwise clearly $G_1$ and $G_2$ are not isomorphic).

In fact, we
would like to know the minimum Euler genus of a surface $S$
for which both $G_1$ and $G_2$ can be embedded. This can be done in
linear time for fixed $g$, since we know that the upper bound of
Euler genus of $G_1$ and of $G_2$ is at most $g$. Hence we just need to apply
Theorem \ref{3} to both $G_1$ and $G_2$ at most $g$ times. Therefore, after
performing Theorem \ref{3} at most $O(g)$ times, we may assume that
both $G_1$ and $G_2$ can be embedded in the surface $S$ of the minimum Euler genus $g$.

\medskip

{\bf Step 3.} For $G=G_1,G_2$, if $G$ has a polyhedral embedding (including a planar embedding), then apply Theorem \ref{thm:main2} to obtain all polyhedral embeddings in $O(n)$ time (there are at most $f(g)$ different polyhedral embeddings, where $f(g)$ comes from Lemma \ref{finitely}). We then go to Step 6. Note that $G$ has a polyhedral embedding
in a surface $S$ of Euler genus $g$ if and only if $G$ has a minimal embedding of face-width three in $S$ as a surface minor (see Section \ref{secmin}). Thus we have a certificate
(from Theorem \ref{thm:main2}) that
$G$ does not have a polyhedral embedding
in a surface $S$ of Euler genus $g$ because there is no minimal embedding of face-width three in $S$ as a surface minor in $G$.


\medskip

{\bf Step 4.} For $G=G_1,G_2$, suppose $G$ does not have any polyhedral
embedding, but has an embedding of face-width exactly two.
Unfortunately in this case, we cannot enumerate all the embeddings as we did in Step 3, because in contrast with the case when $G$ has
a polyhedral embedding, the number of embeddings of face-width exactly two is not quite bounded by a constant.

Instead, in $O(n)$ time,
we enumerate at most $q'(g)$ (for some
function $q'$ of $g$) different pairs of subgraphs
$(G'_1,L'_1),\dots,(G'_{q'},L'_{q'}) \in \mathcal{Q}$ of $G$,
as in Case A above,
such that $G'_i$ can be embedded in a surface of Euler genus at most $g-1$ and $L'_i$ can be embedded in a plane (since it is a cylinder).
Moreover we are canonical, as discussed in Case A.

Then after Step 4, we apply our whole algorithm recursively to each of $G'_i, L'_i$ in the pair $(G'_i,L'_i)$ with ''marked'' vertices $x_1, y_1, x_2, y_2$
both in $G'_i$ and in $L'_i$.
Note that we just need to apply Step 6 to $L'_i$.
For more details, see Section \ref{secmain}.

\medskip

{\bf Step 5.} For $G=G_1,G_2$, if $G$ does not have any face-width
two embedding, but has an embedding of face-width exactly one, then
in $O(n)$ time, we obtain the set of vertices $V_1$ of order at most $q(g)$ (for some function $q$ of $g$) such that for
each vertex $c \in V_1$, there is an embedding of $G$ of face-width
exactly one and moreover there is a non-contractible curve $C$ that hits
only $c$ in this embedding. Furthermore, there is no such a vertex in $G-V_1$ and we are canonical, as discussed in Case B. We shall show this in Theorem
\ref{faceone}.
%
This
allows us to create $q \leq q(g)$ different subgraphs $G_1,\dots,G_{q}$
of $G$ of Euler genus at most $g-1$ that can be obtained from $G$ by
splitting each vertex of $V_1$ into the ``right'' side and the ''left'' side.
Let us observe that at Step 5, we know that $C$ must be orientation-preserving (for otherwise, $G$ can be embedded in a surface of smaller Euler genus, see Lemma \ref{project}, due to Vitray \cite{vite}.)
Then we recursively apply our whole
algorithm (from Step 1) to each of these graphs $G_1,\dots,G_{q}$
with ''marked'' vertices in $V_1$.

\medskip

{\bf Step 6.} Testing graph isomorphism of embedded graphs.

\medskip

When the current graph comes to Step 6, it comes from Step 3. Thus at the moment,
we have either a planar embedding of a 3-connected graph or a polyhedral embedding of a 3-connected graph in some surface.

By Theorem \ref{thm:main3}, we can check map isomorphism of the embedding of some graph $G'_1$ and of the embedding of some other graph $G'_2$ in $O(n)$ time. Note that if $G'_1$ and $G'_2$ are map isomorphic, then $G'_1$ and $G'_2$ are isomorphic.


\medskip

We shall show that after Step 6, we can, in $O(n)$ time, figure out whether or not $G_1$ and $G_2$ are isomorphic in Section \ref{secmain}.

We now discuss time complexity.
Let us observe that in Steps 4 and 5,
we create at most $q'(g),q(g)$ different subgraphs of $G_1$ and of $G_2$, respectively, and we recursively apply
our whole algorithm again to each of these different subgraphs of $G_1$ and of $G_2$. However, when we recurse, we know that
Euler genus of each subgraph already goes down by at least one. Also, note that in Step 3,
we create at most $f(g)$ different subgraphs of $G_1$ and of $G_2$, respectively. Since $g$ is a fixed constant and in addition, we recurse at most $g$ times, therefore in our recursion process,
we create at most $w(g)$ different subgraphs of $G_1$ and of $G_2$ in total, for some function $w$ of $g$.

In Step 6, we can figure out all pairs of graphs $(H_1,H'_1),\dots$ with $H_i \subseteq G_1$ and $H'_i \subseteq G_2$, where
both $H_i$ and $H'_i$ are graphs at Step 6, such that
$H_i$ and $H'_i$ are isomorphic for all $i$ (with respect to the marked vertices). This can be done in $O(n)$ time by Theorem \ref{thm:main3}, since we create at most $w(g)$ subgraphs of $G_i$ for some function $w$ of $g$ in our recursion process ($i=1,2$).

For each subgraph of $G_i$ ($i=1,2$) in Step 6,
we can easily go back to the reverse order of Steps 4 and 5 to come up with the original graphs $G_1$ and $G_2$ in $O(n)$ time, because
in both Steps 4 and 5, we only ``split'' a few vertices, and these
vertices are all marked.
Thus having known all pairs of graphs $(H_1,H'_1),\dots$ with $H_i \subseteq G_1$ and $H'_i \subseteq G_2$ such that
$H_i$ and $H'_i$ are isomorphic for all $i$ (with respect to the marked vertices), we can see if $G_1$ and $G_2$ are isomorphic in $O(n)$ time.

In summary, we create only constantly many subgraphs in our
recursion process.
Since all of Steps 1-6 can be done in $O(n)$ time, so the time
complexity is $O(n)$.

Steps 2, 3 and Step 6 are already described above. So it remains to consider Steps 1, 4 and 5, and the correctness of our algorithm.
Some details of Step 1 will be given in Section \ref{secmain}, but this
is all standard (see \cite{logspace, jgaa}).

The rest of the paper is organized as follows. In Section \ref{secmin}, we give several
facts about minimal embeddings of face-width $k$, which are one key
in our proof. In Section
\ref{sechom}, we define homology in a surface, which is
necessary in our proof.  In Section \ref{sectwo}, we deal with the case when a
given graph has an embedding in a surface $S$ with face-width
exactly two (but does not have an embedding with face-width three).
In Section \ref{secone}, we deal with the case when a given graph
has an embedding in a surface $S$ with face-width exactly one (but
does not have an embedding with face-width two). Finally in Section
\ref{secmain}, we give several remarks for our algorithm for Theorem
\ref{thm:main1}, including the correctness of our algorithm.

\section{Minimal embedding of face-width $k$ and minimal subgraph of face-width $k$}
\label{secmin}
\showlabel{secmin}

Recall
that an embedding of a given graph is \DEF{minimal of face-width $k$},
if it has face-width $k$, but for each edge $e$ of
$G$, the face-width of $G-e$ and of $G/e$ are both less than $k$.
By Theorems 5.6.1 and 5.4.1 in \cite{MT}, any
minimal embedding of face-width $k \geq 2$ has at most $l'(g,k)$ vertices for some function $l'$ of $g,k$ (therefore
there are only bounded number of minimal embeddings of face-width $k$).
Most importantly, a given graph $G$ has an embedding in the surface
$S$ with face-width at least $k$ if and only if $G$ contains a
minimal embedding of face-width $k$ as a surface minor.

Let us now state one result in \cite{kmstoc08}.
\begin{theorem}
\label{find1}\showlabel{find1}
Suppose $g, l$ are fixed integers.
Let $H$ be a graph of order $l$ that is embedded in a surface of Euler genus $g$.

Given a graph $G$ that has an embedding in $S$,
we can determine in $O(n)$ time
whether or not $G$ has $H$ as a surface minor of an embedding of $G$ in $S$.
\end{theorem}

For a completeness of our proof, we give a proof of Theorem~\ref{find1} in the appendix.

We consider a family of minimal embeddings of face-width $k$.
From Theorem~\ref{find1}, we can obtain the following.

\begin{theorem}
\label{forbidden} \showlabel{forbidden}
Suppose $G$ can be embedded in a surface $S$ of Euler genus $g$ with face-width $k \geq 2$ (for fixed $k$).
We can in $O(n)$ find
a family of graphs ${\bf F}= \{F_1,\dots,F_l\}$ with the following properties:
\begin{enumerate}
\myitemsep
\item
For all $i$, the embedding $II_i$ of $F_i$ is a minimal embedding of face-width $k$, and
$F_i$ (with the embedding $II_i$) is a surface minor
of some embedding of $G$ in $S$ of face-width $k$.
\item
$l \leq N(g,k)$ for some function $N$ of $g,k$ (i.e., $|\bf{F}|$ is bounded by some constant only depending on $g,k$).
\item
$|F_i| \leq l'(g,k)$ for all $i$, where $l'$ is some function of $g,k$.
\item
For any embedding of $G$ of face-width at least $k$ in a surface
$S$, there is a graph $F_i$ (with its corresponding embedding
$II_i$ of face-width $k$) in ${\bf F}$ such that this embedding of $G$ has $F_i$ (and its embedding $II_i$)
as a surface minor.
\end{enumerate}

\end{theorem}

The next lemma, which we stick to ''subgraphs'' instead of ''minors'', is easy to show by  reversing the contractions, except for
the last statement of Lemma \ref{expansion}, which will be clarified right after Lemma~\ref{expansion}.

\begin{lemma}
\label{expansion}
\showlabel{expansion}
Let $G$ be a graph that has an embedding in a surface $S$ of the Euler genus $g$.

Suppose ${\bf F}=\{F_1,\dots,F_l\}$ (with their corresponding
embeddings $II_1,\dots,II_l$, respectively) is a set of graphs having all minimal embeddings of
face-width $k \geq 2$ (for fixed $k$) as in Theorem \ref{forbidden}.

Then we can in $O(n)$ time find a family of subgraphs ${\bf F'}=
\{F'_1,\dots,F'_l\}$ of $G$ such that for all $i$, $F'_i$ is
obtained from the surface minor of $F_i$ by reversing the contractions.
Moreover the following holds as well:

\begin{enumerate}
\myitemsep
\item
For all $i$, the embedding $II'_i$ of $F'_i$ (of face-width exactly $k$) in $S$ can be extended from $II_i$.
\item
$l \leq N(g,k)$ (as in the second item of Theorem~\ref{forbidden}).
\item
$\bsize(F'_i) \leq l'(g,k)$ for all $i$, where $l'$ is some function of $g,k$.
\item
The embedding $II'_i$ of $F'_i$ can be extended to an embedding of $G$ in
$S$, by embedding each $F'_i$-bridge in some face of $F'_i$ (we call
this embedding ``\emph{the embedding of $F'_i$ can be extended to an
embedding of $G$}.'')
\item
For any embedding of $G$ of face-width at least $k$ in a surface
$S$, there is a subgraph $F'_i$ (with its corresponding embedding
$II'_i$ of face-width $k$) in ${\bf F'}$ such that the embedding $II'_i$ of $F'_i$ can
be extended to this embedding of $G$.
\end{enumerate}
\end{lemma}

\medskip

{\bf Remark:} Let us clarify the last point which is the key in the
algorithm given in \cite{kmstoc08}. Indeed, we can show the
following.
\begin{quote}
For any embedding of $G$ in a surface $S$ having $H$ as a surface
minor (with $|H| \leq h(g)$ for some function $h$ of $g$), there is
a subgraph $H'$ (with its corresponding embedding $II'$)
that is obtained from the surface minor of $H$
by reversing the contractions ($II'$ is also obtained
from the embedding of the surface minor of $H$ by revising the contractions). Moreover
the embedding of $G$ induces the
embedding $II'$ of $H'$.

Furthermore, if a surface minor of a graph $H$ is
guaranteed to exist in some embedding of $G$,
the above subgraph $H'$ (that is obtained from a surface minor of $H$ by reversing the contractions)
can be found in $O(n)$ time (so we can find the surface minor $H$ as well), even without knowing the actual embedding of $G$.
\end{quote}
The hard part of the above remark is the algorithmic statement (i.e., even without knowing the actual embedding of $G$,
we have to find $H'$ and its embedding in $O(n)$ time). Indeed, the rest follows
trivially from the definition of the surface minor.

To show our algorithmic claim, we need the following result due to Mohar
\cite{mohar1,mohar2} (see Theorem 6.1 in \cite{mohar2}\footnote{We apply this result to the surface of Euler genus exactly $g-1$. 
Then we obtain Theorem \ref{obst}.}.

\begin{theorem}
\label{obst} \showlabel{obst} Let $G$ be a graph that can be
embedded in a surface $S$ of Euler genus $g$, but cannot be embedded in
a surface of smaller Euler genus. Then in $O(n)$, we can obtain a
subgraph $F$ of $G$ that cannot be embedded in a surface of smaller Euler
genus, but can be embedded in $S$. Moreover, $F$ is minimal with respect to this property(i.e, any deletion of
an edge or a vertex of $F$ results in a graph that is embeddable in a surface of smaller Euler genus), and $\bsize(F) \leq l''(g)$ for some function $l''$ of $g$.
\end{theorem}

Let us give a proof of the algorithmic claim. Since the proof is
almost identical to that of Theorem \ref{forbidden}, we just give a
sketch here.

Fix the embedding $II$ of the graph $H$,
and fix one embedding $II''$ of $F$, where $F$ comes from Theorem \ref{obst}.
We first remark that ALL embeddings of $G$ we consider are obtained from some embedding of $F$ by adding all $F$-bridges to
some faces of the embedding of $F$.

The main idea is the following:
By the standard graph minor argument (finding irrelevant vertices, see \cite{RS7}), for each embedding of $F$,
we can modify $F$ so that the embedding of $F$ (and the branch vertices) is
the same, but $F$ is contained
in a small tree-width graph. The same is true for any surface minor $H$ we
consider. So the proof goes as follows: Fix the endpoints of $F$. Apply
the irrelevant vertices argument with respect to the existence of $F$ (and its
fixed embedding) and the existence of $H$ (and its fixed embedding).
Then we obtain a small tree-width subgraph $G'$ of $G$, but both $F$ and $H$ are contained in $G'$. We can then find both $H$ and $F$
(with their embeddings) in $G'$  in linear time by the standard dynamic programming.

Let us be more precise.
We consider which branch vertices of $F$ in the embedding $II''$ can go
to which face in the embedding $II$ of a surface minor of $H$. Again there are $u(g)$ ways to enumerate,
since $\bsize(F) \leq l''(g)$ and $|H| \leq h(g)$. Let us say a \emph{pattern} for each way, and
enumerate all patterns ${\bf P}$.

Fix one pattern $P \in {\bf P}$. As in the proof of Theorem \ref{find1},
we shall try to find a surface minor of $H$ (with the embedding $II$)
satisfying this pattern $P$. This
can be done in $O(n)$ time by mimicking the proof of Theorem
\ref{find1}\footnote{We need to find a surface minor of $H$ with the embedding $II$ satisfying
this pattern $P$. This requires to find a rooted subdivision. This problem is almost the same as the disjoint paths problem, instead of just
finding a surface minor of $H$ only. But the proof given in Appendix for Theorem \ref{find1} works for this problem setting as well}. Note that the proof for Theorem \ref{find1} presented in the appendix  is
the standard way of the graph minor technique, see \cite{RS13}\footnote{An $O(n^2)$ time algorithm is very easy. The difficult part is to get it down to 
an $O(n^2)$ time algorithm.}.

Therefore, by examining all the patterns in ${\bf P}$,
we can enumerate all
surface minors (of $H$ and its embedding $II$) satisfying some pattern
in ${\bf P}$.

Because $F$ in Theorem \ref{obst} cannot be
embedded in a surface of smaller Euler genus, so any embedding of $G$ in $S$ induces a 2-cell embedding of $F$ in $S$.
Hence there is one embedding
of $F$ that can be extended to the embedding of $G$, and therefore
this pattern in ${\bf P}$ (with a surface minor of $H$) is covered by our enumerations.
Thus for any embedding of $G$ in a surface $S$ that is guaranteed to have $H$ as a
surface minor, there is a subgraph $H'$ (with its corresponding
embedding $II'$) that is obtained from a surface minor of $H$ by reversing the minor
operations (and the embedding $II'$ is also obtained from the embedding $II$ of $H$ by reversing the contractions), and moreover the embedding of $G$ induces
the embedding $II'$ of $H'$. Furthermore, even without knowing the actual embedding of $G$,
we can find such a subgraph $H'$ and its embedding $II'$ in $O(n)$ time (so we can find the surface minor $H$ as well). This proves our
claim for our remark.

\medskip

Let us mention one algorithmic result that is needed in this paper,
see \cite{mohar1,mohar2}.

\begin{theorem}
\label{spec} \showlabel{spec} Let $G$ be a graph and $K$ be a
subgraph of $G$. Suppose that $K$ has an embedding $II$ in a surface
of Euler genus $g$. Then in $O(n)$, we can test whether or not the
embedding $II$ of $K$ can be extended to an embedding of $G$ in $S$.
If such an embedding exists, this algorithm can give an embedding of
$G$ in $S$ that extends the embedding $II$ of $K$.
\end{theorem}

In the rest of our proof, given a subgraph $W$ of a 3-connected graph $G$ embedded in a surface $S$,
we want all $W$-bridges to be stable. To do that, we need some  ``local" changes.
Let us make it more precise.
  Let $P$ be a branch of $W$ of length at least two,
and let $Q$ be a path
in $G$ with endpoints $x,y\in V(P)$ and otherwise disjoint from $W$.
Let $W'$ be obtained from $W$ by replacing the path $xPy$
(the subpath of $P$ with endpoints $x$ and $y$) by $Q$;
then we say that
$W'$ is obtained from $W$ by {\em rerouting} $P$ along $Q$, or
simply that $W'$ is obtained from $W$ by {\em rerouting}. Note that $P$ is required to have length at least two,
and hence this relation is not symmetric.  We say that the
rerouting is {\em proper} if
all the attachments
of the $W$-bridge that contains $Q$  belong to $P$.
The following lemma is essentially due to Tutte.
(for the proof, see \cite{K6,MT} for example).

\begin{lemma}
\showlabel{prestable}
Let $G, W$ be as above. Note that $G$ is 3-connected.
Then there exists a subgraph $W'$ of $G$
 obtained from $W$ by a sequence of proper reroutings
such that every $W'$-bridge is stable. Moreover, $W'$ is still embedded in a surface $S$.
\end{lemma}
Note that the last statement trivially follows because proper routings only change a branch but do not changes any branch vertex.
Note also that we can perform a sequence of proper reroutings in $O(n)$ time such that every $W'$-bridge is stable (see \cite{moharal1}).
Hence we obtain the following useful result.

\begin{lemma}
\label{local} \showlabel{local}
Let $G$ and the subgraphs $F'_1,\dots,F'_l$ be as in Lemma \ref{expansion}, and suppose that $G$ is 3-connected.
In $O(n)$ time, we can modify all the subgraphs $F'_1,\dots,F'_l$ by a sequence of proper reroutings,
so that every $F'_i$-bridge is stable for $i=1,\dots,l$.
\end{lemma}

\section{Homotopy and non-contractible curve}
\label{sechom}
\showlabel{sechom}

%
%

In this section, we discuss homotopy on a surface.
We now follow the notation in \cite{mmohar}.
Let $S$ be a surface. A (closed) \emph{curve} in $S$ is a continuous mapping
from $S^1$ to $S$, where $S^1$ denotes the sphere. The curve is
{\it simple} if it is a l-l mapping. A curve $\gamma$ will usually be identified with its
image $\gamma(S^1)$ in $S$, particularly when considering topological properties:
simple closed curves on $S$ correspond to subsets of $S$, homeomorphic to
the sphere. If $G$ is a graph embedded in $S$ then any  cycle in $G$ may
also be viewed as a simple closed curve in $S$.

Two-sided simple closed curves are either bounding (i.e., $S \backslash \gamma(S^1)$ has two
connected components) or non-bounding ($S \backslash \gamma(S^1)$ is connected). One-sided
simple closed curves are always non-bounding.
Recall that closed
curves $\gamma_0, \gamma_1$ from $S^1$ to a surface $S$ are \emph{homotopic} if there is a continuous mapping
$H: S^1 \times [0, 1]$ to $S$ such that $H(s, 0) = \gamma_0(s)$ and $H(s, 1) = \gamma_1 (s)$ for each
$s \in S^1$. The mapping $H$ itself is called a \emph{homotopy} between $\gamma_0$. and $\gamma_1$.

If for
some $s_0 \in S^1$, $\gamma_0(s_0) = \gamma_1(s_0) = v_0$, and there is a homotopy $H$ between $\gamma_0$
and $\gamma_1$ such that $H(s_0, t) = v_0$ for all $t \in [0, 1]$ then the two curves are said
to be \emph{homotopic relative to the point $v_0$}. To distinguish these two types of
homotopy we sometimes use the name \emph{free homotopy} for the general case,
and homotopy in $(S, v_0)$ for the case of homotopy relative to $v_0$.
Homotopy gives rise to the equivalence relation, also termed homotopy,
and the corresponding equivalence classes are called \emph{homotopy classes}. The
trivial homotopy class, for instance, is the class of the constant mapping.

\begin{lemma}
\label{homology} \showlabel{homology}
Let $F$ be a subgraph of $G$ that is embedded in $S$ with face-width $l \geq 2$.
We can specify at most $\bsize(F)^l$ different nontrivial homotopy classes such that any noncontractible curve hitting exactly $l$ vertices of $G$ (in an embedding of $G$ that extends the embedding of $F$) must lie in one of these homotopy classes.
\end{lemma}
\begin{proof}
By the existence of the embedding of $F$, any noncontractible curve $\gamma$ hitting exactly $l$ vertices of $G$ (in an embedding of $G$ that extends the embedding of $F$) must hit exactly $l$ vertices of $F$.  Let us consider this curve $\gamma$ in the embedding of $F$.
Since each face of $F$ in this embedding bounds a disk (because face-width of this embedding is at least $l \geq 2$, see \cite{MT}),
we can modify $\gamma$ so that it only hits branch vertices of $F$ (and moreover the resulting curve is homotopic to $\gamma$).
There are
at most $\bsize(F)^l$ possible choices of $l$ branch vertices of $F$. Thus there are also at most $\bsize(F)^l$
different nontrivial homotopy classes such that any noncontractible curve hitting exactly $l$ vertices of $G$ (in an embedding of $G$ that extends the embedding of $F$) must lie in one of these homotopy classes.

\end{proof}

{\bf Remark:} By the proof of Lemma \ref{homology}, once an embedding of $F$ (of face-width $l$) is given, we can easily find,
in $\bsize(F)^l$ time,
at most $\bsize(F)^l$ different nontrivial homotopy classes such that any noncontractible curve hitting exactly $l$ vertices of $G$ (in an embedding of $G$ that extends the embedding of $F$) must lie in one of these homotopy classes.

\medskip

We also give the following lemma which is useful in our proof when a
given non-contractible curve is not orientation-preserving.

\begin{lemma}
\label{project} \showlabel{project} Let $G$ be a graph embedded in a
non-orientable surface $S$ of the Euler genus $g$. If there is a
non-contractible curve $C$  that is not orientation-preserving such
that $C$ hits exactly one vertex, then $G$ can be embedded in a
surface $S'$ of smaller Euler genus $g'$.
\end{lemma}
This lemma is essentially due to Vitray \cite{vite}.
See more details in Robertson and Vitray \cite{RV}.
The proof goes
as follows. Suppose $x$ is the only vertex that is hit by $C$. Then
$C$ divides the edges incident with $x$ into two parts. Twisting the
edges of one part by reversing their order in the embedding
transforms the embedding to an embedding in a surface $S'$ of
smaller Euler genus $g'$. Thus the lemma follows.

\section{Face-Width Two Case}
\label{sectwo}
\showlabel{sectwo}

Let $G$ be a 3-connected graph that can be embedded in a surface $S$
of Euler genus $g$. As mentioned in Lemma \ref{finitely}, there is a
constant $f(g)$ such that every graph $G$ admits at most $f(g)$
polyhedral embeddings. In this lemma, as mentioned before, we have
to require the face-width of the embedding to be at least 3 (i.e,
polyhedral embedding), since there are 3-connected graphs with
exponentially many non-polyhedral embeddings in any surface (other
than the sphere). So if $G$ does not have a polyhedral embedding in $S$, then
we cannot use Lemma \ref{finitely}.

We now restrict our attention to the face-width two embedding case, i.e, a given graph $G$ does not have
a polyhedral embedding in a surface $S$ but has a face-width two embedding in $S$. Moreover
we assume that $G$ cannot be embedded in a surface $S'$ of
smaller Euler genus (hence the Euler genus of $S$ is positive).

By Theorem \ref{forbidden}, we can in $O(n)$ time find
all graphs ${\bf F}= \{F_1,\dots,F_l\}$ with their embeddings
$II=\{II_1,\dots,II_i\}$ of face-width two, respectively,  with the following properties: each of them is a surface minor of
an embedding of $G$, and each embedding of ${\bf F}$ is a minimal
embedding of face-width $k =2$. Note that some two graphs in ${\bf F}$ may be
isomorphic, but their embeddings are different.

By Lemmas \ref{expansion} and \ref{local} (and the remark right
after Lemma \ref{expansion}), we can in $O(n)$ time obtain a family
of subgraphs ${\bf F'}= \{F'_1,\dots,F'_l\}$ of $G$ such that the
following holds:

\begin{enumerate}
\myitemsep
\item
For all $i$, $F'_i$ is obtained from the surface minor $F_i$ by reversing the minor-operations.
\item
For all $i$, each $F'_i$-bridge is stable.
\item
For all $i$, $F'_i$ is embedded in $S$ of face-width two, and this embedding is
extended from $II_i$.
\item
$l \leq N(g)$ for some function $N$ of $g$.
\item
$\bsize(F'_i) \leq l'(g)$ for all $i$, where $l'$ is some function of $g$.
\item
the embedding of $F'_i$ can be extended to an embedding of $G$ in $S$, by embedding each $F'_i$-bridge in some face of $F'_i$.
\item
For any embedding of $G$ of face-width exactly two in a surface $S$,
there is a subgraph $F'_i$ (with its corresponding embedding $II'_i$ of face-width two)
in ${\bf F'}$ such that the embedding $II'_i$ of $F'_i$ can be
extended to this embedding of $G$.
\end{enumerate}

We now prove the following main result in this section.
\begin{theorem}
\label{algotwo}\showlabel{algotwo}
 Let $G,S,{\bf F'}$ be as above.
Suppose that $G$ does not have a face-width three embedding in $S$. Fix one graph $F' \in {\bf F'}$ with the face-width two embedding in $S$. Fix one
homotopy class $H$ of $S$.

Suppose furthermore that if $g = 2$, then $H$ either is surface-separating or hits only one crosscap.

In $O(n)$ time, we can find two non-contractible curves
$C_1,C_2$  in $H$ that hit exactly two vertices in some embedding of $G$ that extends the embedding of $F'$, with  the following
properties:
\begin{enumerate}
\item
There is a cylinder with the outer face $F_1$ and
the inner face $F_2$ with the following property:
Suppose that $C_i$ hits only $x_i,y_i$ in some embedding of $G$ of face-width two for $i=1,2$. Then $x_1, y_1$ are contained in $F_1$ and
$x_2, y_2$ are contained in $F_2$. Moreover, $F_i$ is obtained  by cutting along $C_i$ for $i=1,2$.
\item
For any embedding of $G$ that extends the embedding of $F'$, there is no curve $C'$ in $H$ such that $C'$ hits exactly two vertices $u,v$,
and at least one of $u,v$ is outside the cylinder.
\end{enumerate}

\end{theorem}
We note that the two curves $C_1$ and $C_2$ may share a vertex (or even two vertices).

\medskip

{\bf Remark 1.} The following proof is somewhat complicated and lengthly. In addition, one crucial lemma, which we call ``Canonical Lemma'' will
be shown at the end of Section \ref{secone}, because it is the most convenient for us to present, first, the proof of ``Canonical Claim'' in the proof
of Lemma \ref{faceone}. Intuition behind the proof can be found in Figure \ref{fige}.
What we want is to take a curve $C_1$ hitting only $c,d$ and a curve hitting $C_2$ hitting only $e,f$.
Then we obtain the graph bounded by $C_1$ and $C_2$, which is the ''cylinder'' we want to take and which contains all flexible bridges.
Then we want to recurse our algorithm to
the rest of the graph. Note that all the non-contractible curves that are homotopic to $C_1$ (and $C_2$) and that hits exactly two vertices are in
this ``thin'' cylinder. So the reader may consult Figure \ref{fige}.

However, if we are
satisfied with an $O(n^3)$ algorithm in Theorem \ref{algotwo}, the proof will be much easier and shorter.
The difficulty of the proof below actually comes from Claim \ref{clone}.
But if we are satisfied with an $O(n^3)$ time algorithm, we can do this in a much easier and simpler way, as follows:
just guess
two distinct vertices of $F'$\footnote{As remarked below, we only have to guess two vertices in $F'$}, and try to see if we can cut along these two vertices of $G$ to
obtain an embedding of smaller Euler genus, by applying
Theorem \ref{spec} with the corresponding embedding of $F'$ (this is also done in the proof of Claim \ref{clone} (see Line 7
in the proof of Claim \ref{clone})). For this argument, we have to assume ``Canonical Lemma'', but
the rest of the arguments are exactly the same as the proof below.

\medskip

{\bf Remark 2.} If $g=2$ and $H$ neither is surface-separating nor hits only one crosscap, then the following proof does not work. This is
exactly the case when $x_1=x_2$ and $y_1=y_2$ (i.e, the cylinder
is ``degenerated'' in this sense).
We need Theorem \ref{algotorus}.

\medskip

{\bf Remark 3.} If $H$ is not orientation-preserving,
then a curve $\bar C$ in $H$ divides the edges incident with any vertex $x$
of $\bar C$ into two parts (See Figure \ref{figb}).  When we cut the graph $G$ along $\bar C$, we
obtain the embedding of $G'$ in a surface $S'$ of smaller
Euler genus $g'$ such that the vertices that are hit by $\bar C$ can be
``splitted'' into two vertices, and twisting the edges of
one part of the vertex $x$ in $\bar C$ by reversing their order in the
embedding of $G$ transforms to the embedding of $G'$  in a surface $S'$ of smaller Euler genus $g'$ (See Figures \ref{figb}, \ref{figc} and \ref{figd}).

As we have discussed ``Remark for the non orientation-preserving case'' in Overview of our algorithm, we need to change the cylinder.
Namely, if we obtain two curves $C_1$ and $C_2$ in $H$ as above,
after cutting along $C_1$ and $C_2$, we can obtain a planar graph with the outer face boundary $W$
that contains both the vertices in $C_1$ and in $C_2$ (instead of getting a cylinder as above), and moreover
$x_1,y_2,x_2,y_1$ appear in this order listed when we walk along $W$.
We also include this case as a ``cylinder''. See Figure \ref{figh}.

\medskip

\medskip

\begin{proof}
Since this embedding of $F'$ is a face-width two embedding in $S$,
if there is a non-contractible curve $C$ in $H$ that hits exactly two vertices in an
embedding of $G$ that extends the embedding of $F'$, then $C$ must
hit two vertices in $F'$.

We first show the following, which is trying to find just one
non-contractible curve $\bar C$ in $H$:

\begin{claim}\label{clone}
In $O(n)$ time, we can find a non-contractible curve $\bar C$ in $H$ that hits exactly two vertices in an embedding $II$ of $G$ that extends the embedding of $F'$, if it exists.
\end{claim}

\begin{proof}
In the following proof, we do not have to distinguish the ''orientation-preserving'' case and the ''not orientation-preserving'' case for $H$,
because we only need to find ONE non-contractible curve in $H$.

For
two adjacent faces $W_1,W_2$ of $F'$ and for any two vertices $u,v$
in $V(W_1) \cap V(W_2)$ (but not in the same component of $W_1 \cap W_2$),
 we can figure out whether or not there is a non-contractible curve $C'$ in $H$ that hits exactly two
vertices $u, v$ in some embedding of $G$ that extends the embedding of $F'$ in $O(n)$ time, as follows:

\begin{quote}
We just apply Theorem \ref{spec} to $K=F' + uw_1v + uw_2v$, the embedding $\bar K$ of $K$ and $G=G + uw_1v+uw_2v$,
where $+$ means to add two paths $uw_1v$ and $uw_2v$ to $K$ and $G$, respectively
and moreover, $w_1$ must
be embedded in $W_1$, while $w_2$ must be embedded in $W_2$, to obtain the embedding $\bar K$ of $K$. See Figure \ref{fig10}.
\end{quote}

\begin{figure*}
\centering
\includegraphics[height=4cm]{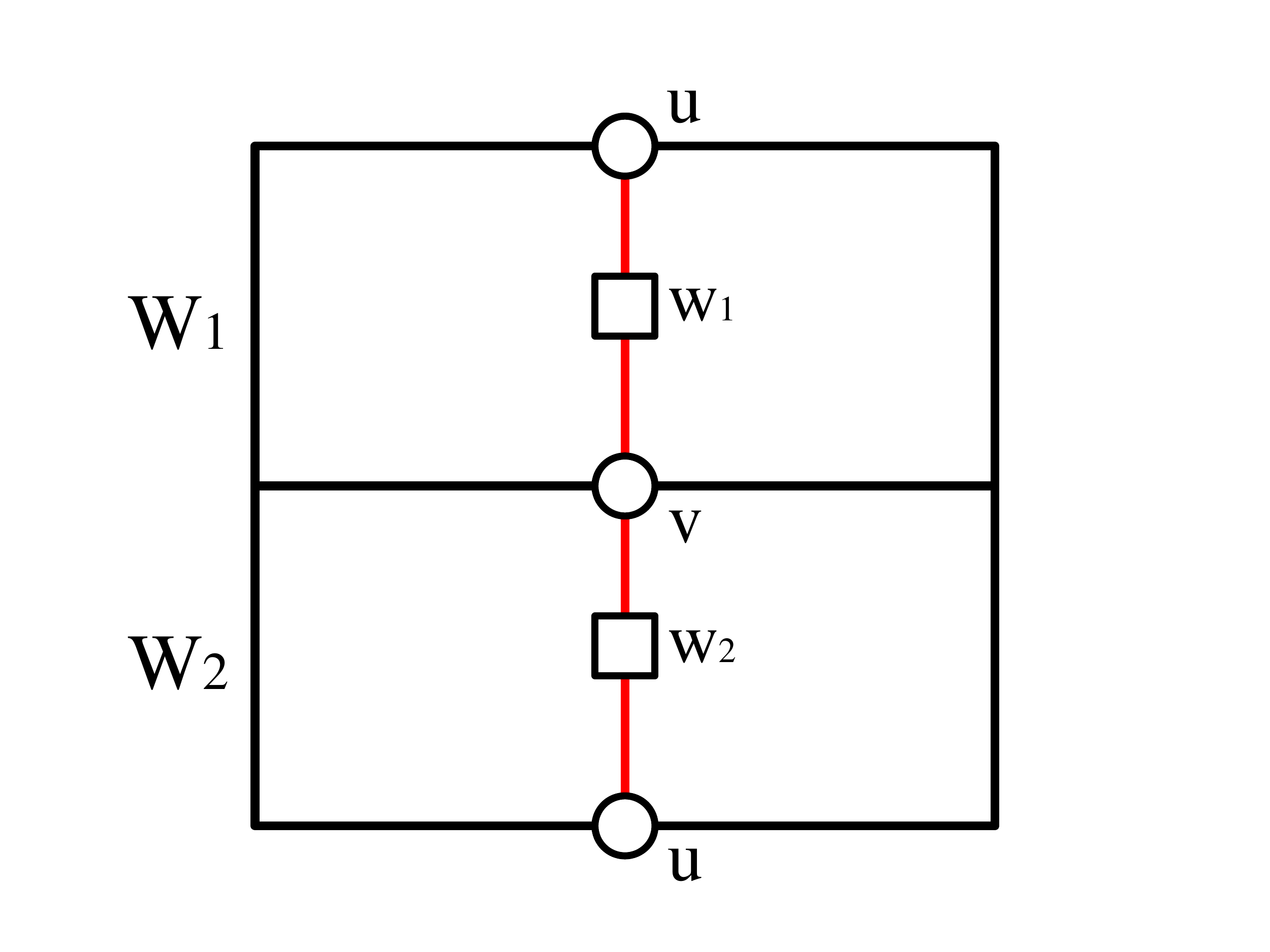}
\caption{Testing a face-width two embedding with non-contractible curve hitting only $u,v$.}
  \label{fig10}
\end{figure*}

Note that the cycle $uw_1vw_2u$ is in $H$.
Since $\bsize(F') \leq l'(g)$ for some function $l'$ of $g$,
it remains to show that, given any two adjacent faces $W_1, W_2$ of $F'$,
\begin{quote}
in $O(n)$ time, we can find the vertices $u$ and $v$
in $V(W_1) \cap V(W_2)$.
\end{quote}

Since $G$ is 3-connected (and since each $F'$-bridge is stable),
any vertex in $V(W_1) \cup V(W_2)$, except for the branch vertices of $F'$, is an attachment of some $F'$-bridge that is stable.
It follows that:\\
~\\
(1) Each bridge having an attachment in some component of $W_1 \cap W_2$
must be embedded in $W_1 \cup W_2$. Moreover, if some bridge $B$ has an attachment in some component of $W_1 \cap W_2$ but also has an
attachment in $W_1$ ($W_2$, resp.) that is not in any component of $W_1 \cap W_2$, then $B$ has to be embedded in $W_1$ ($W_2$, resp.).\par
~\\
Let us consider branches $R_1, \dots R_k$ of $W_1 \cap W_2$ for some
$k$. Note that these branches are paths (some branch vertex $v$ could be in $W_1 \cap W_2$ though). Again, since the embedding of $F'$ is a face-with two embedding in $S$ and since
each $F'$-bridge is stable, it follows that:\par
~\\
(2) If $B$ is a bridge having attachments in two branches of $W_1$, then $B$ has to be uniquely embedded in $W_1$ except for the case
that $B$ has attachments only in two branches of $W_1 \cap W_2$.\par
~\\
Let us assume that $u$ is in $R_1$ and $v$ is in $R_2$.
We now give the following ``Canonical Lemma'', which is crucial in the proof of Claim \ref{clone}.
Since we need some tool in Section \ref{secone} and in addition,
it is the most convenient for us to present, first, the proofs of Lemma \ref{canon1} and ``Canonical Claim (1)'' in the proof
of Lemma \ref{faceone}, the proof will be given at the end of Section \ref{secone}.

\begin{quote}{\bf Canonical Lemma.}
 If there is a non-contractible curve $C$ that hits only $u$ and $v$ in the embedding of $G$ that extends the embedding of $F'$,
 then all $F'$-bridges with at least one attachment in $R_1 \cup R_2$ and with at least one attachment outside $R_1 \cup R_2$  are uniquely placed  
 into the ``left'' side and the ``right side'' (or into the ``one'' part and the ``other'' part, if $C$ is non orientation-preserving) 
 of $C$ in $W_1 \cup W_2$. It follows that $C$ uniquely splits the incidents edges of both $u$ and $v$ into the ``left'' side and the ``right side'' (if $C$ is non orientation-preserving, then the ``left'' side and the ``right side'' are replaced by the ``one'' part and the ``other'' part).\par
 
 Moreover, given $W_1, W_2$, in $O(n)$ time, either we can place all $F'$-bridges $\mathcal{B}$ into the ``left'' side and the ``right side'' (or into the ``one'' part and the ``other'' part, if $C$ is non orientation-preserving) of $C$ in  $W_1 \cup W_2$, or we can conclude that such a non-contractible curve $C$ does not exist.
\end{quote}

Assume that there is a non-contractible curve $C$ that hits only $u$ and $v$ as in Canonical Lemma.  Then 
 we can place all $F'$-bridges $\mathcal{B}$ into the ``left'' side and the ``right side'' (or into the ``one'' part and the ``other'' part, if $C$ is non orientation-preserving) of $C$ in $W_1 \cup W_2$, in $O(n)$ time, as claimed in Canonical Lemma.
Let $a_1, b_1$ be the endvertices of $R_1$, and let $a_2, b_2$ be the endvertices of $R_2$, respectively. So they are branch vertices.

By the canonical lemma and since every vertex of $R_1 \cup R_2$ (except possibly for $a_1, a_2, b_1, b_2$) is an attachment of an $F'$-bridge,  we have the following:
\begin{quote}
There are at most two vertices $a', b'$ in $R_1$ with $a'$ closer to $a_1$
with the following properties:

For each vertex $z'$ between $a_1$ and $a'$
(except for $a_1, a'$),
there is a $F'$-bridge $M$ that has an attachment between $a_1$ and $a'$ such that $M$ blocks a non-contractible
curve of order exactly two that is homotopic to $C$ and that contains $z'$.

Moreover,
for each vertex $z'$ between $b_1$ and $b'$
(except for $b_1, b'$),
there is a $F'$-bridge $M$ that has an attachment between $b_1$ and $b'$ such that $M$ blocks a non-contractible
curve of order exactly two that is homotopic to $C$ and that contains $z'$.

The same thing also holds for $R_2$. Let $a'', b''$ be
the corresponding vertices of $a', b'$ in $R_2$.
\end{quote}

Thus we know that $u$ must be between $a'$ and $b'$ and $v$ must be
between $a''$ and $b''$. Let $P'$ be the subpath of $R_1$ between $a'$ and $b'$, and let $P''$ be the subpath of $R_2$ between $a''$ and $b''$.
By Canonical Lemma, all $F'$-bridges that have an attachment in $P' \cup P''$ must have all attachments in $P' \cup P''$.

Let $L$ be a plane graph obtained from $P' \cup P''$ together with
all $F'$-bridges with all attachments in $P' \cup P''$.
So we have an embedding in the cylinder so that $P' , P''$ are two disjoint paths from the inner cycle to the outer cycle (we can find such a planar embedding in $O(n)$ time using any planarity testing algorithm in $O(n)$ time, say \cite{HT2}).

We then find desired $u, v$
in $O(n)$ time by finding a two vertex cut that separates the inner cycle and the outer cycle of the cylinder. Note that two vertices in the vertex cut
must consist of one vertex out of the paths $P', P''$, because $P' , P''$ are two disjoint paths from the inner cycle to the outer cycle. 
Note also that there may be many choices for $u,v$ in $L$, but we only need one choice of $u,v$. 
This proves Claim \ref{clone}.
\end{proof}

\medskip

We now try to complete our proof of Theorem \ref{algotwo} by using
a non-contractible curve $\overline{C}$ obtained in Claim \ref{clone} in $O(n)$ time.

Suppose first that any curve in $H$ is orientation-preserving. Then
any curve $\bar C$ in $H$ tells us which side is ``left'' of $\bar C$ and ``right''
of $\bar C$. Suppose next that any curve in $H$ is not orientation-preserving. Then $\bar C$ divides the edges incident with any vertex $x$
of $\bar C$ into two parts (See Figure \ref{figb}). We now cut the graph $G$ (with the embedding $II$) along $\bar C$ to
obtain the embedding $\bar II$ of $G'$ in a surface $S'$ of smaller
Euler genus $g'$ such that the vertices that are hit by $\bar C$ can be
``splitted'' into two vertices.  If $C$ is orientation-preserving, we split
$\bar C$ into two $\bar C_1$ and $\bar C_2$ so that $\bar C_1$ ($\bar C_2$, resp.)
has neighbors only in the ``left'' side of $\bar C$ (right side of $\bar C$,
resp.). If $\bar C$ is not orientation-preserving, twisting the edges of
one part of any vertex $x$ of $\bar C$ by reversing their order in the
embedding $II$ of $G$ transforms to the embedding $\bar II$ of $G'$  in a surface $S'$ of smaller Euler genus $g'$ (see Figures \ref{figb}, \ref{figc} and \ref{figd}).

\bigskip

Suppose first that $\bar C$ is orientation-preserving and hits two vertices $x,y$. We
now cut the graph $G$ (with the embedding $II$) along $\bar C$ to obtain an embedding of $G'$ in a
surface $S'$ of smaller Euler genus $g'$ such that the vertices
$x,y$ are ``splitted'' into two vertices $x_1,x_2$ and $y_1,y_2$, respectively, and moreover
both $x_1$ and $y_1$ ($x_2$ and $y_2$, resp.) have neighbors only in
the ``left'' side of $\bar C$ (right side of $\bar C$, resp.). See Figure \ref{figf}.

We now add two edge $x_1y_1, x_2y_2$ and
let $G'$ be the resulting graph. We first show:

\medskip

(4) there are two vertex disjoint paths between $x_1,y_1$ and $x_2,y_2$.

\medskip

{\em Proof.} For otherwise, there is a separation $(A',B')$ of order at most one in $G'$ such that
$A'$ contains $x_1,y_1$ and $B'$ contains $x_2,y_2$.
In this case, $A' \cap B'$ induces a non-contractible curve (in $H$) in some embedding of $G$ in the surface $S$,
but this contradicts the fact that
the embedding of $G$ is a face-width two embedding. Thus such two vertex disjoint paths exist\footnote{In Figure \ref{fige}, if we cut the surface with a non-contractible curve hitting only $a,b$ or $a',b'$, then we can obtain two disjoint paths obtained by $P_1, P_2$.}.\qed

\medskip

Similarly, we show that

\medskip

(5) $G'$ is 2-connected.

\medskip

{\em Proof.}  For otherwise, there is a separation $(A, B)$ of order at most one in $G'$.
By our construction and since $G$ is 3-connected, at least one of $A-B$ and $B-A$ must contain at least two vertices of $x_1,x_2,y_1,y_2$, but this is not possible
because of $x_1y_1, x_2y_2 \in E(G')$ and the existences of the two disjoint paths by (4).  Thus $G'$ is 2-connected.\qed

\medskip

Suppose there is a non-trivial separation $(A,B)$ of order exactly two in $G'$ (i.e., $A-B\not=\emptyset$ and $B-A\not=\emptyset$). 
Since $G$ is 3-connected, so both $A$ and $B$ must contain at least one vertex of $x_1, y_1, x_2, y_2$. 
If $A$ contains at most one vertex of $x_1, y_1, x_2, y_2$, say $x_1$, then $x_1 \in A \cap B$, because $x_1y_1 \in E(G')$.
But then $A \cap B$ is also a 2-separation in $G$, a contradiction.
This implies that

\medskip

(6) for any separation $(A,B)$ of order exactly two in $G'$,
$A$ contains all of
$x_1,y_1$ or all of $x_2, y_2$. Moreover,
the two disjoint paths paths $P_1, P_2$ from $x_1,y_1$ to $x_2, y_2$
imply that $A \cap B$ consists of two vertices with one vertex in $P_1$ and the other
in $P_2$.

\medskip

We now apply Theorem~\ref{3conunique} to $G'$ to obtain a triconnected component tree decomposition $(T,R)$.
The important point here is that this triconnected component tree decomposition is unique by Theorem~\ref{3conunique}.

As shown in (6), for any $tt' \in T$, $R_t \cap R_{t'}$ must contain one vertex in $P_1$ and the other vertex in $P_2$.
This indeed implies that $T$ is a path $P$ with two endpoints $a, b$ such that
$R_a$ contains both $x_1$ and $y_1$ and $R_b$ contains both $x_2$ and $y_2$.
Thus we have the following path decomposition: We have
$R_1,R_2,\dots,R_l$ for some integer $l \geq 1$ such that
\begin{enumerate}
\item
$R_1$ contains $x_1, y_1$ and $R_l$ contains $x_2, y_2$,
\item
each $R_i$ has no separation $(A,B)$ with $R_{i-1} \cap R_i$
in $A$, $R_{i+1} \cap R_i$ in $B$ and $|A \cap B|=2$, and
\item
$|R_i \cap R_{i+1}|=2$ for $i=1,\dots,l$.
\end{enumerate}
We now try to test the following from $R_1$, and from $R_l$, respectively.
\begin{quote}
$R_i$ is a cylinder bounded by two cycles $C''_1,C''_2$ with $C''_1$ containing $R_{i-1} \cap R_i$ and
$C''_2$ containing $R_{i+1} \cap R_i$.
\end{quote}
This can be done in $O(n)$ time by the planarity testing.

Take the largest $j_1$ that satisfies the above criteria from $R_1$, and take the smallest $j_2$ that satisfies the above criteria from $R_l$.
Then it is straightforward to see that $\bigcup_{i=1}^{j_1} R_i$ induces a cylinder $T_1$ with $x_1, y_1$ in the outer face boundary $U_1$ and with $v_1,v_2$ in the inner face boundary $U_2$,
where $v_1, v_2 \in R_{j_1} \cap R_{j_1+1}$. Moreover any non-contractible curve hitting only $v_1, v_2$ is in the same homotopy class as $\bar C$.

Similarly,
$\bigcup_{i=j_2}^{l} R_i$ induces a cylinder $T_2$ with $y_2, x_2$ in the outer face boundary $U'_1$ and with $v'_1,v'_2$ in the inner face boundary $U'_2$, where $v'_1, v'_2 \in R_{j_2} \cap R_{j_2-1}$.
Again any  non-contractible curve hitting only $v'_1, v'_2$ is in the same homotopy class as $\bar C$.

Then the cylinder bounded by $U_2$ and $U'_2$ is $L_i$, and hence we obtain a desired pair $(G_i,L_i)$.

\medskip

\paragraph{Canonical issue.}
We now show that this choice allows us to be canonical, which shows the second assertion of Theorem \ref{algotwo}.
Let us give intuition from Figure \ref{fige}.
If we start with the curve hitting only $a$ and $b$, we would obtain
the cylinder bounded by curves hitting $c,d$ and $e, f$, respectively. This cylinder certainly contains the curve $C'$ hitting $a'$ and $b'$.
Even we start with the curve $C'$, we would obtain the same cylinder.

Essentially this canonical claim follows from the following three facts:
\begin{enumerate}
\item
Canonical Lemma allows us to confirm that all $F_i$-bridges with at least one attachment in $R_1 \cup R_2$ and with at least one attachment outside $R_1 \cup R_2$ are uniquely placed into the ``left'' side and the ``right side'' (or into the ``one'' part and the ``other'' part, if $C$ is non orientation-preserving)
 of $C$ in $W_1 \cup W_2$. (see
the definitions in the proof of Claim \ref{clone}),
\item
we take the extremal $R_{j_1}, R_{j_2}$, and
\item
the triconnected component tree decomposition is unique by Theorem~\ref{3conunique}.
\end{enumerate}

The first fact implies that if we can find one non-contractible curve in the same homotopy class (as $\hat C$) that hits exactly two vertices,
then only flexible $F'$-bridges with at least one attachment in $R_1 \cup R_2$ are the ones with attachments all in $W_1 \cup W_2$.
We can then show that
if we start with a different non-contractible curve $C''$ in the same homotopy class (as $\hat C$) that hits exactly two vertices,
it is hidden somewhere in the cylinder we constructed, and we would find the same cylinder. To this end,
if $C''$ is contained in $W_1 \cup W_2$, then $C''$ would be in the cylinder we constructed, because
we can confirm that all flexible $F'$-bridges are those with all attachments in $R_1 \cup R_2$ by Canonical Lemma, and moreover
$C''$ must give rise to a 2-separation in the above proof of Claim \ref{clone}, and hence
we would obtain the same cylinder bounded by the same non-contractible curves.

Assume finally that
$C''$ is not contained in $W_1 \cup W_2$. Much of the same things happens. If $C''$ is contained in some other faces $W'_1, W''_2$, then again by Canonical Lemma,
we can confirm that all flexible $F'$-bridges are those with all attachments in $R'_1 \cup R'_2$,
where $C''$ hit branches $R'_1$ and $R'_2$ that are in the
intersection of $W'_1$ and $W'_2$. By our choice of the cylinder, $C''$ must give rise to a 2-separation in the above proof of Claim \ref{clone}, and hence
we would obtain two homotopic curves of order two, such that
the cylinder bounded by these two curves must contain $C''$. In both cases, $C''$ is
hidden somewhere in the cylinder we constructed, and we would find the same cylinder because the above arguments can apply with $\hat C$ replaced by $C''$.

This indeed allows us to work on the same graph that can be embedded in a surface of smaller Euler genus, because for each
homotopy class, we obtain the same graph $G_i$.

\medskip

\paragraph{Non-orientable case.}
Finally, suppose that $\bar C$ is not orientation-preserving. Much of the
same thing happens. Indeed, ``left side'' and ``right side'' can be
replaced by ``one part'' and ``the other part'' (See Figure \ref{figb}), and all the same arguments give rise to a cylinder bounded by $C_1$ and $C_2$ in $G$, but the definition of the cylinder is changed as in 5 in ``Remark for the non orientation-preserving case'' in Overview of our algorithm.

More specifically, suppose we find one such a non-contractible curve $\bar C$; let $x,y$ be the vertices of $F'$ that this curve hits. We cut the graph along this curve
by twisting the edges of
one part of $x,y$ of $\bar C$ by reversing their orders in the
embedding $II$ of $G$, which transforms to the embedding $\bar II$ of $G'$  in a surface $S'$ of smaller Euler genus $g'$. This allows us to split the incident edges of $x,y$ into two parts, so that we can define $x_1,x_2,y_1,y_2$. See Figures \ref{figb}, \ref{figc} and \ref{figd}.

As in (4), we obtain two disjoint paths $P_1, P_2$, but in this case, $P_1$ joins $x_1$ and $y_1$, and $P_2$ joins $x_2$ and $y_2$. See Figure \ref{figg}. The rest of the arguments is the same. Note that the ``cylinder'' we shall find corresponds to Figure \ref{figh}.
Namely, we first follow $v_1$ to $v'_2$ along the face $W$, then walk from $v'_2$ to $v'_1$ through the non-contractible curve, then walk from $v'_1$ to $v_2$ through the face $W$, and finally walk from $v_2$ to $v_1$ through the non-contractible curve.
Thus we can obtain $L'_i$ which is a planar graph with the outer face $W'$ with four vertices $v_1, v'_2, v'_1,v_2$ appearing in this order listed when we walk along $W'$.

Thus $C_1$ and $C_2$ are as desired, and we can find $C_1,C_2$ in
$O(n)$ time.
\end{proof}

As mentioned in Remark 2 right after Theorem \ref{algotwo},
the above proof for Theorem \ref{algotwo} does not work when
$W=V(G)$ and $G'$ is a cylinder with the boundaries $C_1$ and
$C_2$ (so $G$ is obtained from $G'$ by gluing $C_1$ and $C_2$).
This is exactly the case when the surface $S$ is torus or the Kleinbottle, and moreover, cutting along a curve in $H$ reduces the Euler genus by two
(thus when $S$ is the Kleinbottle, $H$ neither is surface-separating nor hits only one crosscap).
In this case, we also need the following result.

\begin{theorem}
\label{algotorus} \showlabel{algotorus}
 Let $G,S,{\bf F'}$ be as above,  where $S$ is either torus or the Kleinbottle.
Fix one graph $F' \in {\bf F'}$ with the face-width two embedding in
$S$. Fix one nontrivial homotopy class $H$ of $S$ that neither is a surface-separating nor hits only one crosscap.

Suppose that $G$ has an embedding of face-width exactly two that extends the embedding of $F'$, and that contains a
non-contractible curve in $H$ that hits exactly two vertices. In
$O(n)$ time, we obtain the unique circular chain decomposition of
$G$: $B_1,B_2,\dots,B_l$ for some integer $l \geq 1$ such that
\begin{enumerate}
\item
each $B_i$ is a cylinder bounded by two cycles $C_1,C_2$ with $C_1$ containing $B_{i-1} \cap B_i$ and
$C_2$ containing $B_{i+1} \cap B_i$,
\item
each $B_i$ has no separation $(A,B)$ with $B_{i-1} \cap B_i$
in $A$, $B_{i+1} \cap B_i$ in $B$ and $|A \cap B|=2$,
\item
$|B_i \cap B_{i+1}|=2$ for $i=1,\dots,l$, and
\item
if $u \in B_j$ and $u \in B_i$ for $i < j$, then $u \in B_m$
for $m=i,\dots,j$.
\end{enumerate}
So this circular chain
decomposition can be thought of a generalization of
a triconnected  component tree decomposition $(T, R)$ such that $T$ is a path $P$. If we identify two vertices of $R_a$ and two vertices of $R_b$,
where $a, b$ are endpoints of $P$, then we obtain the above circular chain decomposition.
\end{theorem}
\begin{proof}
We follow the notation and the proof of Theorem \ref{algotwo} (in particular,
$\bar C, G', x_1, x_2, y_1,y_2$ are as in the proof of Theorem \ref{algotwo} after the proof of (4)).
Let us observe that (4)-(6) in the proof of Theorem \ref{algotwo} are still true.
Hence in $G'$, there are two disjoint
paths $P_1,P_2$ such that $P_1$ ($P_2$, resp.) joins $x_1$ and $x_2$ ($y_1$ and $y_2$) or $x_1$ and $y_2$ ($y_1$ and $x_2$).

Let us add edges $x_1y_1,x_2y_2$ if they are not present in $G'$.
Since $G$ is 3-connected, by the existence of two disjoint paths
$P_1,P_2$, we have the following chain decomposition:
$B_1,B_2,\dots,B_l$ for some integer $l \geq 1$ such that
\begin{enumerate}
\item
each $B_i$ is a cylinder bounded by two cycles $C_1,C_2$ with $C_1$ containing $B_{i-1} \cap B_i$ and
$C_2$ containing $B_{i+1} \cap B_i$,
\item
$B_1$ contains $x_1,y_1$ and $B_l$ contains $x_2, y_2$,
\item
each $B_i$ has no separation $(A,B)$ with $(B_{i-1} \cap B_i)$ in
$A$, $(B_{i+1} \cap B_i)$ in $B$ and $|A \cap B|=2$,
\item
$|B_i \cap B_{i+1}|=2$ for $i=1,\dots,l-1$ and moreover $B_i \cap
B_{i+1}$ consists of one vertex in $P_1$ and the other vertex in
$P_2$, and
\item
if $u \in B_j$ and $u \in B_i$ for $i < j$, then $u \in B_m$ for
$m=i,\dots,j$.
\end{enumerate}

Note that this is a triconnected  component tree decomposition $(T, R)$ such that $T$ is a path. Note also that
this decomposition can be found in $O(n)$ time by Theorem \ref{3conunique}, since $G' \cup \{x_1y_1, x_2y_2\}$ is
2-connected. Moreover this decomposition is unique.

By the uniqueness of the decomposition mentioned above, it follows that by identifying $x_1$ and $x_2$, and $y_1$ and $y_2$,
we obtain the unique circular decomposition as in Theorem \ref{algotorus}. In particular,
for any non-contractible curve $C'$ in $H$ that hits
exactly two vertices, we can obtain the above unique chain
decomposition $B'_1,B'_2,\dots,B'_{l}$ such that the two vertices
in $C'$ are in one of $B_i \cap B_{i+1}$, and moreover for any
$i=1,\dots,l-1$, $B'_i \cap B'_{i+1}$ corresponds to $B_{j+i} \cap
B_{j+i+1}$ for some $j$.
\end{proof}

\section{Face-Width One Case}
\label{secone}
\showlabel{secone}

Let $G$ be a 3-connected graph that can be embedded in a surface $S$
of Euler genus $g >0$ and of face-width exactly one. In this section, we assume that $G$ can neither be embedded in a surface $S'$ of smaller
genus $g'$ nor be embedded in the same surface $S$ with face-width at least two.

We first prove the following lemma.

\begin{lemma}
\label{showone} \showlabel{showone} Let ${\bf F''}$ be a set of
all embeddings of $F$ (as in Theorem \ref{obst}) in $S$. For each
face-width one embedding of $G$ in $S$, there is an embedding of $F$
in ${\bf F''}$ such that each face in this embedding bounds a disk
(with possibly some boundary vertices appearing twice or more), and
moreover this embedding can be extended to the embedding of $G$.
\end{lemma}
\begin{proof}
Fix one face-width one embedding of $G$ in $S$. This embedding
induces an embedding $II$ of $F$ which can be extended to the embedding
of $G$. Since $F$ cannot be embedded in a surface of smaller Euler
genus,
it follows that each face is bounded by a disk of $II$ (with
possibly some boundary vertices appearing twice or more).
\end{proof}

By Theorems \ref{obst} and \ref{spec}, we can find a
family of all embeddings of $F$ ${\bf F''}= \{\hat F_1,\dots,\hat F_l\}$ such that each of them can be extended to an embedding of $G$
in $O(n)$ time, by taking all possible embeddings of $F$ in $S$ and then applying Theorem \ref{spec} with these embeddings
(Note that $\bsize(F) \leq l''(g)$ for some function $l''$ of $g$, so finding all possible embeddings of $F$ in $S$ can be done in constant time).

Therefore, by Lemmas \ref{local} and \ref{showone}, we can in $O(n)$ time obtain the
following.

\begin{enumerate}
\myitemsep
\item
For all $i$, each $F$-bridge is stable in the embedding $\hat F_i$.
\item
$l \leq N''(g)$ for some function $N''$ of $g$ (since $\bsize(F) \leq l''(g)$).
\item
The embedding $\hat F_i$ can be extended to an embedding of $G$ in
$S$, by embedding each $F$-bridge in some face of $\hat F_i$.
\item
For each face-width one embedding of $G$ in $S$, there is an
embedding $\hat F_i$ of $F$ in ${\bf F''}$ such that each face in this
embedding bounds a disk (with possibly some boundary vertices
appearing twice or more), and moreover the embedding $\hat F$ can be
extended to the embedding of $G$.
\end{enumerate}

Let us now prove the following simple, but important lemma, which tells us how we get ``canonical''.

\begin{lemma}
\label{canon1}
Let $\hat F$ be an embedding of $F$ in $S$. 
Suppose there are four branches $R_1, R_2, R_3, R_4$ appearing in this order listed in a face $W$ of $\hat F$, when we walk along $W$, and
$R_1$ and $R_3$ are same and $R_2$ and $R_4$ are same (See Figures \ref{fig3} and \ref{fig4}).

If there is a non-contractible curve $C_1$ that hits only one vertex $u$ in $R_1$ (and hence in $R_3$) in $\hat F$,
then there is no non-contractible curve $C_2$, which is not homotopic to $C_1$ and which hits only one vertex $v$ in $R_2$ (and hence in $R_4$) in $\hat F$.
\end{lemma}
\begin{proof}
Suppose for a contradiction that such two vertices $u, v$ exist. We shall show that we can embed $F$ in a surface of smaller Euler genus.
This would be a contradiction to Theorem \ref{obst}.

To this end, let us first remind that both $C_1$ and $C_2$ are orientation-preserving by Lemma \ref{project}. So we can split both $u$ and $v$ into two vertices $u_1, u_2$ and $v_1, v_2$, respectively, such that both $u_1$ and $v_1$ have neighbors (in $F$) that are ``right-side'' of $C_1, C_2$, respectively, while both $u_2$ and $v_2$ have  neighbors (in $F$) that are ``left-side'' of $C_1, C_2$, respectively. We now add $u_1u_2, v_1v_2$.

Note that if we paste $R_1$ and $R_3$, and $R_2$ and $R_4$, we obtain a torus. We now delete $u_1u_2, v_1v_2$ from the embedding $\hat F$ of $F$. Then we destroy
the torus from the embedding $\hat F$, and let $\hat F'$ be the resulting embedding. But then we can add a single-cross to the embedding $\hat F'$
by adding $u_1u_2, v_1v_2$. This implies that the resulting embedding of $\hat F'$ has smaller Euler genus. We can now contract edges $u_1u_2, v_1v_2$ to obtain
the original graph $F$. Then the resulting embedding is still in a surface of smaller Euler genus, a contradiction to the assumption on
the minimum Euler genus embedding $\hat F$. See Figures \ref{fig3} and \ref{fig4}.

\begin{figure*}
\centering
\includegraphics[height=6cm]{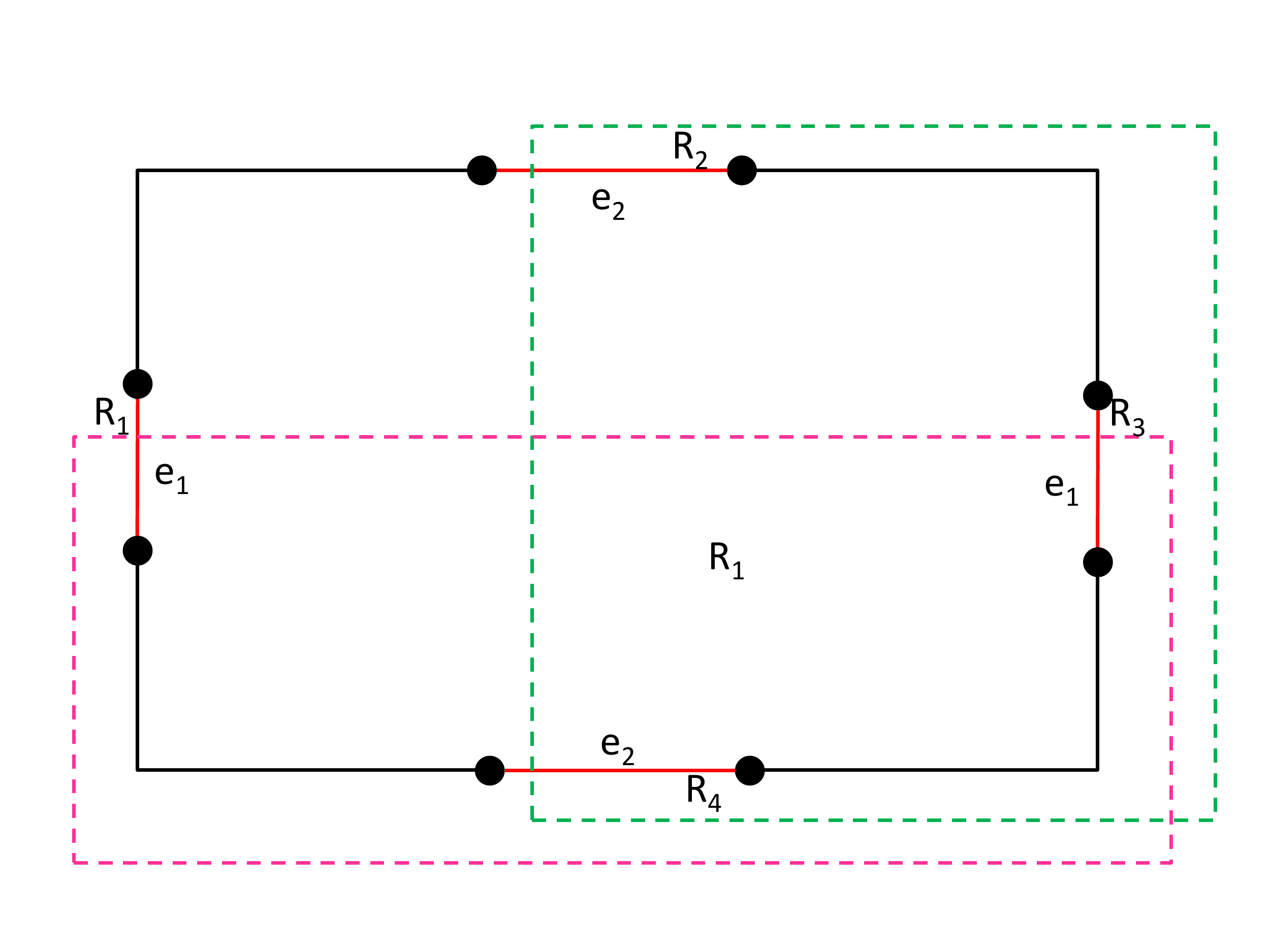}
\caption{Two non-contractible curves of order exactly one}
  \label{fig3}
\end{figure*}

\begin{figure*}
\centering
\includegraphics[height=6cm]{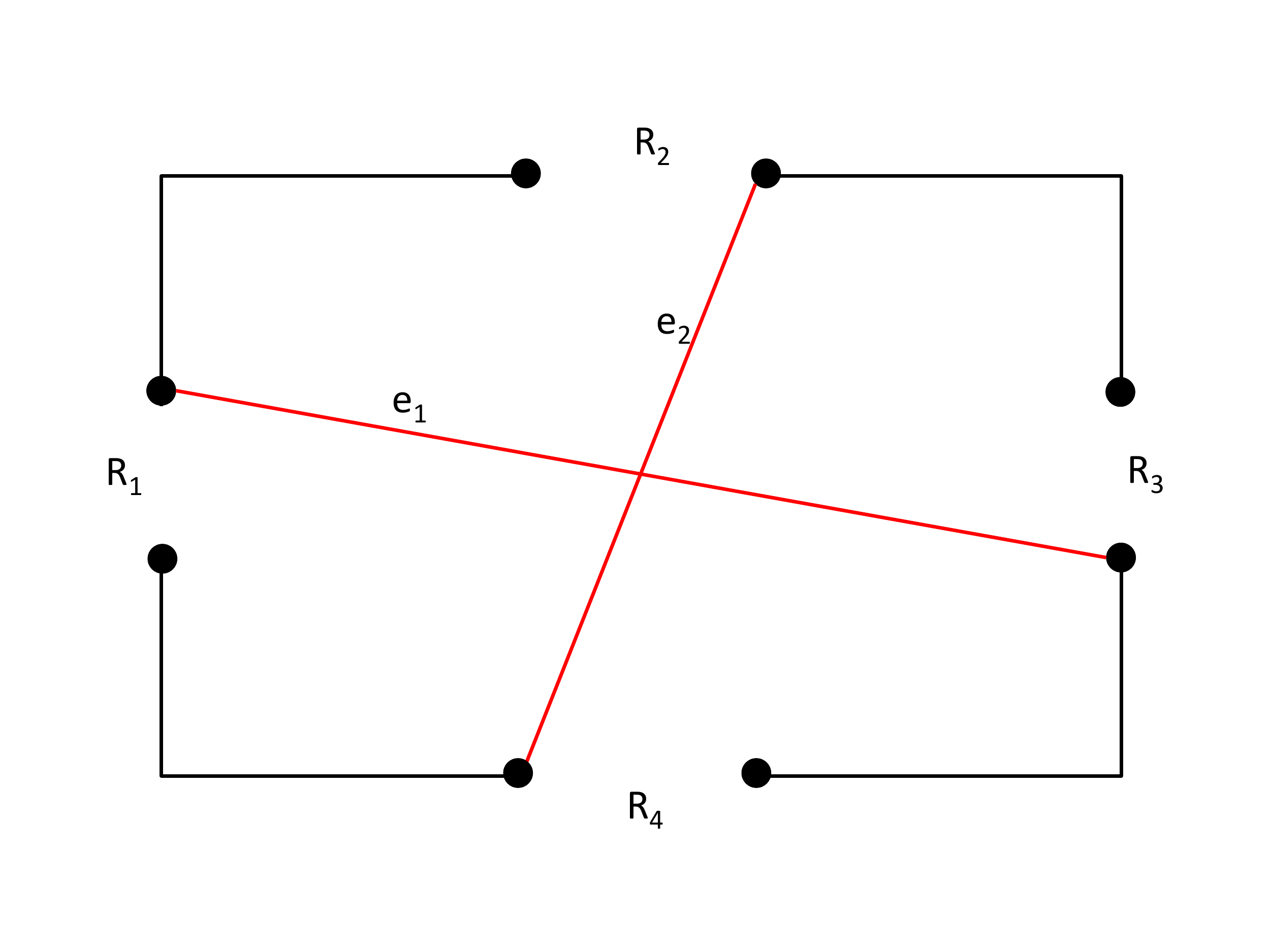}
\caption{Crossing curves}
  \label{fig4}
\end{figure*}

\end{proof}

We now mention the most important result in this section. The following lemma holds.
\begin{lemma}
\label{faceone}
\showlabel{faceone}
Let $G$ be a 3-connected graph that can be embedded in a surface $S$ of Euler genus $g >0$ and of face-width exactly one, but
that can neither be embedded in a surface $S'$ of smaller Euler genus $g'$ nor be embedded in the same surface $S$ with face-width at least two.

Let $F$ be as above.
Then $F$ contains all the vertices $V_1$ such that each vertex $u$
in $V_1$ is hit by a non-contractible curve that hits only $u$ in some
face-width one embedding of $G$ in $S$. Moreover, we can uniquely split the incidents edges of $u$ into the ``left'' side and the ``right side''.

In addition there are at most $q(g)$
such vertices $V_1$ for some function $q$ of $g$, no vertex in $V(G)-V_1$ is hit by a non-contractible curve of order exactly one in any embedding of $G$ in $S$,
and we can find all such vertices $V_1$ in $O(n)$ time.
\end{lemma}

\medskip
\begin{proof}
For each embedding of $G$ of face-width exactly one in $S$,
as mentioned above, there is an embedding $\hat F_i$ of $F$ in ${\bf F''}$ that can be extended to
this embedding of $G$. Therefore, each non-contractible curve that hits  exactly one vertex would hit a vertex in $F$. Thus all the vertices $V_1$ are
in $F$.

Fix one embedding $\hat F_i$ of $F$ in ${\bf F''}$.
Since $l \leq N''(g)$, it remains to find in $O(n)$ time all the
vertices $V'_1$ such that each vertex $u$ in $V'_1$ is hit by a
non-contractible curve that hits only $u$ in an embedding of $G$
that extends the embedding $\hat F_i$, and moreover we are canonical.

In order to figure out whether or not, in some embedding of $G$ that extends the embedding $\hat F_i$, there is a non-contractible curve that hits exactly one vertex $u$ in a face $W$ of $\hat F_i$,
we can do the following:
\begin{quote}
By Theorem \ref{spec} applied to $K=F + uvv'u$, the embedding $\hat K$ of $K$ and $G=G + uvv'u$, where $+$ means to add a path of the form $uvv'u$, $u$ is in $W$  and both
$v$ and $v'$ are dummy vertices in the face $W$, and moreover $uvv'u$ induces a non-contractible cycle in the embedding $\hat K$ of $K$, we can figure out whether or not $u$ can be hit by a non-contractible curve that hits exactly one vertex $u$ in some embedding of $G$ that extends the embedding $\hat F_i$. See Figure \ref{fig11}.
\end{quote}

\begin{figure*}
\centering
\includegraphics[height=4cm]{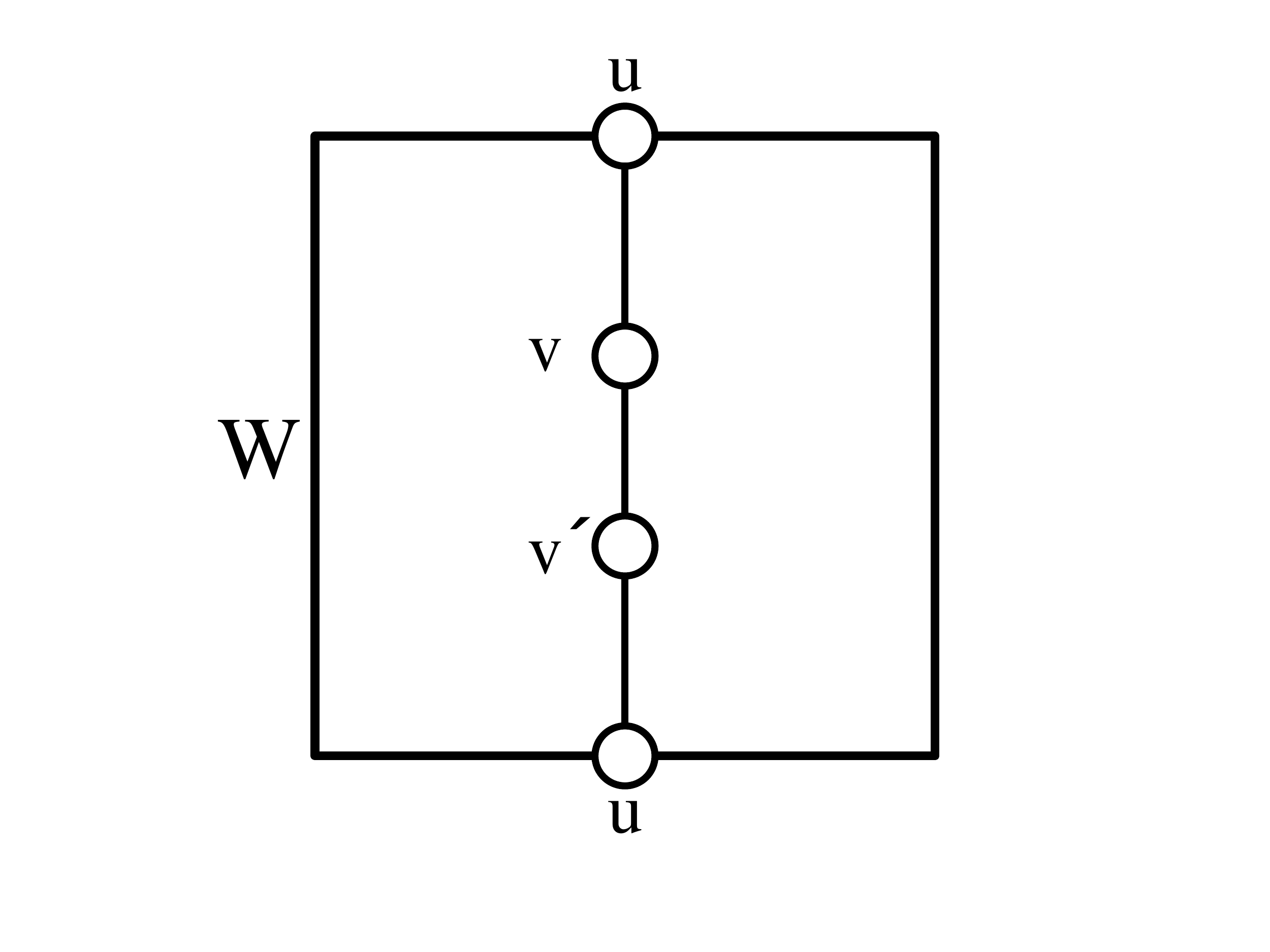}
 \caption{Testing a face-width one embedding with non-contractible curve hitting only $u$.}
  \label{fig11}
\end{figure*}

It remains to show that there are actually at most $O(1)$ vertices $u$
in $W$ that we have to guess, and moreover we are canonical.

Let us consider one face $W$ of $F$. So $W$ consists of branches
$E_1, \dots, E_t$ (for some $t$ that depends on $g$).
Suppose $E_1$ and $E_j$ are same. Then by Lemma \ref{canon1}, any
of $E_2,\dots,E_{j-1}$ is a different branch from $E_{j+1},\dots,E_i$.
This tells one important observation, which follows from Lemma \ref{canon1} and the fact that $F$ cannot be embedded in a surface of smaller Euler genus (by Theorem \ref{obst}):\\
~\\
(1) We are {\em canonical} in the following sense; suppose there is a non-contractible curve $C'$ that hits exactly one vertex $v$ in a branch $P$ of $F$
in an embedding of $G$ that extends the embedding of $F$. Then $C'$ uniquely splits the incidents edges of $v$ into the ``left'' side and the ``right side'' and hence we can uniquely split $v$ into the ``left'' side and the ``right side''  (note that the curve $C$ is orientation-preserving).\par
~\\

\medskip

Let us remind that since $G$ is 3-connected (and since each $F'$-bridge is stable), any vertex in $W$, except for the branch vertices of $F$, is an attachment of some $F$-bridge.
Let $a,b$ be the two branch vertices of $E_1$ (and $E_j$).
By (1), each $F$-bridge having an attachment in $E_1$ (and $E_j$) is uniquely placed either with an attachment in the ``left side''
$E_2,\dots,E_{j-1}$ or with an attachment in the ``right side'' $E_{j+1},\dots,E_i$ or both.
Then there is
an edge $e=a'b'$ of $E_1$ (and $E_j$) with $a'$ closer to $a$ in $E_1$ with the following property: 
For each vertex $z'$ between $a$ and $a'$
(except for $a, a'$),
there is a $F'$-bridge $M$ that has an attachment between $a$ and $a'$ such that $M$ blocks a non-contractible
curve of order exactly one that is homotopic to $C$ and that contains $z'$.
Moreover,
for each vertex $z'$ between $b$ and $b'$
(except for $b, b'$),
there is a $F'$-bridge $M$ that has an attachment between $b$ and $b'$ such that $M$ blocks a non-contractible
curve of order exactly one that is homotopic to $C$ and that contains $z'$.

Since every vertex in $E_1$ (and $E_j$) is an attachment of some $F$-bridge that is stable,
the only candidates for $u$ in $E_1$ (and $E_j$)
are $a', b'$.
Hence there are at most two choices for $u$ as above, and moreover, by (1) we can find such at most two vertices $u$ in $O(n)$ time.
Note that by (1), we can uniquely
split the incidents edges of $u$ into the ``left'' side and the ``right side'', if such a desired curve $C$ that hits $u$ exists.

Since $l \leq N''(g)$ and $\bsize(F) \leq l''(g)$ for some function $l''$ of $g$, and
since we need to test at most $\bsize(F)$
nontrivial homotopy classes by Lemma \ref{homology}, we can detect all the vertices $V_1$ in $O(n)$ time. Moreover,
there are at most $q(g)$ vertices $V_1$
for some function $q$ of $g$, and we can uniquely
split the incidents edges of $u \in V_1$ into the ``left'' side and the ``right side'', if a non-contractible curve $C$ that hits only $u$ exists (and hence we are canonical). This proves Lemma \ref{faceone}.
\end{proof}

\bigskip

As promised Remark 1 right Theorem \ref{algotwo} in Section \ref{sectwo},
we now show the following ``Canonical Lemma'', which is crucial in the proof of Claim \ref{clone} and hence Theorem \ref{algotwo}.
We follow the proof of Lemma \ref{faceone}, together with that of Lemma \ref{canon1} (see (1) in the proof of Lemma \ref{faceone}).
Let us state it again (we shall follow the notations in the proof of Claim \ref{clone}).

\begin{quote}{\bf Canonical Lemma.}
 If there is a non-contractible curve $C$ that hits only $u$ and $v$ in the embedding of $G$ that extends the embedding of $F'$,
 then all $F'$-bridges with at least one attachment in $R_1 \cup R_2$ and with at least one attachment outside $R_1 \cup R_2$  are uniquely placed
 into the ``left'' side and the ``right side'' (or into the ``one'' part and the ``other'' part, if $C$ is non orientation-preserving)
 of $C$ in $W_1 \cup W_2$. It follows that $C$ uniquely splits the incidents edges of both $u$ and $v$ into the ``left'' side and the ``right side'' (if $C$ is non orientation-preserving, then the ``left'' side and the ``right side'' are replaced by the ``one'' part and the ``other'' part).

Moreover, given $W_1, W_2$, in $O(n)$ time, either we can place all $F'$-bridges $\mathcal{B}$ into the ``left'' side and the ``right side'' (or into the ``one'' part and the ``other'' part, if $C$ is non orientation-preserving) of $C$ in $W_1 \cup W_2$, or we can conclude that such a non-contractible curve $C$ does not exist.
\end{quote}

{\em Proof.} Let us remind the reader that we are following the proof of Claim \ref{clone}. In particular, we are assuming (1) and (2).
Following (2), we now look at the case when a bridge $B$ has attachments only in components of $W_1 \cap W_2$ in the embedding of $F'$.  We give the following crucial observation, which is analogue to Lemma \ref{canon1}:\par
~\\
(3) Suppose there are three branches $R_1, R_2, R_3$ appearing in this counter clockwise order listed when we walk along $W_1$, and
suppose furthermore, there are three branches $R'_1, R'_2, R'_3$ in $W_2$ in this counter clockwise
order listed when we walk along $W_2$ (see Figures \ref{fig1} and \ref{fig2}), such that $R_i$ and $R'_i$ are the same for $i=1,2,3$.

If there is a non-contractible curve $C_1$ that hits only two vertices, one $u$ in $R_1$ (and hence in $R'_1$) and
the other $v$ in $R_2$ (and hence in $R'_2$) in some embedding of $G$,
then there is no non-contractible curve $C_2$ that is not homotopic to $C_1$ and that hits only two vertices, one $u$ in $R_1$ (and hence in $R'_1$) and the other in $R_3$ (and hence in $R'_3$)  in this embedding of $G$.\par
~\\
{\em Proof.}  The proof is also identical to that of Lemma \ref{canon1}. For completeness, suppose for a contradiction that such a curve $C_2$ exists.
We show that we can embed $G$ in a surface of smaller Euler genus.
This can be easily achieved as in  Figures \ref{fig1} and \ref{fig2}.
Note that two curves $C_1, C_2$ may be viewed as two
non-contractible curves in the torus or in the Kleinbottle, but
the new contractible curve is going through only one crosscap.

Then the resulting embedding of $G$ is in a surface of smaller Euler genus, a contradiction to the assumption on
the minimum Euler genus embedding of $G$.\qed
\begin{figure*}
\centering
\includegraphics[height=6cm]{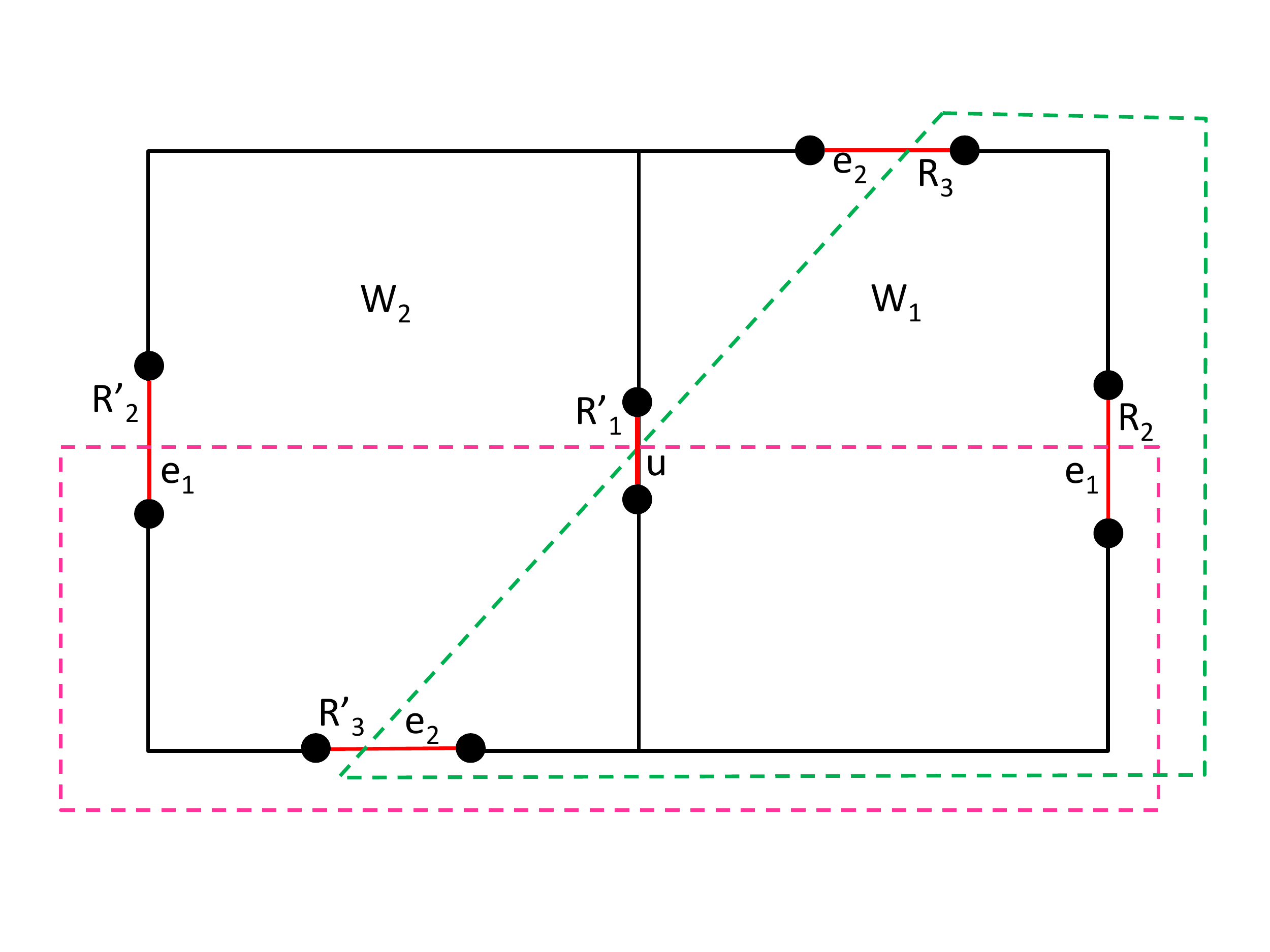}
\caption{Two non-contractible curves of order exactly two}
  \label{fig1}
\end{figure*}

\begin{figure*}
\centering
\includegraphics[height=6cm]{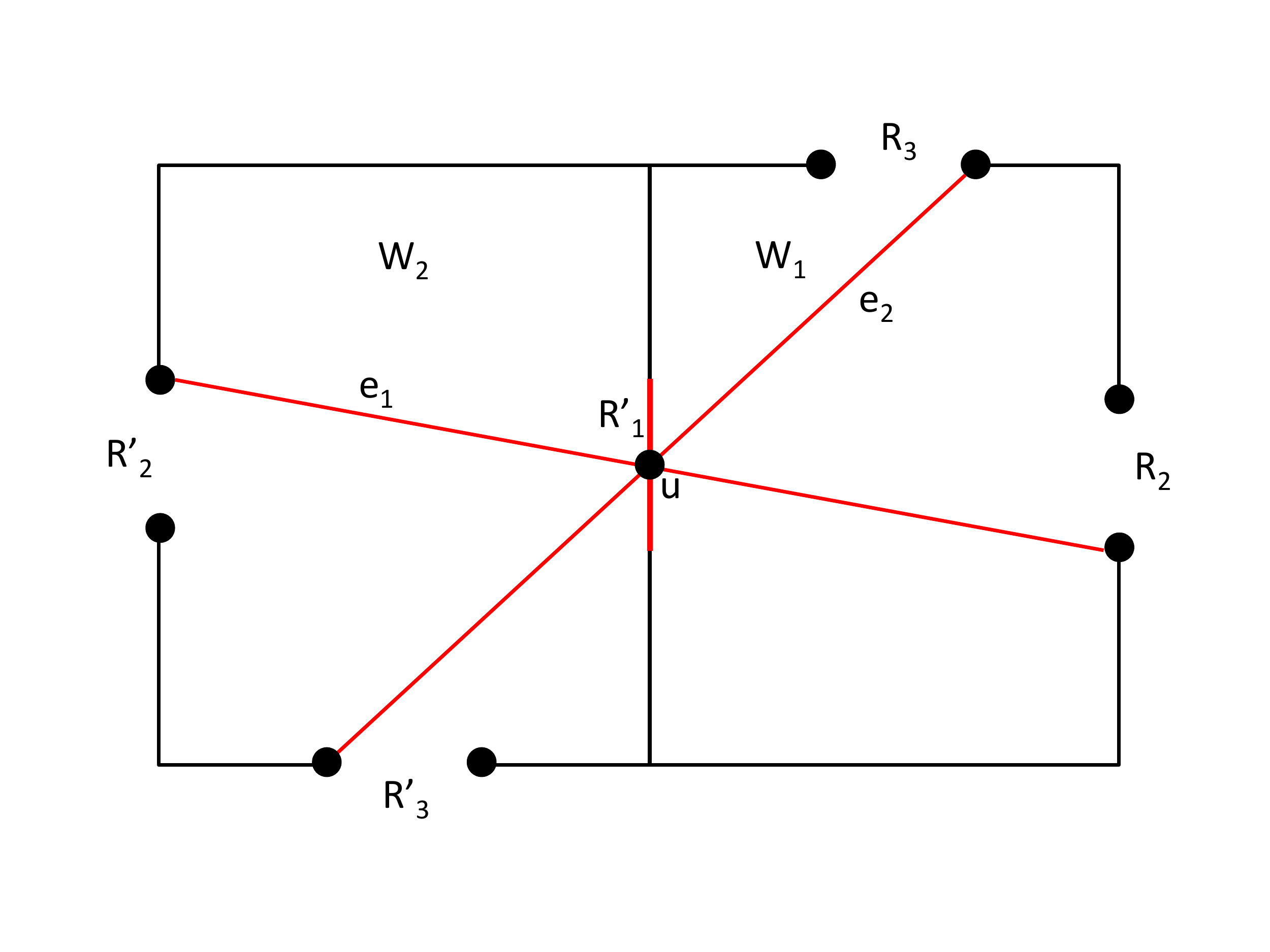}
\caption{Crossing curves}
  \label{fig2}
\end{figure*}

\medskip

Let us remind (1) in the proof of Lemma \ref{faceone} which essentially follows from the fact that $F$ cannot be embedded in a surface of smaller Euler genus.
(3) implies that if $F'$ cannot be embedded in a surface of smaller Euler genus, then  Canonical Lemma would follow
because there is no branch $R$ (other than $R_1, R_2$) in $W_1 \cup W_2$ such that $R$ appears in both $W_1$ and $W_2$ as in (3), and there
is a non-contractible curve of order exactly two in the embedding of $F'$, which hits exactly one vertex in $R_1$ and exactly one vertex in $R$.

However, this is not true; namely $F'$ could be a graph that can be embedded in a surface of smaller Euler genus. For example,
we can take $K_4$ in a projective plane. It is known (see Section 5.6 in \cite{BM}) that $K_4$ is the only minimal face-width two embedding in a projective plane (with all faces size four), yet it can be surely embedded in a plane. So we cannot follow the proof of Lemma \ref{canon1}.

Instead, we first take $F$ from Theorem \ref{obst}, which is a subgraph of $G$ that cannot be embed in a surface of smaller Euler genus.
In addition, as the arguments before Lemma \ref{canon1}, we can find a
family of all embeddings of $F$ ${\bf F''}= \{\hat F_1,\dots,\hat F_l\}$ such that each of them can be extended to an embedding of $G$
in $O(n)$ time, and such that 1-4 above are satisfied.
Let $F''=F' \cup F$. Since $F$ cannot be embedded in a surface of smaller Euler genus, neither can $F''$. Moreover, applying Lemma \ref{local} allows us to confirm that each $F''$-bridge is stable.
Since the embedding of $F'$ can be extended to an embedding of $G$, for each $i$ we can define an embedding $\hat F'_i$ of $F''$ in $S$
that is obtained from $\hat F_i$ and the embedding of $F'$. Note that the embedding $\hat F'_i$ is a face-width two embedding, because $F'$ is a subgraph of $F''$.

Let us observe that if there is a non-contractible curve $C''$ in $F'$ that is homotopic to $C$, then $C''$ must hit at least two vertices of $F''$.
For each $i$, if there is still a non-contractible curve $C''$ in $F''$ that is homotopic to $C$, that is contained in $W_1 \cup W_2$ of the embedding of $F'$,
and that hits exactly two vertices of $F''$, one in $R_1$ and the other in $R_2$, in the embedding $\hat F'_i$, then there must exist a subpath $R'_1 \subseteq R_1$ and a subpath $R'_2 \subseteq R_2$ such that both $R'_1$ and $R'_2$ are branches of $F''$, and $C''$ hits exactly one vertex in $R'_1$ and exactly one vertex in $R'_2$.

Since $F''$ cannot be embedded in a surface of smaller Euler genus, there is no closed curve $C_2$ as in (3) in the embedding $\hat F'_i$ of $F''$, which implies
that there is no branch $R_3$ (nor branch $R'_3$) as in (3). Moreover, no face of $W_i$ appears twice in the embedding $\hat F'_i$ of $F''$ 
since $\hat F'_i$
is of face-width two (for $i=1,2$). 
 
This implies that $F''$-bridges $\mathcal{B'}$ that have at least one attachment in $R'_1 \cup R'_2$ and at least one attachment in $R'_1 \cup R'_2$, are uniquely placed into the ``left'' side and the ``right side'' (or into the ``one'' part and the ``other'' part, if $C$ is non orientation-preserving)
 of $C$ in $W_1 \cup W_2$. It follows that $C$ uniquely splits the incidents edges of both $u$ and $v$ into the ``left'' side and the ``right side'' (if $C$ is non orientation-preserving, then the ``left'' side and the ``right side'' are replaced by the ``one'' part and the ``other'' part).  We can also place all $F''$-bridges $\mathcal{B'}$ into the ``left'' side and the ``right side'' (or into the ``one'' part and the ``other'' part, if $C$ is non orientation-preserving) of $C$ in $W_1 \cup W_2$, in $O(n)$ time, by simply looking at each bridge. Note that for this argument, we only need the assumption ``there is a non-contractible
 curve $C''$ that hits exactly two vertices of $F'$, one in $R_1$ and the other in $R_2$'', and no assumption on the vertices $u,v$ is needed. 
 It follows that given $W_1, W_2$, in $O(n)$ time, either we can place all $F''$-bridges $\mathcal{B'}$ into the ``left'' side and the ``right side'' (or into the ``one'' part and the ``other'' part, if $C$ is non orientation-preserving) of $C$ in $W_1 \cup W_2$, or we can conclude that such a non-contractible curve $C$ does not exist.

Since the embeddings of $F''$-bridges $\mathcal{B'}$ induce the embeddings of  all $F'$-bridges $\mathcal{B}$ that have at least one attachment in $R_1 \cup R_2$ and at least one attachment outside $R_1 \cup R_2$,  we can also place all $F'$-bridges $\mathcal{B}$ into the ``left'' side and the ``right side'' (or into the ``one'' part and the ``other'' part, if $C$ is non orientation-preserving) of $C$ in $W_1 \cup W_2$. in $O(n)$ time.
This proves Canonical Lemma.\qed

\section{Final Remarks and Correctness of our Algorithm}
\label{secmain}
\showlabel{secmain}

Finally we give our whole algorithm for Theorem \ref{thm:main1}. Basically our algorithm is described in
overview. Indeed, Steps 2, 3 and Step 6 are already described there. So we only give several remarks about
Steps 1, 4 and 5.

\paragraph{Step 1.}


Our first step is to reduce both graphs $G_1$ and $G_2$ to be 3-connected. This is quite standard in this literature, see \cite{logspace, jgaa}, but for the completeness, we give a sketch here.

 So let us consider the following decompositions of $G_1,G_2$, respectively.
\begin{enumerate}
\item
Decompose $G_1$ and $G_2$ into biconnected components by constructing
a biconnected component tree decomposition $(T, R)$.
\item
Decompose each biconnected component of $G_1$ and of $G_2$
into triconnected components
by constructing a triconnected component tree decomposition $(T',R')$.
\end{enumerate}

Again, by Theorems \ref{2conunique} and \ref{3conunique}, both
the biconnected component tree decomposition and
the triconnected component tree decomposition are unique.
For our convenience, let us assume that both $T$ and $T'$ are rooted.

For each biconnected component $R_t$ of the biconnected component tree decomposition $(T, R)$, we assign colors to the adhesion sets of $R_t \cap R_{t'}$ for each $tt' \in T$. So each adhesion set receives a color and any two different adhesion sets receive different colors. 
This allows us to define a ``rooted tree'', where for each edge $tt' \in T$ where $t$ is closer to the root, we can define the subtree $T''$ of $T$ that takes all nodes of $T$ that are in
the component of $T-tt'$ that does not contain the root. The vertex of $G$ corresponding to $R_t \cap R_{t'}$ will tell us which node of $T''$ is the root. Thus the colored vertex can be thought of a way to tell the parent node of any subtree of $T$.

It follows that given two graphs $G_1$ and $G_2$ and
their biconnected component tree decompositions
$(T_{G_1}, R_{G_2})$ and $(T_{G_1}, R_{G_2})$, $G_1$ and $G_2$ are isomorphic if and only if for each $t \in T_{G_1}$ and its corresponding
node $t' \in T_{G_2}$, the subtree $T''_{G_1}$ rooted at $t$ is
isomorphic to  the subtree $T''_{G_2}$ rooted at $t'$, and moreover,
the graph induced by $\bigcup_{t \in T''_{G_1}} R_t$
is isomorphic to the graph induced by
$\bigcup_{t' \in T''_{G_2}} R_t$ with respect to
the rooted vertex (see Theorem 5.8 in \cite{logspace}, or \cite{jgaa}).

Similarly, for each triconnected component $R'_t$ of the triconnected component tree decomposition $(T', R')$, we assign colors to the adhesion sets of $R'_t \cap R'_{t'}$ for $tt' \in T'$. So each adhesion set receives a color and any two different adhesion sets receive different colors.
 Note that $|R'_t \cap R'_{t'}| =2$ for each $tt' \in T'$. This, again, allows us to define a ``rooted tree'', where for each edge $tt' \in T'$ where $t$ is closer to the root, we can define the subtree $T''$ of $T'$ that takes all nodes of $T'$ that are in
the component of $T'-tt'$ that does not contain the root. The vertices of $G$ corresponding to $R'_t \cap R'_{t'}$ will tell us which node of $T''$ is the root. Thus the colored vertices can be thought of a way to tell the parent node of any subtree of $T'$.
Let us observe that $\{u, v\} = R'_t \cap R'_{t'}$
receive the same color,
but there is an ``orientation'' between $u$ and $v$. This way,
we make sure that how we glue $R'_{t}$ and $R'_{t'}$ together
at $u, v$.
(More precisely, we make sure that $u \in R'_{t'}$ should not map to $v \in R'_t$, and $v \in R'_{t'}$ should not map to $u \in R'_t$ either ).

It follows that given two biconnected graphs $G_1$ and $G_2$ and
their triconnected component tree decompositions
$(T'_{G_1}, R'_{G_2})$ and $(T'_{G_1}, R'_{G_2})$, $G_1$ and $G_2$ are isomorphic if and only if for each $t \in T'_{G_1}$ and its corresponding
node $t' \in T'_{G_2}$, the subtree $T''_{G_1}$ rooted at $t$ is
isomorphic to  the subtree $T''_{G_2}$ rooted at $t'$, and moreover,
the graph induced by $\bigcup_{t \in T''_{G_1}} R'_t$
is isomorphic to the graph induced by $\bigcup_{t' \in T''_{G_2}} R'_t$ with respect to the rooted
vertices (see Theorem 4.2 in \cite{logspace}, or \cite{jgaa}).

Let us observe that both in the biconnected component tree decomposition and in the triconnected component tree decomposition,
there are at most $g$ components that are
not planar. This follows from the following fact:
\begin{quote}
If $G =G_1 \cup G_2$ with $|G_1 \cap G_2|\leq 2$,
Euler genus of $G'$ is that of $G'_1$ plus that of $G'_2$, where
$G', G'_1, G'_2$ are obtained from $G, G_1, G_2$ by adding the edge
in $G_1 \cap G_2$, if it is not present (see \cite{BM}).
\end{quote}

Let us first start with biconnected component tree decompositions of $G_1, G_2$
respectively.
We first group all the ``planar'' biconnected components into
one component $G_{i,0}$. As
remarked above, there are at most $g$ non-planar biconnected components.
Therefore, we can enumerate all biconnected components $G_{i,j}$  that are not planar for $j=1,\dots,l \leq g$ and
for $i=1,2$.

For each non-planar biconnected component $G_{i,j}$, we obtain
a triconnected component tree decomposition.
We then group all the ``planar'' triconnected components into
one set $G'_{i,0}$. As
remarked above, there are at most $g$ non-planar triconnected components.
Therefore, we can enumerate all triconnected components $G'_{i,j}$  that are not planar for $j=1,\dots,l' \leq g$ and
for $i=1,2$.

We now check, for each $j$, whether or not $G'_{1,j}$ and $G'_{2,j}$ are  isomorphic
for $j=0,1,\dots,l'$ (with respect to the colored vertices). For
each component of $G'_{i,0}$, this can be done by Hopcroft and Wong \cite{HW}.
Note that since the triconnected tree decompositions are the same, we know which components we have to compare.
Assume for the moment that we can check isomorphism
for $G'_{1,j}$ and $G'_{2,j}$ for $j=1,\dots,l'$.

We now glue these graphs
$G'_{i,j}$ together at the colored vertices (with orientation), to obtain the original biconnected graph $G_{i,k}$, and check
whether or not $G_{1,k}$ and $G_{2,k}$ are isomorphic.
 This can be clearly done in $O(n)$ time, because the abstract trees of the triconnected component tree-decompositions
 are the same, so we just need to look at the colored vertices.

 Similarly, we now glue these graphs
$G_{i,k}$ together at the colored vertices, to obtain the original graph $G_{i}$, and check
whether or not $G_{1}$ and $G_{2}$ are isomorphic.
 This can be done in $O(n)$ time, because, again,  the abstract trees of the biconnected component tree-decompositions
 are the same, so we just need to look at the colored vertices.

Therefore, it remains to consider each ``non-planar''
3-connected component of $G_1$ and of $G_2$, respectively.
Hereafter,
we may assume that both graphs $G_1$ and $G_2$ are $3$-connected non-planar.

\paragraph{Step 4.}

In Step 4, let $G,S,{\bf F'}$ be as in Theorem \ref{algotwo}. Let us fix one graph $F'_i \in {\bf F'}$ and its face-width two embedding.
By Theorem \ref{algotwo}, for any homotopy class $H$, there is an $O(n)$ time algorithm to find
the cylinder, as in Theorem \ref{algotwo}.
We claim that if we only care about non-contractible curves that hit exactly two vertices in an embedding of $G$ that extends the embedding of $F'_i$, we only have to consider
at most $r(g)$ non-contractible curves (for some function $r$ of $g$), each in different nontrivial homotopy classes. To see this, we first need the remark right after Lemma \ref{homology}, i.e, $\bsize(F'_i) \leq l'(g)$ (for some function $l'$ of $g$.). Then by Lemma \ref{homology}, we just need to
consider $r(g)$ nontrivial homotopy classes, as claimed.

Since there are at most $l \leq N(g)$ subgraphs of $G$ in ${\bf F'}$ and since we only have to consider at most $r(g)$ different nontrivial homotopy classes for  non-contractible curves that hit exactly two vertices in the embedding of a graph in ${\bf F'}$ (for some function $r$ of $g$), thus in $O(n)$ time, we can enumerate the following pairs of subgraphs:

There is a $q'(g)$ for some function $q'$ of $g$ such that
\begin{enumerate}
\item
there are $q' \leq q'(g)$ pairs $(G'_1,L'_1),\dots,(G'_{q'},L'_{q'})$,
\item
for all $i$, $G=G'_i \cup L'_i$ and $|G'_1 \cap L'_i| = 4$,
\item
for all $i$, $G'_i$ can be embedded in a surface of Euler genus at most $g-1$,
\item
pairs $(G'_i, L'_i)$ are canonical in a sense that graph isomorphism would preserve these pairs,
\item
for all $i$, $L'_i$ is a cylinder with the outer face $F_1$ and
the inner face $F_2$ with the following property: there
 is a non-contractible curve $C_j$ that hits exactly two vertices $x_j,y_j$ in some embedding of $G$ of face-width two for $j=1,2$, and $x_1, y_1$ are contained in $F_1$ and $x_2, y_2$ are contained in $F_2$, where
$L'_i,C_1,C_2$ are obtained from Theorem \ref{algotwo} (we do not distinguish between the orientation-preserving case and the non orientation-preserving case.
For the second case, we refer the reader to 5 in ''Overview of our algorithm''.).
\item
an embedding of $G$ of face-width two in $S$ can be obtained from
some embedding of $G'_i$ in  a surface of Euler genus at most $g-1$ and
an embedding of the cylinder $L'_i$ by identifying respective copies of vertices $x_1, x_2, y_1$ and $y_2$ (thus $G'_i$ also contains all the vertices of $x_1,x_2, y_1, y_2$ and they are on the border of $G'_i$ and $L'_i$, respectively),
\item
for any non-contractible curve that hits exactly two vertices $x, y$ in some embedding of $G$ of face-width two, both $x$ and $y$ are contained in $L'_i$ for some $i$, and
\item
for $i_1\not=i_2$, either $L'_{i_1}$ and $L'_{i_2}$ come from different graphs in ${\bf F'}$, or $L'_{i_1}$ and $L'_{i_2}$ come from different nontrivial homotopy classes for the embedding of a single graph in ${\bf F'}$.

Note that the cylinder includes the case mentioned in Remark 3 right after Theorem \ref{algotwo} if the homotopy class is not orientation-preserving.
\end{enumerate}
Thus after Step 4, we apply our whole algorithm recursively to each of $G'_i, L'_i$ in the pair $(G'_i,L'_i)$ with colored vertices $x_1, y_1, x_2, y_2$
both in $G'_i$ and in $L'_i$. 
Note that we just need to apply Step 6 to $L'_i$.

Note also that we may obtain the unique decomposition by
   vertex-two cuts as in Theorem \ref{algotorus} such that each piece is
   planar. In this case, we directly go to Step 6.

\paragraph{Step 5.}

In Step 5, we create a set ${\bf G}$ of subgraphs of $G$ which is obtained from $G$ by splitting each vertex $v$ of $V_1$ into the ``right'' side and the ''left'' side, where $V_1$ comes from Lemma \ref{faceone}. Let us observe that at Step 5, as in Lemma \ref{faceone}, the ``left'' side and
the ``right'' side can be uniquely  determined and hence we are canonical.
Since $|V_1| \leq q(g)$, thus $|{\bf G}| \leq q(g)$. We apply our whole algorithm recursively to each graph in ${\bf G}$ and the colored ``splitted'' vertex $v$.

\paragraph{Time Complexity.}

Let us observe that in Steps 4 and 5, we have created at most $q'(g),q(g)$ subgraphs of $G$, respectively, and we recursively
apply our whole algorithm to each of these different subgraphs. On the other hand, each of these subgraphs can be embedded in a surface of Euler genus at most $g-1$ and hence, we recurse at most $g$ times. Thus in our recursion process, we create at most $w(g)$ subgraphs in total (for some function $w$ of $g$). Since all Steps 1-6 can be done in $O(n)$ time, and since we only deal with constantly many subgraphs in our recursion process, so the time complexity of our algorithm for Theorem \ref{thm:main1} is $O(n)$, as claimed.

\paragraph{Correctness.}

It remains to show the correctness of our algorithm. Note that by Step 1, we may assume that
the current graph is 3-connected.

Suppose we want to test the graph
isomorphism of two 3-connected graphs $G_1,G_2$, both admit an embedding in a surface $S$ of the Euler genus $g$.
We assume that this embedding is a minimum Euler genus embedding, i.e,
neither $G_1$ nor $G_2$ can be embedded in a surface of smaller Euler genus $g'$.

Suppose one (or both) of $G_1$ and $G_2$ has a polyhedral embedding in $S$.
As in Step 3, we apply Theorem \ref{thm:main2} to both $G_1$ and $G_2$, and hence we obtain
all polyhedral embeddings of both $G_1$ and $G_2$. Then for each of these polyhedral embeddings of $G_1$, and for each of these polyhedral  embeddings of $G_2$, we just apply Theorem \ref{thm:main3} to figure out whether or not these two embeddings are isomorphism. If there are
embeddings of $G_1$ and of $G_2$ that are isomorphic, then we know that $G_1$ and $G_2$ are isomorphic. Otherwise, they are not.

 Suppose none of $G_1$ and $G_2$ has a polyhedral embedding in $S$. Suppose first that one (or both) of $G_1$ and $G_2$ has a face-width
 two embedding in $S$. As above, in Step 4, we create $q'$ different pairs of subgraphs ${\bf G_1}=(G'_{1,1},L'_{1,1}),\dots,(G'_{q',1},L'_{q',1})$ of $G_1$, and $q''$ different pairs of subgraphs ${\bf G_2}=(G'_{1,2},L'_{1,2}),\dots,(G'_{q'',2},L'_{q'',2})$ of $G_2$.

  We claim that if $G_1$ and $G_2$ are isomorphic, we can create
 a pair $(G'_{i,1},L'_{i,1})$ of $G_1$ as above and a pair $(G'_{j,2},L'_{j,2})$ of $G_2$ as above,
 such that $G'_{i,1}$ and $G'_{j,2}$ are isomorphic and
 $L'_{i,1}$ and $L'_{j,2}$ are isomorphic.
    Let us fix the same embedding of $G_1, G_2$. Then as in Section \ref{sectwo},
    there is a  graph $F'_{i,1} \in {\bf F'}$ with its embedding that can be extended to
    the embedding of $G_1$, and there is also a graph $F'_{i',2} \in {\bf F'}$ with its embedding that can be extended to the embedding of $G_2$. 
 Note that $F'_{i,1}$ may not be the same graph as $F'_{i',2}$.
 However, we know that any non-contractible curve $C$ that hits exactly two vertices in $G_1$ (and in $G_2$)  must hit two vertices of $F'_{i,1}$ (and $F'_{i',2}$). So if $g\not=2$
 or $g=2$ but $H$ (a nontrivial homotopy class $H$) is as in Theorem \ref{algotwo}, then by Theorem \ref{algotwo}, we can find $L'_1$ in $G_1$ and $L'_2$ in $G_2$, respectively,
   such that every non-contractible curve in $H$ that hits   exactly two vertices in $G_1$ and in $G_2$, respectively, is contained in $L'_1$ and $L'_2$,
   respectively.
  Moreover, both $L'_1$ and $L'_2$ are bounded by such curves, and
  both $L'_1$ and $L'_2$ are canonical.
  Therefore $L'_1$ is isomorphic to $L'_2$, and hence
  $G_1-L'_1$ is isomorphic to $G_2-L'_2$, as claimed.

 In summary, we have the following:
 \begin{quote}
 If $G_1$ and $G_2$ are isomorphic, then $q'=q''$ and there is one pair of graphs $(G'_{i,1},L'_{i,1})$ in ${\bf G_1}$ and
 the other pair of graphs $(G'_{j,2},L'_{j,2})$ in ${\bf G_2}$ such that
 $G'_{i,1}$ and $G'_{j,2}$ are isomorphic and $L'_{i,1}$ and $L'_{j,2}$ are isomorphic.
 If $G_1$ and $G_2$ are not isomorphic, there are no such pair of graphs.
 \end{quote}

 If $g=2$ and $H$ is as in Theorem \ref{algotorus}, we obtain the unique decomposition by
   vertex-two cuts as in Theorem \ref{algotorus} such that each piece is planar. Then we go to Step
   6.

 Suppose finally none of $G_1$ and $G_2$ has a face-width
 two embedding in $S$. By Lemma \ref{faceone}, we obtain the vertex set $V'_1$ in $G_1$ and the vertex set $V'_2$ in $G_2$, where
 $V'_i$ corresponds to $V_1$ in Lemma \ref{faceone} for $i=1,2$.
 As above, we create subgraphs ${\bf G_i}$ of $G_i$ which is obtained from $G_i$ by splitting each vertex of $V_i$ into the ``right'' side and the ''left'' side.
 Moreover, we are canonical.

 By Lemma \ref{faceone}, it is straightforward to see the following;
 \begin{quote}
 If $G_1$ and $G_2$ are isomorphic, then $|V'_1|=|V''_1|$, and there is one graph in ${\bf G_1}$ and the other graph in ${\bf G_2}$
 that are isomorphic. If  $G_1$ and $G_2$ are not isomorphic, there are no such two graphs.
 \end{quote}

Finally, when the current graph comes to Step 6, it comes from
either Step 3 or Step 4 (with the unique decomposition by
   vertex-two cuts as in Theorem \ref{algotorus} such that each piece is
   planar). In the second case, we just need to consider each piece
   of a 3-connected planar graph.
   Thus at the moment, we have either a
planar embedding of a 3-connected graph or a polyhedral embedding of
a 3-connected graph in a surface of Euler genus $g >0$.

Since we already have all the polyhedral embeddings of $G'_1$ (which is a subgraph of $G_1$) and of $G'_2$ (which is a subgraph of $G_2$),
if two graphs $G'_1$ and $G'_2$ are
isomorphic at Step 6, there must exist an embedding of $G'_1$ and an
embedding of $G'_2$ that represent isomorphic maps (with respect to some colored vertices). By Theorem
\ref{thm:main3}, we can check map isomorphism of the embedding of
$G'_1$ and the embedding of $G'_2$ in $O(n)$ time, and hence we
can check whether or not $G'_1$ and $G'_2$ are isomorphic (with respect to some colored vertices) in $O(n)$
time.

Therefore in Step 6, we can figure out all pairs of subgraphs $(H_1,H'_1),\dots$ with $H_i \subseteq G_1$ and $H'_i \subseteq G_2$, where
both $H_i$ and $H'_i$ are graphs at Step 6, such that
$H_i$ and $H'_i$ are isomorphic (with respect to some colored vertices) for all $i$. This can be done in $O(n)$ time by Theorem \ref{thm:main3}, since we create at most $w(g)$ subgraphs of $G_i$ for some function $w$ of $g$ in our recursion process ($i=1,2$).

For each subgraph of $G_i$ ($i=1,2$) in Step 6,
we can easily go back to the reverse order of Steps 4 and 5 to come up with the original graphs $G_1$ and $G_2$ in $O(n)$ time, because
in both Steps 4 and 5, we only ``split'' a few vertices, and these
vertices are all colored so that we can identify two graphs.

Since we are canonical at Steps 4 and 5, thus having known all pairs of graphs $(H_1,H'_1),\dots$ with $H_i \subseteq G_1$ and $H'_i \subseteq G_2$ such that
$H_i$ and $H'_i$ are isomorphic for all $i$ (with respect to all colored vertices), we can see if $G_1$ and $G_2$ are isomorphic in $O(n)$ time. \qed

\section{Appendix}
\label{appendix1}

In this section, we give a proof of Theorem \ref{find1} (which was actually given in \cite{kmstoc08}, but we give the proof again).

We need some definitions.

A \DEF{tree decomposition} of a graph $G$ is a pair $(T,R)$, where
$T$ is a tree and $R$ is a family $\{R_t \mid t \in V(T)\}$ of
vertex sets $R_t\subseteq V(G)$, such that the following two
properties hold:

\begin{enumerate}
\item[(W1)] $\bigcup_{t \in V(T)} R_t = V(G)$, and every edge of $G$ has
both ends in some $R_t$.
\item[(W2)] If  $t,t',t''\in V(T)$ and $t'$ lies on the path in $T$
between $t$ and $t''$, then $R_t \cap R_{t''} \subseteq R_{t'}$.
\end{enumerate}

The \emph{width} of a tree decomposition $(T,R)$ is $\max\{|R_t|\mid
t\in V(T)\}-1$, and the \DEF{tree width} of $G$ is defined as
the minimum width taken over all tree decompositions of $G$.
The \emph{adhesion} of our decomposition $(T, R)$ for $tt' \in T$ is $R_t\cap R_{t'}$.

One of the most important results about graphs whose tree-width is
large is the existence of a large grid minor or, equivalently, a large
wall. Let us recall that an \DEF{$r$-wall} is a graph which is
isomorphic to a subdivision of the graph $W_r$ with vertex set
$V(W_r) = \{ (i,j) \mid 1\le i \le r,\ 1\le j \le r \}$ in which two
vertices $(i,j)$ and $(i',j')$ are adjacent if and only if one of
the following possibilities holds:
\begin{itemize}
\item[(1)] $i' = i$ and $j' \in \{j-1,j+1\}$.
\item[(2)] $j' = j$ and $i' = i + (-1)^{i+j}$.
\end{itemize}

We can also define an $(a \times b)$-wall in a natural way, so that
an $r$-wall is the same as an $(r\times r)$-wall. It is easy to
see that if $G$ has an $(a \times b)$-wall, then it has an
$(\lfloor\frac{1}{2}a\rfloor \times b)$-grid minor, and conversely,
if $G$ has an $(a \times b)$-grid minor, then it has an $(a \times
b)$-wall. Let us recall that the $(a \times b)$-grid is the
Cartesian product of paths $P_a\times P_b$.

%

The main result in \cite{RS5} says the following (see also
\cite{rein,yusuke,reed1,RST2}).

\begin{theorem}\label{gridgeneral}
For every positive integer $r$, there exists a constant $f(r)$ such
that if\/ a graph $G$ is of tree-width at least $f(r)$, then $G$
contains an $r$-wall.
\end{theorem}

Very recently, Chekuri and Chuzhoy \cite{ChekuriChuzhoy} gives a polynomial upper bound
for $f(r)$.
The best known lower
bound on $f(r)$ is of order $\Theta (r^2 \log r)$, see~\cite{RST2}.

Let $H$ be an $r$-wall in $G$. If $G$ is embedded in a surface $S$,
then we say that the wall $H$ is \DEF{flat} if the outer cycle of $H$
bounds a disk in $S$ and $H$ is contained in this disk.
The following theorem follows from Demaine et al. (Theorem 4.3) \cite{demaine1}, together with Thomassen \cite{carsten}
(see Proposition 7.3.1 in \cite{MT}).

\begin{theorem}
\label{grid1} \showlabel{grid1}
Suppose $G$ is embedded in a surface with Euler genus $g$. For
any $l$, if
$G$ is of tree-width at least $400lg^{3/2}$, then it contains a flat $l$-wall.
If there is no flat
$l$-wall in $G$, then tree-width of $G$ is less than
$400lg^{3/2}$.
\end{theorem}

Let $G$ be a graph that can be embedded in a surface $S$ of Euler genus $g$ and of face-width $k$.

The proof of Theorem \ref{find1} consists of the following two steps.
\begin{enumerate}
\item
If $G$ is of tree-width $w$ for fixed $w$ (which only depends on $g,k$), then we apply the dynamic programming technique
of Arnborg and Proskurowski \cite{Arn} to obtain in $O(n)$ time the graph $H$ and its embedding
as a surface minor of some embedding of $G$ in $S$.
\item
On the other hand, if $G$ is of tree-width at least $w$, then we keep deleting ``irrelevant'' vertices in $G$ to obtain
a graph $G'$ of tree-width at most $w$ (and moreover, $G'$ does not have such an irrelevant vertex). .
\end{enumerate}

For the second, we need the following result. For the proof,
see \cite{KMR,MT}.
Define a vertex $v$ of $G$ to be an \DEF{irrelevant vertex} if
$G$ has a surface minor of a minimal embedding of
face-width $k$ if and only if $G-v$ has.  Given a planar graph $H$, \textit{face-distance in $H$} of two vertices $x,y\in H$ is the minimal value of $|H\cap C|$ taken over all curves $C$ in $H$ that link $x$ to~$y$ and that meet $H$ only in vertices in this embedding of $H$.

\begin{theorem}
\label{delete}
\showlabel{delete}
Suppose that\/ $G$ contains a planar subgraph $Q$ and that $C$ is the
outer cycle of a planar embedding of $Q$. Suppose also that
for every vertex in $Q -C$, all its neighbors in the graph $G$
are contained in $Q$.
Then every vertex $v$ of\/ $Q$, which is of
face-distance in $Q$ at least $k$ from all the vertices of the
outer cycle $C$, is irrelevant.
\end{theorem}

By Theorem \ref{grid1}, if $G$ does not contain a vertex $v$ as in Theorem \ref{delete}, then tree-width of $G$ is less than $400kg^{3/2}$.
Thus by setting $w=400kg^{3/2}$,
it remains to show the above two points. The first point will be discussed in Subsection \ref{bdtw}, while the second
point will be discussed in Subsection \ref{bounding}.

\subsection{Bounded tree-width case}
\label{bdtw}
\showlabel{bdtw}

Our algorithm needs to test whether or not a given graph $G$ is of
bounded tree-width. This can be done in linear time by the algorithm
of Bodlaender \cite{bod}.

\begin{theorem}
\label{consttr}
\showlabel{consttr}
For every fixed $l$, there is a linear time algorithm to determine whether
or not a given graph $G$ is of tree-width at most $l$. Moreover, if
this is the case, then the algorithm gives a tree-decomposition
of tree-width at most $l$.
\end{theorem}

We need to use some tools from the graph minor theory in \cite{RS13}.

A {\em rooted graph} is an undirected graph $G$ with a set
$R(G)\subseteq V(G)$ of vertices specified as roots and an injective
mapping $\rho_G:R(G)\to \mathbb{N}$ assigning a distinct positive
integer label to each root vertex.
Isomorphisms of rooted graphs are
defined in the obvious way, i.e., roots must be mapped to roots with
the same label.

We say that a rooted graph $H$ is a {\em minor} of a rooted graph
$G$ if there is a mapping $\phi$ (a {\em model} of $H$ in $G$) that
assigns to each vertex $v\in V(H)$ a connected subgraph
$\phi(v)\subseteq G$ and to each edge $e\in E(H)$ an edge $\phi(e)$
in $G$ such that the following holds:
\begin{enumerate}
\item The subgraphs $\phi(v)$ ($v\in V(H)$) are pairwise vertex-disjoint connected subgraphs of $G$.
\item
Each edge $\phi(e)$ ($e\in E(H)$) is disjoint from all other edges $\phi(e')$ ($e'\in E(H)$)
and intersects $\cup_{v\in V(H)} \phi(v)$ only at its endvertices.
\item
If $u,v\in V(H)$ are the endpoints of $e\in E(H)$, then $\phi(e)$ is
incident in $G$ with a vertex in $\phi(u)$ and with a vertex in
$\phi(v)$.
\item For every $v\in R(H)$, $\phi(v)$ contains the vertex $u \in R(G)$ such that $\rho_G(u)=\rho_H(v)$.
\end{enumerate}

The {\em folio} of a (rooted) graph $G$ is the set of all rooted minors of $G$.
Clearly, the folio is closed under isomorphism, i.e., if rooted graphs $H$ and
$H'$ are isomorphic and $H$ is in the folio of $G$, then $H'$ is in
the folio as well. Note that there are $2^{\binom{|R(G)|}{2}}$
possible undirected graphs on $R(G)$. If $\delta$ is an integer, then the $\delta$-folio of $G$ contains every model $H$ of $G$ with $|V(H)| \leq \delta$. Obviously, every graph in the
$\delta$-folio has at most $\delta$ vertices.

The folio of a graph $G$ {\em relative to} a set $Z\subseteq
V(G)$ is the $2^{|Z|}$-folio of the rooted graph $G'$, where $G'$ is isomorphic
as unrooted graphs,
but $R(G')=Z$.

By a {\em surface folio} $\mathcal{F}$ of a rooted graph $G$ that can be embedded in a surface $S$ of Euler genus $g$, we mean that
each model $Z_i$ in $\mathcal{F}$ is in a folio of $G$ and moreover, one embedding $II_i$ such that each face is homeomorphic to a disk
and each $Z_i$-bridge can be embedded in a face of $II_i$, is associated with $Z_i$. Moreover,
the embedding $II_i$ of $Z_i$ can be extended to an embedding of $G$ in $S$.
Note that there may be two models $Z_i,Z_j$ in $\mathcal{F}$ that are isomorphic, but their embeddings in $S$ are different.

If $\delta$ is an integer, then the $\delta$-surface-folio of $G$ can be defined in the same way as $\delta$-folio.

It is known that the folio relative to bounded number of vertices can be determined in polynomial time if the
tree-width is bounded.

\begin{theorem}[See~\cite{Arn,RS13}]
\label{thm:twvertex}\showlabel{thm:twvertex}
 For integers $w$ and $l$, there exists a
$(w+l)^{{O}(w+l)} O(n)$ time algorithm for computing the folio
relative to a set of $l$ vertices in graphs of tree-width $w$.
In particular, if $w$ and $l$ are fixed, there exists a linear-time
algorithm.
\end{theorem}

We prove the following analogue of Theorem \ref{thm:twvertex} for the surface-folio.
\begin{theorem}
\label{thm:twvertexs}\showlabel{thm:twvertexs}
 For integers $w$ and $l$, there exists a
$w^{{O}(w)} O(n)$ time algorithm for computing the surface-folio
relative to a set of $l \leq w$ vertices $Z$ in graphs of tree-width $w$.
In particular, if $w$ and $l$ are fixed, there exists a linear-time
algorithm.
\end{theorem}
\begin{proof}
Our algorithm follows the standard
dynamic programming approach of Arnborg and Proskurowski \cite{Arn}. So let us give just a sketch.
As in \cite{Arn}, we may assume that each degree in $T$ is at most three.

Given a tree-decomposition $(T,R)$,
the dynamic programming approach of Arnborg and Proskurowski \cite{Arn} assumes that $T$
is a rooted tree whose edges are directed away from the root. We fix the root node $t$ and assume that
$Z$ is in $R_t$.

For $t_1t'_1\in E(T)$ (where $t_1$ is closer to the root than $t'_1$), define
$S(t_1,t'_1)=R_{t_1}\cap R_{t'_1}$ and $G(t_1,t'_1)$ to be the induced subgraph of $G$
on vertices $\bigcup R_s$, where the union runs over all nodes of $T$ that are
in the component of $T-t_{1}t'_1$ that does not contain the root. The algorithm of
Arnborg and Proskurowski starts at all the leaves of $T$ and then
we have to compute the following:
\begin{quote}
For every $t_{1}t'_1\in E(T)$ (where $t_1$ is closer to the root
than $t'_1$), we compute the $w^w$-surface-folio relative to $S(t_1,t'_1)$ in
$G(t_1,t'_1)$.
\end{quote}

Note that since $|S(t_1,t'_1)| \leq w$, the size of surface-folio relative to $S(t_1,t'_1)$ is at most $w^w$.

If $t'_1$ is a leaf, we can compute the $w^w$-surface-folio relative to $S(t_1,t'_1)$ in
$G(t_1,t'_1)$ by a brute force in $O(w^w)$ time.

We assume we have this information for each child $t'_2$ and $t''_2$ of
  $t'_1$.
  A simple brute force solution goes as follows.  We try combining the
 $w^w$-surface-folio relative to $S(t'_1,t'_2)$ in $G(t'_1,t'_2)$ and the $w^w$-surface-folio relative to $S(t'_1,t''_2)$ in $G(t'_1,t''_2)$,
 together with each model in the $w^w$-surface-folio relative to $S(t_1,t'_1)$ in $R_{t'_1}$.
    We can
  easily check if these three objects are consistent to
  represent a model in the $w^w$-surface-folio relative to $S(t_1,t'_1)$ in $G(t_1,t'_1)$.
  For each model in the $w^w$-surface-folio relative to $S(t_1,t'_1)$ in $R_{t_1}$,
  we keep such a solution.  The number of
  combinations to consider is $w^w \times w^w \times w^w$, and each
  can be checked in $O(w^{w})$ time, so the total time needed to
  compute the information for $t_1$ is $O(w^{w})$.

  When we come to the root $t$, we can compute the $w^w$-surface-folio relative to $Z$. Since each iteration can be done in $O(w^{w})$ time,
  thus we can compute the $w^w$-surface-folio relative to $Z$ in $O(n)$ time.
\end{proof}

\subsection{Bounding tree-width}
\label{bounding}
\showlabel{bounding}

We first give the following result shown in \cite{KMR}.

\begin{theorem}
\label{algmatch2} \showlabel{algmatch2}
Let $G$ be a graph with minimum degree at least 2 with at most\/
$4n$ edges. Let $d > 8g\cdot 2^{16\sqrt{g}}$ and $\epsilon = d^{-6}$.
Then we can find in linear time one of the following:
  \begin{enumerate}\myitemsep
  \item
  A vertex set $Z$ of at least $5\epsilon n$ vertices of degree 2,
  each of which has the same pair of neighbors as at least one
  other vertex in $Z$.
  \item
    An induced matching $M$ in $G$ containing at least $\epsilon n$ edges.
  \item
    A minor $G'$ of $G$ which is a forbidden minor for the
    surface $S$ of Euler genus $g$.
  \end{enumerate}
\end{theorem}

We are given a graph $G$ on a surface $S$ with Euler genus $g$. We want to
bound its tree-width by deleting many vertices at once, and our goal
is to do this in linear time. Moreover, we want that the deleted
vertex set $U$ is \DEF{irrelevant}.
Recall that a cycle $C$ in $G$ in
a surface $S$ is called \emph{flat} if $C$ bounds an open disc
$D(C)$ in $S$.
We say that a vertex $v \in G$ is \emph{$k$-nested}, if there are $k$ disjoint cycles $C_1,\ldots,C_k$ such
that $D(C_k) \supseteq \dots \supseteq D(C_1)$, and
$v$ is contained in the disk $D(C_1)$.
Therefore, following Theorem \ref{delete}, a vertex in $G$ is \emph{irrelevant} if $v$ is
$k$-nested in $G$. Let us restate our result here.

\begin{lemma}
\label{irrelevant11}
\showlabel{irrelevant11}
Given a graph $G$ that can be embedded in a surface $S$ with Euler genus $g$,
for fixed $g,k$, there is a linear time algorithm to find a
vertex set $X\subseteq V(G)$, such that
each vertex in $X$ is irrelevant.
Moreover tree-width of the resulting graph $G-X$ is less than $400kg^{3/2}$.
\end{lemma}
\begin{proof}
Here is a description.
Hereafter, we assume that $G$ has minimum degree at least 2.

\smallskip

{\bf Step 1.} Find a sequence of graphs $G=G_0,G_1,...,G_b$ such
that $G_i$ is obtained from $G_{i-1}$ by either contracting an induced
matching $M_i$ with at least $\epsilon|G_{i-1}|$ edges for
some small but constant $\epsilon>0$, or deleting a stable set of
$\epsilon |G_{i-1}|$ vertices, each of degree 2. In the second case, every deleted vertex has
the same neighbors as another vertex of degree 2 in the stable set. In addition, we add an edge between
two neighbors of each vertex $x$ in the stable set.

In both cases,
the resulting graph $G_{i}$ is a minor of $G_{i-1}$.

\smallskip

This step can be done as discussed in Theorem \ref{algmatch2}.
We may assume that the third output in Theorem \ref{algmatch2} would not happen.

We keep doing it $b$ steps, where $b$ is minimum integer
such that $G_b$ has fewer than $B$ vertices for some absolute constant $B$.
Then $b \le log_{1/\epsilon} n$ and the sum of the sizes of all $G_i$
is $O(n)$.

At each step $i$, we can either find a desired induced matching or a desired
stable set in time $O(|G_i|)$ as explained in Theorem \ref{algmatch2}. Note that
since $G$ can be embedded into the surface $\Sigma$ of Euler genus $g$,
we never get the third outcome of Theorem \ref{algmatch2}.

\smallskip

{\bf Step 2.} Apply a brute force algorithm to find irrelevant
vertices of $G_b$. Since $|G_b| < B$, this can
be done in constant time. Let $G'_b$ be the subgraph of $G_b$
obtained from $G_b$ by deleting irrelevant vertices. Since $G'_b$
has no vertex that is $k$-nested, so $G'_b$ has
tree-width less than $400kg^{3/2}$ by Theorem \ref{grid1}.

\smallskip

We recursively apply Step 3 for $i=b,b-1,\dots$.

\smallskip

Let $G_{i+1}$ be the graph obtained in the previous iteration.
$G'_{i+1}$ is a subgraph of $G_{i+1}$ with the following properties;

\begin{enumerate}
\item
$G'_{i+1}$ is embedded into a surface $\Sigma'$ of Euler genus $g$.
\item
$G'_{i+1}$ does not have a vertex that is irrelevant.
\item
Each vertex in $V(G_{i+1}) - V(G'_{i+1})$ is irrelevant.
\end{enumerate}

The purpose of Step 3 is to start with $G'_{i+1}$, and then to
construct a graph $G'_{i}$ satisfying the above properties for $i$
in $O(|G_i|)$ time. A short computation implies that if we can do it
in $O(|G_i|)$ time for each $i$, Step 2 can be done in $O(n)$ time.
Note that by the above properties, $G'_{i+1}$ is of tree-width at most
$400kg^{3/2}$ by Theorem \ref{grid1}.

\smallskip

{\bf Step 3.} We shall find a vertex set
$X$ that consists of irrelevant vertices in the graph $G''_i$ in
time $O(|G_i|)$, where $G''_i$ can be obtained from $G'_{i+1}$ by
uncontracting the induced matching, or adding a stable set of
$\epsilon |G_i|$ vertices each of degree 2, as in Theorem
\ref{algmatch2}.
Then
output the graph $G'_i=G''_i - X$.

\medskip

This step is crucial. It consists of several phases.
Let us first observe the following;

\begin{quote}
If a vertex $x$ is irrelevant for $i$,
then $x$ is irrelevant for $i'
< i$.
\end{quote}

In order to show this observation, we must prove that after deleting
irrelevant vertices in $G_i$, all the previously deleted vertices
are also irrelevant in $G_{i}$. We now argue that this
is, indeed, true. Suppose not.
In this case, we may assume that
 in the current graph $G_i$, each of all the previously
deleted vertices is $k$-nested, but when we delete an irrelevant vertex $v$ from $G_i$,
there is a vertex $w$ which was deleted previously, such that $w$ is
not $k$-nested in $G_i-v$. Let $C_1,\dots,C_{k}$ be the
$k$ nested cycles surrounding $w$ in $G_i$, and let
$C'_1,\dots,C'_{k}$ be the $k$ nested cycles surrounding $v$ in $G_i$.
Assume that $v$ is in one of $C_1,\dots,C_{k}$, say $C_l$. Let us assume $l \geq k/2$, as the other case is identical.

If we
cannot reroute $C_l$ using $C'_1$, this means that $C'_1$ hits both
$C_{l-1}$ and $C_{l+1}$. Inductively, it can be shown that if we
cannot reroute $C_{l-j+1},\dots,C_l,\dots,C_{l+j-1}$ using
$C'_1,\dots,C'_j$, this means that $C'_j$ hits both $C_{l-j}$ and
$C_{l+j}$. However, we can reroute
$C_{2l-1},\dots,C_l,\dots,C_{1}$ using $C'_1,\dots,C'_{l}$. So, there
are $k$ nested cycles in $G_i-v$ surrounding $w$, a contradiction. Thus
the observation holds.

\medskip

This observation implies that we only need to consider the graph $G''_i$
to construct the subgraph $G'_i$ of $G_i$.

\medskip

First, if there is a stable set of $\epsilon |G_i|$ vertices in
$G_i$, each of degree 2, and $G_{i+1}$ is obtained from $G_i$ by
deleting this stable set, then since every vertex in the stable set has the
same neighbors as at least one vertex in the stable set, and
moreover, the edge in its neighbors is added to $G_{i}$, it is easy
to see that the resulting graph $G''_i$ has no vertex that is
$k$-nested, and hence we are done, as we just
output $G''_i$.

So we may assume that $G_i$ has an induced matching $M_i$ of order
$\epsilon |G_i|$. Recall that $G''_i$ is the graph obtained from
$G'_{i+1}$ by uncontracting the matching $M_i$ restricted to the
graph $G'_{i+1}$. First, let us observe that tree-width of  $G''_i$
is at most twice of that of $G'_{i+1}$ (since the uncontraction
increases tree-width by factor $2$).
So it follows that $G''_{i}$ is of tree-width $w$
at most $800kg^{3/2}$. Thus
let us keep in mind that we are only working on the tree-width
bounded graph $G''_i$, and we just need to find a desired set $X$ as
in Lemma \ref{irrelevant11} in $G''_i$.

We now show how to obtain the graph $G'_i$ from $G''_i$.
Recall that $G''_i$ is embedded into a surface $S$ of
Euler genus $g$.

%

By Theorem \ref{consttr}, we can obtain a tree-decomposition $(T,R)$ of $G''_i$ of width $w$.
As in the proof of Theorem \ref{thm:twvertexs}, we may assume that each degree in $T$ is at most three.
We fix the root note $t$.
Thus $T$ is a rooted tree. As above, for $t_1t'_1\in E(T)$ (where $t_1$ is closer to the root than $t'_1$), define
$S(t_1,t'_1)=R_{t_1}\cap R_{t'_1}$ and $G(t_1,t'_1)$ to be the induced subgraph of $G$
on vertices $\bigcup R_s$, where the union runs over all nodes of $T$ that are
in the component of $T-t_{1}t'_1$ that does not contain the root.


The main idea in the rest of the proof is the following:

\begin{quote}
For each $R_t$ in the tree-decomposition $(T,R)$, if
we can compute the $w^w$-surface-folio relative to $R_t$ in $G$,
then we can find all the irrelevant vertices in $R_t$.
\end{quote}

Indeed, if an irrelevant vertex is contained in one graph $R_t$,
then we can detect it by finding the $w^w$-surface-folio relative to $R_t$.

So our algorithm will do the following two things simultaneously:
constructing the $w^w$-surface-folio relative to $R_t$ and deleting
an irrelevant vertex is contained in $R_t$

We are now ready to describe our algorithm here.
Because we need to compute the $w^w$-surface-folio relative to $R_{t'}$ in $G$ for each $t' \in T$, thus
we need to consider the two phases; working from the leaves, and
working from the root.

\smallskip

{\bf Phase 1.} Working from the leaves.

\smallskip

We first work from the leaves of the tree-decomposition.
For all the
leaves of $T$,
we can find all the irrelevant vertices in constant time, as
each leaf has at most $w \leq 800kg^{3/2}$ vertices.

%
%
%
%
%
%
%
Let us look at a node $t' \in T$.
Let $F_{t_i}$ be the $w^w$-surface-folio relative to $S(t',t_i)$ in $G(t',t_i)$ for $i=1,2$, where $t_1,t_2$ are the children of $t'$.
For each model $F \in F_{t_1}$ and for each model $F' \in F_{t_2}$, we compute the $w^w$-surface-folio relative to $S(t'',t')$
in $R_{t'} \cup F \cup F'$,
where $t''$ is the parent of $t'$ (if $t'$ is a leaf of $T$, then $F_{t_1}=F_{t_2}=\emptyset$).
This can be easily done in $O(|R_{t'}\cup F \cup F'|^{|R_{t'}\cup F \cup F'|})$ time by a simple brute force. So at this moment, we can compute
the $w^w$-surface-folio relative to $S(t'',t')$ in $G(t'',t')$.
Then we delete all the the
irrelevant vertices in $R_{t'}$ in $O(|R_{t'}\cup F \cup F'|^{|R_{t'}\cup F \cup F'|})$ time by again a simple brute force.
After deleting the irrelevant vertices in $R_{t'}$, we update the $w^w$-surface-folio relative to $S(t'',t')$
in $R_{t'} \cup F \cup F'$, in time $O(|R_{t'}\cup F \cup F'|^{|R_{t'}\cup F \cup F'|})$.

Since $|F_{t_1}|, |F_{t_2}|, |R_{t'}|$ are all bounded in terms of $k,g$, thus
in order to compute the $w^w$-surface-folio relative to $S(t'',t')$ and delete all the irrelevant vertices in $R_{t'}$, it only takes
$O(f_1(k,g)|R_{t'}|)$ time in total for some function $f_1$ of $k,g$. Then we look at the parent of $t''$,
and so on.

By doing this procedure,
we can reach the root node $t$ from
all the leaves. When we perform this algorithm at the root node $t$,
we can delete all the irrelevant vertices in $R_{t}$, because
we can compute
the $w^w$-surface-folio relative to $R_{t}$ in $G$.

In each iteration, the time complexity is $O(f_1(k,g)|R_{t'}|)$ for each $t' \in T$.
Thus in total, we can do Phase 1
in time $O(f_1(k,g)n)$, which is linear with respect to $n$.

This finishes the phase 1.
Note that at the moment, we can detect all the irrelevant vertices in the root $R_{t}$, but
we may not be able to detect all the irrelevant vertices in other nodes $R_{t'}$. This is
because we need the information about the $w^w$-surface-folio relative to $R_{t'}$ in $G$. So far, for each $t' \in T$, we only get
the information about the $w^w$-surface-folio relative to $R_{t'}$ in $G(t',t_1)$ and $G(t',t_2)$ where $t_1,t_2$ are the children of $t'$.

\smallskip

{\bf Phase 2.} Working from the root in the resulting graph.

\smallskip

After the first phase, we need to work from the root. Let $G'$ be
the resulting graph from the phase 1, and let $(T,R)$ be the
resulting tree-decomposition of $G'$. Note that this
tree-decomposition has still tree-width at most $w \leq 800kg^{3/2}$.

We now work from the root $t$ to the leaves of $(T,R)$.
For each $t'$, we need to compute the $w^w$-surface-folio relative to $S(t'',t')$ in
$(G-G(t'',t')) \cup S(t'',t')$, where
$t''$ is the parent of $t'$. As in Phase 1, we are done with the root $t$. Suppose $t'\not=t$.
Note also that the $w^w$-surface-folio $F_{t_1}$ relative to $S(t'',t_1)$ in $G(t'',t_1)$ is already computed by Step 1, where
$t_1$ is the child of $t''$ with $t_1\not=t'$. Suppose we know the $w^w$-surface-folio $F_{t''}$ relative to
$S(t''',t'')$ in $(G-G(t''',t'')) \cup S(t''',t'')$, where
$t'''$ is the parent of $t''$.

For each model $F \in F_{t_1}$ and for each model $F' \in F_{t''}$, we compute the $w^w$-surface-folio relative to $S(t'',t')$
in $R_{t''} \cup F \cup F'$.
This can be easily done in $O(|R_{t''}\cup F \cup F'|^{|R_{t''}\cup F \cup F'|})$ time by a simple brute force.
Note that we have already deleted the irrelevant vertices in $R_{t''}$ because $t''$ is the parent of $t'$.
We also note that at this moment, together with
two $w^w$-surface-folios relative to $R_{t'}$ in $G(t',t'_1)$ and in $G(t',t'_2)$ (computed in Phase 1), where $t'_1,t'_2$ are the children of $t'$, we can compute the $w^w$-surface-folio
relative $R_{t'}$.
Then we delete all the the
irrelevant vertices in $R_{t'}$ by using the $w^w$-surface-folio relative to $S(t'',t')$
in $R_{t''} \cup F \cup F'$, together with two $w^w$-surface-folios relative to $R_{t'}$ in $G(t',t'_1)$ and in $G(t',t'_2)$ 
 (computed in Phase 1),
in $O((|R_{t'}|+3w^w)^{|R_{t'}|+3w^w})$ time. Since $|F_{t_1}|, |F_{t''}|, |R_{t'}|$ are all bounded in terms of $k,g$, thus
in order to compute the $w^w$-surface-folio relative to $S(t'',t')$ in $(G-G(t'',t')) \cup R_{t'}$, and
delete all the irrelevant vertices in $R_{t'}$, it only takes
$O(f_2(k,g)|R_{t'}|)$ time in total for some function $f_2$ of $k,g$. 
After deleting the irrelevant vertices in $R_{t'}$, we update the $w^w$-surface-folio relative to $S(t'',t')$. 
This can be also done in $O(f_2(k,g)|R_{t'}|)$ time, by following the above arguments. 

Then we look at the children of $t'$,
and so on.
%
%
%
%

We keep applying this procedure until we reach all the leaves. Then for each $t' \in T$, we can find the $w^w$-surface-folio relative to $R_{t'}$ in $G$, and
detect all the irrelevant vertices in all the nodes $R_{t'}$.

In each iteration of Phase 2,
the time complexity is $O(f_2(k,g)|R_{t'}|)$  for each $t' \in T$.
Thus in total, we can do Phase 2
in time $O(n \times f_2(k,g))$, which is linear with respect to $n$.

This completes the description of the algorithm.\qed

\medskip

As observed above,
all the irrelevant vertices in $G''_i$ are deleted in the
above algorithm.

In summary, we can, in time $O(|G''_i|)$, find a vertex set $X$ in $G''_i$ such that each vertex in $X$ is
irrelevant in $G''_i$, and the graph $G'_i=G''_i-X$ has no irrelevant vertex.
So $G'_i$ is of
tree-width at most $400kg^{3/2}$ by Theorem \ref{grid1}, and hence $G'_i$ is as desired.

Thus Step 3 is done and this completes the proof of Lemma \ref{irrelevant11}.
\end{proof}
\end{document}